\documentclass{article}
\usepackage[utf8]{inputenc}
\usepackage[paper=a4paper,left=25mm,right=25mm,top=25mm,bottom=25mm]{geometry}

\usepackage{amsmath, amsfonts, amssymb, mathtools, amsthm}
\usepackage{enumitem}
\usepackage{bbm}
\usepackage{authblk}
\usepackage[authoryear]{natbib}

\newtheorem{lemma}{Lemma}

\usepackage[]{graphicx}

\title{Confirmatory Adaptive Hypothesis Tests in Markovian Illness-Death Models}
\author[1,*]{Rene Schmidt}
\author[1]{Moritz Fabian Danzer}

\date{\today}

\affil[1]{Institute of Biostatistics and Clinical Research, University of Münster, 48149 Münster, Germany}
\affil[*]{Corresponding author: rene.schmidt@ukmuenster.de}
\begin{document}

\maketitle

\begin{abstract}
Classic adaptive designs for time-to-event trials are based on the log-rank statistic and its increments. Thereby, only information from the time-to-event endpoint on which the selected log-rank statistic is based may be used for data-dependent design modifications in interim analyses. Further information (e.g. surrogate parameters) may not be used. As pointed out in a letter by P. Bauer and M. Posch in 2004, adaptive tests on overall survival (OS) based on the log-rank statistic do in general not control the significance level if interim information on progression-free survival (PFS) is used for sample size adjustments, because progression is associated with increased risk of death. 

In contrast, in adaptive designs for time-to-event trials, which are constructed according to the principle of patient-wise separation, all trial data observed in interim analyses may be used for design modifications without compromizing type one error rate control. But by design, this comes at the price of incomplete use of the primary endpoint data in the final test decision or worst-case considerations which lead to a loss of power. Thus, the patient-wise separation approach cannot be regarded as a general solution to the problem described by Bauer and Posch.

We address this problem within the framework of a comprehensive independent increments approach. We develop adaptive tests on OS in which sample size adjustments may be based on the observed interim data of both OS and PFS, while avoiding the problems of the patient-wise separation approach. We provide this methodology for both single-arm trials, in which a new therapy is compared with a pre-specified deterministic reference, and randomized trials, in which a new therapy is compared with a concurrent control group. The underlying assumption is that the joint distribution of OS and PFS is induced by a Markovian illness-death model. 
\end{abstract}

\noindent {\bf KEY WORDS:} Adaptive design; Illness-death model; Log-rank test; Survival data

\section{Introduction}\label{Sec01}

For ethical and economic reasons, interim analyses are carried out in clinical trials in order to be able to decide promptly on the effectiveness of a new treatment and, if necessary, to carry out design modifications (e.g. recalculation of the sample size) based on the data collected so far. Study designs that allow such flexibility while maintaining the probability of a type one error are referred to as confirmatory adaptive designs. The basic idea of adaptive designs was formulated by Peter Bauer \citep{B89,BauerPosch04}. For studies with short-term endpoints (i.e. the primary endpoints are more or less immediately observable), the work of \cite{PH95}, \cite{LW99}, \cite{SM01}, \cite{BPB02} and \cite{H01} provides a basic methodological toolkit of adaptive techniques that allows extensive flexibility (from sample size recalculation to data-dependent modification of (multiple) testing problems). A current overview, in particular of the classic adaptive methods and their application, can be found in \cite{BW16} and \cite{B16}. It is crucial for the applicability of these techniques that the stepwise p-values are independent or at least p-clud \cite[c.f.][]{BPB02}. This is usually fulfilled in studies with short-term endpoints, as each patient effectively only contributes data to one stage of the adaptive design. In chronic or potentially life-threatening diseases (e.g. in oncology), however, not only the short-term response but also the achievement of long-term remission is of central importance from the patient's point of view. In these studies, the primary endpoints are often event time endpoints, such as the time from the start of therapy to a predefined critical event. This justifies a need for adaptive methodology for event time endpoints.

However, subtle problems arise in the context of adaptive designs for time-to-event endpoints \citep[c.f.][]{BauerPosch04,B16}. This is because in a study with time-to-event endpoint, the outcome of the patients is usually recorded over long-term periods. This means that if the study is analyzed in several stages and a study patient is still event-free at the time of an interim analysis, they will continue to be at risk in subsequent stages until either an event occurs or the patient withdraws from the study without an event (censoring). Such a patient therefore contributes data to several stages of the adaptive design. Consequently, this patient and his data (unlike in studies with a short-term endpoint) cannot be clearly assigned to a fixed stage. This induces potential dependencies between the stages.

Historically, the first approaches to achieve independent stage-wise p-values in a time-to-event context and thus to construct adaptive designs were developed for trials with a single primary time-to-event endpoint. These procedures are based on the independent increment property of the underlying (sequence of) univariate test statistics \citep[c.f.][]{T81, T82, OS86, FL00,bib3}. For the two-sample log-rank test, adaptive designs were provided by \cite{SM01}, \cite{W06}, and \cite{DP07}; for the Cox proportional hazards model by \cite{J-EI09} and for situations with non-proportional hazards by \cite{BB17}. For single-arm phase IIa time-to-event trials, a corresponding adaptive one-sample log-rank test was developed by \cite{S18}. With these methods, however, the interim data that can be used for design modifications are very limited: On the one hand, these methods are tailored to study situations with a single primary event time endpoint. On the other hand, only information relating to this single primary time-to-event endpoint may be used in the context of data-dependent design modifications. Additional information from other time-to-event endpoints must not be used for design modifications, as this can lead to inflation of the type one error probability \cite[c.f.][]{BauerPosch04}. This need to restrict design modifications to a single time-to-event endpoint is unsatisfactory from a clinical perspective, because the dynamics and complexity of many diseases suggests a simultaneous consideration of several event time endpoints, especially for interim decisions.

Currently there are some proposals for adaptive designs that allow the use of short-term endpoints for design modifications in adaptive survival trials \cite[c.f.][]{S10,Fea11,F12}. However, these methods are based on specific assumptions about the joint distribution of time-to-event and short-term endpoint. Approaches have also been developed that allow the use of categorical surrogate endpoints for design modifications in adaptive event time studies in a non-parametric setting, e.g. tumor response status as a surrogate for overall survival \citep[c.f.][]{Bea09,BBB18}. However, these methods do not allow the simultaneous use of information from multiple time-to-event endpoints for design modifications. \cite{Rea16} have developed a two-stage adaptive design for studies with two event time endpoints, which strives for proof of significance for at least one of the two endpoints. The design allows sample size recalculation based on the event time endpoint that shows the smaller p-value in an interim analysis.

The first methods that allow comprehensive and simultaneous use of information from multiple time-to-event endpoints for design modifications were proposed by \cite{JSJ11}, \cite{IS12}, \cite{Mea14}, and \cite{CGB15}. By design, however, these methods have the disadvantage that a certain proportion of the primary outcome data generally cannot be taken into account in the final confirmatory test decision \citep[see][Table 2]{Mea16}. In addition, the possibility of an early rejection of the null hypothesis is limited, as the calculation of the p-value of an interim analysis by design only takes place at a predetermined (sufficiently late) time $t_{max}$. An extension of this approach proposed by \cite{Jea19} strengthens the possibility of an early rejection of the null hypothesis in the context of an interim analysis, whereby the problem of incomplete data use in the final test decision remains. A solution to the problem of incomplete data use in the final test decision was proposed by \cite{Mea16}. However, this approach is based on a worst-case consideration and results in a conservative test decision.

The subject of this paper is confirmatory adaptive tests for clinical trials with two correlated time-to-event endpoints that allow data-dependent design modifications (e.g. recalculation of the sample size) on the basis of both time-to-event endpoints while avoiding the problems described above (i.e. with full control and exhaustion of the significance level and full use of all available event time data in the final test decision). Due to the particular clinical relevance, we always consider the two time-to-event endpoints progression-free survival (PFS, time to progression or death whatever occurs first) and overall survival (OS, time to death) as examples in the following. We describe the joint distribution of PFS and OS using an illness-death model. This is because an approach using multi-stage models appears to be particularly suitable for clinical practice. In this setting, the following natural classes of test problems arise: Evidence of significant superiority of a new therapy 
\begin{itemize}
	\item[(i)] with regard to the joint distribution of PFS and OS, or
	\item[(ii)] with respect to the marginal distributions of both PFS and OS, or
	\item[(iii)] with respect to the marginal distribution of PFS, or
	\item[(iv)] with respect to the marginal distribution of OS.
\end{itemize}
Adaptive hypothesis tests on the joint distribution of PFS and OS, which allow the simultaneous use of PFS and OS interim data, were described for single-arm phase IIa trials in \cite{Danzer22}, and for two-arm randomized phase III trials in \cite{Danzer23}. Here, single-arm means that the proof of superiority is provided against a reference known in advance and thus regarded as deterministic, while in the two-arm randomized trials the proof of superiority is provided against a control group collected concurrently.

Adaptive hypothesis tests on the marginal distribution of PFS, which allow the simultaneous use of PFS and OS interim data, can be obtained as described in \cite{W06} on the basis of the classical log-rank test with respect to PFS. This is because the classic log-rank statistic for PFS is already adapted to the joint filtration generated by PFS and OS events \citep[c.f.][]{Danzer22}. This does not apply to the classic log-rank statistics for OS. Therefore, adaptive hypothesis tests of class (iii) and (iv) cannot be obtained as described in \cite{W06} on the basis of the classical log-rank test with respect to OS if PFS interim data are also to be used for design modifications \cite[c.f.][]{BauerPosch04}.

Simultaneous adaptive hypothesis tests on the marginal distribution of PFS and OS that allow simultaneous use of PFS and OS interim data are constructed in Chapter \ref{sec:single_arm_pfs+os}. Adaptive hypothesis tests on the marginal distribution of OS that allow simultaneous use of PFS and OS interim data are constructed in the setting of single-arm phase IIa trials in Chapter \ref{sec:single_arm_os}, and in the setting of two-arm randomized phase III trials in Chapter \ref{sec:rct_os}. In Chapter \ref{sec:rct_pfs_os_gap}, we additionally consider adaptive hypothesis test on the gap between the PFS and OS curves. The general underlying notation is introduced in Chapter \ref{Sec02}. We conclude with a discussion of the results and an outlook for further future research in Chapter \ref{Discussion}. The derivation of the distributional properties of all involved test statistics is based on martingale theory, and is outsourced to the Appendix.

\section{General Aspects and Notation}\label{Sec02}

In confirmatory adaptive trials, patients are recruited sequentially in consecutive stages $\kappa = 1, \ldots, K$, and are assigned to one or more treatment arm $x$ in each stage. In the context of single arm trials, we let the treatment arm variable take the value $x=E$, only, for the experimental treatment under consideration. In the context of  two-armed randomized trials, we let the treatment arm variable take the values $x=E$ for experimental treatment and $x=C$ for control treatment. We denote by $\aleph^{(x \kappa)}$ the set of patients from treatment group $x$ and stage $\kappa$, and by $n^{(x \kappa)}\coloneqq |\aleph^{(x \kappa)}|$ the number of such patients. Likewise let $\aleph^{(x)}\coloneqq  \cup_{\kappa} \aleph^{(x \kappa)}$, $\aleph^{(\kappa)}\coloneqq  \cup_{x} \aleph^{(x \kappa)}$, $\aleph\coloneqq  \cup_{x, \kappa} \aleph^{(x \kappa)}$ as well as $n^{(x)}\coloneqq  \sum_{\kappa} n^{(x \kappa)}$,  $n^{(\kappa)}\coloneqq  \sum_{x} n^{(x \kappa)}$, $n\coloneqq  \sum_{x, \kappa} n^{(x \kappa)}$.

We assume that each patient is at risk for two types of failure: Progression of disease (state $1$, transient) and death with or without prior progression (state $2$, absorbing), where we assume that the potential transitions of a patient from the initial state $0$ (trial entry) to state $1$ and/or state $2$ occur according to Markovian illness--death model with these states. We let $T_i^j$ denote the time from entry to the $j$th type of failure for patient $i$, and $C_i$ the time from entry to censoring. For each $i$, the failure time vector $(T_i^1, T_i^2)$ and the censoring time $C_i$ are assumed to be independent (non-informative censoring). Likewise, any two vectors $(T_i^1, T_i^2, C_i)$ and $(T_{i'}^1, T_{i'}^2, C_{i'})$ for different patients $i \neq i'$ are assumed to be independent. Vectors $(T_i^1, T_i^2, C_i)$ for patients $i \in  \aleph^{(x)}$ from the same treatment group $x$ are assumed to be independent and identically distributed. For a patient $i$ from treatment group $x$, the joint distribution of $(T_i^1, T_i^2)$ is determined by three hazards
\begin{align}\label{eq:02.01}
\begin{split}
\alpha^{(j,\{1,2\}),(x)}(s,s,s) ds &\coloneqq     \mathbb{P}(s < T_i^j \leq s + ds | T_i^1 > s, T_i^2 > s)   \quad \text{for } j \in \{1,2\} \\ 
\alpha^{(2,\{2\}),(x)}(s,s_1,s) ds &\coloneqq     \mathbb{P}(s < T_i^2 \leq s + ds | T_i^1 = s_1, T_i^2 > s)  
\end{split}
\end{align}
The Markov assumption says that $\alpha^{(2,\{2\}),(x)}(s,s_1,s)$ does not depend on the value of $s_1$. A sufficient criterion for that is given in Appendix \ref{subsec:appendix_markov_sufficient}. Presupposing the Markov assumption in the following, we will thus write $\alpha_{01}^{(x)}(s)$, $\alpha_{02}^{(x)}(s)$ and $\alpha_{12}^{(x)}(s)$ instead of $\alpha^{(1,\{1,2\}),(x)}(s,s,s)$ , $\alpha^{(1,\{1,2\}),(x)}(s,s,s)$ and $\alpha^{(2,\{2\}),(x)}(s,s_1,s)$ in the sequel in order to simplify notation. Let $\alpha_{0*}^{(x)}(s)\coloneqq  \alpha_{01}^{(x)}(s)+\alpha_{02}^{(x)}(s)$. When data for patient $i$ are reviewed $s$ time units after his entry (so called \emph{in--study time} scale), we assume that the following information is observed and available for potential design modifications (such as sample size recalculation): 
\begin{align}\label{eq:02.02}
\begin{split}
&I(T_i^1 \leq s \wedge C_i \wedge T_i^2),  T_i^1 \cdot I(T_i^1 \leq s \wedge C_i \wedge T_i^2), \\
&I(T_i^2 \leq s \wedge C_i), 							T_i^2 \cdot I(T_i^2 \leq s \wedge C_i), \\
&I(C_i   \leq s \wedge T_i^2), 						C_i \cdot I(C_i \leq s \wedge T_i^2). 
\end{split}
\end{align}
Let ${\cal{F}}_{i,s}$ denote the $\sigma$--algebra generated by the random variables in \eqref{eq:02.02}, and let ${\cal{F}}_{s} \coloneqq  \sigma(\cup_{i \in \aleph} {\cal{F}}_{i,s})$.
The main focus in this work is adaptive hypothesis testing on overall survival OS (i.e., time $T_i^{OS}\coloneqq T_i^2$ from entry to death) and progression-free survival PFS (i.e., time $T_i^{PFS}\coloneqq T_i^1 \wedge T_i^2$ from entry to death or progression, whichever comes first). Accordingly, let $N_i^{OS}(s)\coloneqq I(T_i^2  \leq s \wedge C_i)$ denote the indicator for death at time $s$ after entry for patient $i$, and $N_i^{PFS}(s)\coloneqq I(T_i^1 \wedge T_i^2  \leq s \wedge C_i)$ the corresponding indicator for event. Basis for our method are patient--wise treatment specific stochastic processes
\begin{align}\label{eq:pfs_os_indiv_martingale}
\begin{split}
M_i^{PFS}(s) &\coloneqq  N_i^{PFS}(s) - \int_0^{T_i^1 \wedge T_i^2 \wedge C_i \wedge s} \alpha_{0 *}^{(x)}(u) du \\
M_i^{OS}(s)  &\coloneqq  N_i^{OS}(s)  - \int_0^{T_i^1 \wedge T_i^2 \wedge C_i \wedge s} \alpha_{0 2}^{(x)}(u) du   - \int_{T_i^1 \wedge T_i^2 \wedge C_i \wedge s}^{T_i^2 \wedge C_i \wedge s} \alpha_{1 2}^{(x)}(u) du 
\end{split}
\end{align}
for patients $i \in \aleph^{(x)}$ from treatment group $x$. Notice that $M_i^{PFS}(s)$ and $M_i^{OS}(s)$ are mean--zero $({\cal{F}}_{s})_{s \geq 0}$--martingales (c.f., \cite{Danzer22}, Corollary 2). 
These are used to define, for each treatment group $x$ and stage $\kappa$, the cumulative $({\cal{F}}_{s})_{s \geq 0}$--martingales
\begin{align}\label{eq:02.04}
\begin{split}
M_{12}^{(x \kappa)}(s) &\coloneqq  \sqrt{n^{(x \kappa)}} \cdot \sum_{i \in \aleph^{(x \kappa)}} \int_0^s \frac{I_{12,i}(s)}{Y_{12}^{(x \kappa)}(s)} dM_i^{OS}(s) \\
M_{02}^{(x \kappa)}(s) &\coloneqq  \sqrt{n^{(x \kappa)}} \cdot \sum_{i \in \aleph^{(x \kappa)}} \int_0^s \frac{I_{02,i}(s)}{Y_{02}^{(x \kappa)}(s)} dM_i^{OS}(s) \\
M_{0*}^{(x \kappa)}(s) &\coloneqq  \sqrt{n^{(x \kappa)}} \cdot \sum_{i \in \aleph^{(x \kappa)}} \int_0^s \frac{I_{02,i}(s)}{Y_{02}^{(x \kappa)}(s)} dM_i^{PFS}(s) 
\end{split}
\end{align}
with at risk indicators 
\begin{align}\label{eq:def_at_risk}
\begin{split}
I_{02,i}(s) &\coloneqq I(T_i^1 \wedge T_i^2 \wedge C_i \geq s), 
\quad I_{12,i}(s)\coloneqq I(T_i^1 \wedge T_i^2 \wedge C_i < s \leq T_i^2 \wedge C_i), \\
Y_{02}^{(x \kappa)}(s)&\coloneqq  \sum_{i \in \aleph^{(x \kappa)}} I_{02,i}(s),
\quad Y_{12}^{(x \kappa)}(s)\coloneqq  \sum_{i \in \aleph^{(x \kappa)}} I_{12,i}(s),  \\
J_{02}^{(x \kappa)}(s)&\coloneqq  I( Y_{02}^{(x \kappa)}(s) > 0 ),
\quad J_{12}^{(x \kappa)}(s)\coloneqq  I( Y_{12}^{(x \kappa)}(s) > 0 ). 
\end{split}
\end{align}
In particular, notice that $I_{12,i}(s) \cdot I_{02,i}(s) = 0$. The pairwise quadratic covariation processes of the martingales from \eqref{eq:02.04} are
\begin{align}\label{eq:na_martingale_covariation}
\begin{split}
[M_{12}^{(x \kappa)}](s) &= n^{(x \kappa)} \sum_{i \in \aleph^{(x \kappa)}} \int_0^s \frac{I_{12,i}(u)}{Y_{12}^{(x \kappa)}(u)^2} dN_i^{OS}(u), \\
[M_{02}^{(x \kappa)}](s) &= n^{(x \kappa)} \sum_{i \in \aleph^{(x \kappa)}} \int_0^s \frac{I_{02,i}(u)}{Y_{02}^{(x \kappa)}(u)^2} dN_i^{OS}(u), \\
[M_{0*}^{(x \kappa)}](s) &= n^{(x \kappa)} \sum_{i \in \aleph^{(x \kappa)}} \int_0^s \frac{I_{02,i}(u)}{Y_{02}^{(x \kappa)}(u)^2} dN_i^{PFS}(u), \\
[M_{12}^{(x \kappa)},M_{02}^{(x \kappa)}](s) &=0, \\
[M_{12}^{(x \kappa)},M_{0*}^{(x \kappa)}](s) &=0, \\
[M_{02}^{(x \kappa)},M_{0*}^{(x \kappa)}](s) &= n^{(x \kappa)} \sum_{i \in \aleph^{(x \kappa)}} \int_0^s \frac{I_{02,i}(u)}{Y_{02}^{(x \kappa)}(u)^2} dN_i^{OS}(u) \equiv [M_{02}^{(x \kappa)}](s).
\end{split}
\end{align}
The predictable variation processes are given in differential form by
\begin{align}\label{eq:na_martingale_covariation_prev}
	\begin{split}
		d\langle M_{12}^{(x \kappa)}\rangle (s) &= n^{(x \kappa)} \frac{J_{12}^{(x \kappa)}(s)}{Y_{12}^{(x \kappa)}(s)} \alpha_{12}(s) ds, \\
		d\langle M_{02}^{(x \kappa)}\rangle (s) &= n^{(x \kappa)} \frac{J_{02}^{(x \kappa)}(s)}{Y_{02}^{(x \kappa)}(s)} \alpha_{02}(s) ds, \\
		d\langle M_{0*}^{(x \kappa)}\rangle (s) &= n^{(x \kappa)} \frac{J_{01}^{(x \kappa)}(s)}{Y_{01}^{(x \kappa)}(s)} \alpha_{01}(s) ds + d\langle M_{02}^{(x \kappa)}\rangle (s), \\
		d\langle M_{12}^{(x \kappa)},M_{02}^{(x \kappa)}\rangle (s) &=0, \\
		d\langle M_{12}^{(x \kappa)},M_{0*}^{(x \kappa)}\rangle (s) &=0, \\
		d\langle M_{02}^{(x \kappa)},M_{0*}^{(x \kappa)}\rangle (s) &= d\langle M_{02}^{(x \kappa)}\rangle (s).
	\end{split}
\end{align}
Finally, for each treatment group $x$, we will also need stage-wise Nelsen--Aalen type estimators 
\begin{align}\label{eq:na_estimators}
\begin{split}
\hat{\Lambda}_{12}^{(x \kappa)}(s)  &= \sum_{i \in \aleph^{(x \kappa)}} \int_0^s \frac{I_{12,i}(u)}{Y_{12}^{(x \kappa)}(u)} dN_i^{OS}(u)   \\
\hat{\Lambda}_{02}^{(x \kappa)} (s) &= \sum_{i \in \aleph^{(x \kappa)}} \int_0^s \frac{I_{02,i}(u)}{Y_{02}^{(x \kappa)}(u)} dN_i^{OS}(u)  \\
\hat{\Lambda}_{0*}^{(x \kappa)} (s) &= \sum_{i \in \aleph^{(x \kappa)}} \int_0^s \frac{I_{02,i}(u)}{Y_{02}^{(x \kappa)}(u)} dN_i^{PFS}(u)  
\end{split}
\quad \longleftrightarrow \quad
\begin{split}
d\hat{\Lambda}_{12}^{(x \kappa)}(s) &= \sum_{i \in \aleph^{(x \kappa)}}  \frac{I_{12,i}(s)}{Y_{12}^{(x \kappa)}(s)} dN_i^{OS}(s)  \\
d\hat{\Lambda}_{02}^{(x \kappa)}(s) &= \sum_{i \in \aleph^{(x \kappa)}}  \frac{I_{02,i}(s)}{Y_{02}^{(x \kappa)}(s)} dN_i^{OS}(s)\\
d\hat{\Lambda}_{0*}^{(x \kappa)}(s) &= \sum_{i \in \aleph^{(x \kappa)}}  \frac{I_{02,i}(s)}{Y_{02}^{(x \kappa)}(s)} dN_i^{PFS}(s) 
\end{split}
\end{align}
of the cumulative hazard functions $\Lambda_{12}^{(x)}(s)$, $\Lambda_{02}^{(x)}(s)$, and $\Lambda_{0*}^{(x)}(s)$. 
Notice that in differential form, we may thus write
\begin{align}\label{eq:na_martingale}
	\begin{split}
		dM_{j2}^{(x \kappa)}(s)&= \sqrt{n^{(x \kappa)}}  \left( d\hat{\Lambda}_{j2}^{(x \kappa)}(s)  - J_{j2}^{(x \kappa)}(s) d{\Lambda}_{j2}^{(x \kappa)}(s)  \right), \quad \text{for } j=1,2, \\ 
		dM_{0*}^{(x \kappa)}(s)&= \sqrt{n^{(x \kappa)}}  \left( d\hat{\Lambda}_{0*}^{(x \kappa)}(s)  - J_{02}^{(x \kappa)}(s) d{\Lambda}_{0*}^{(x \kappa)}(s)  \right).
	\end{split}
\end{align}
Since $\Lambda_{01}^{(x)}(s)=\Lambda_{0*}^{(x)}(s) - \Lambda_{02}^{(x)}(s)$, also set
\begin{align*}
\begin{split}
\hat{\Lambda}_{01}^{(x \kappa)}(s) &\coloneqq  \hat{\Lambda}_{0*}^{(x \kappa)}(s) - \hat{\Lambda}_{02}^{(x \kappa)}(s) \\
M_{01}^{(x \kappa)}(s) &\coloneqq  M_{0*}^{(x \kappa)}(s) - M_{02}^{(x \kappa)}(s) 
\end{split}
\end{align*}
Finally, let us also introduce the progression-free survival function $S_{PFS}^{(x)}(s)\coloneqq \mathbb{P}(T_i^1 \wedge T_i^2 \geq s)$ and the overall survival function $S_{OS}^{(x)}(s)\coloneqq \mathbb{P}(T_i^2 \geq s)$ within treatment group $x$, together with the corresponding cumulative hazards functions $\Lambda_{PFS}^{(x)}(s) \coloneqq  \log(S_{PFS}^{(x)}(s))$ and $\Lambda_{OS}^{(x)}(s) \coloneqq  \log(S_{OS}^{(x)}(s))$ where $\log(\cdot)$ is the natural logarithm. 

\section{Simultaneous adaptive testing of hypotheses on PFS and OS}\label{sec:single_arm_pfs+os}

\subsection{The null hypothesis}\label{Sec03.01}
We consider a single-arm clinical trial in which patients are recruited in successive stages $\kappa = 1, \ldots, K$ in order to be treated with some experimental treatment $x=E$, say. In this setting, we consider testing the two-sided null-hypothesis that the progression-free survival $S_{PFS}^{(E)}(s)$ and the overall survival $S_{OS}^{(E)}(s)$ of the patients both coincide with prefixed reference survival functions $S_{PFS,0}(s)$ and $S_{OS,0}(s)$, respectively, i.e.
\begin{align}\label{eq03.01}
H_0: S_{PFS}^{(E)}(s) = S_{PFS,0}(s) \text{ and } S_{OS}^{(E)}(s) = S_{OS,0}(s). 
\end{align}
For purposes of power calculation we additionally introduce the (stage-wise) \emph{contiguous alternatives}
\begin{align}\label{eq:contiguous_alternatives_both}
H_{1,n}: S_{PFS}^{(E)}(s) = S_{PFS,0}(s)^{\omega_{PFS}} \text{ and } S_{OS}^{(E)}(s) = S_{OS,0}(s)^{\omega_{OS}}
\end{align} 
with hazard ratios ${\omega_{PFS}}= \exp(- {\gamma_{PFS}}/\sqrt{n^{(E \kappa )}})$ and ${\omega_{OS}}= \exp(- {\gamma_{OS}}/\sqrt{n^{(E \kappa)}})$ for $\gamma_{PFS}, \gamma_{OS} > 0$. Whereas ${\omega_{OS}}$ and ${\omega_{PFS}}$ are the clinically relevant hazard ratios, the parameters $\gamma_{PFS}, \gamma_{OS} > 0$ are the quantities that will appear in the normal approximation of our test statistic.

\subsection{The test statistic for the OS component of $H_0$}\label{Sec03.02}
On first sight, the classical one--sample log--rank statistic $LR_{OS}(s)\coloneqq  \sum_i [N_i^{OS}(s) - \Lambda_{OS}(T_i^2 \wedge C_i \wedge s)]$ appears as natural candidate to test the OS component of null hypothesis $H_0$, as it is a mean--zero martingale when $H_0$ holds true. The martingale property of $LR_{OS}(s)$, however, only holds with respect to the filtration generated exclusively by OS-events. This is not sufficient for our purposes, as we aim for an adaptive test of $H_0$ that allows for data-dependent sample size adjustments based on PFS and OS interim data simultaneously. We will instead construct alternative stage-wise statistics $\Psi_{OS}^{(\kappa)}(s)$ that ensure the martingale property with respect to the filtration $({\cal{F}}_{s})_{s \geq 0}$ generated jointly by OS- and PFS-events.

Starting point is the following identity for the overall survival function that holds true under the Markov assumption (see Appendix \ref{subsec:appendix_os_identity})
\begin{align}\label{eq:os_identity}
S_{OS}^{(E)}(t) = 1 - \int_0^t S_{PFS}^{(E)}(u) \alpha_{02}^{(E)}(u) du - \int_0^t [S_{OS}^{(E)}(u) - S_{PFS}^{(E)}(u)] \alpha_{12}^{(E)}(u) du.
\end{align}
Identity \eqref{eq:os_identity} motivates to introduce the stochastic process
\begin{align}\label{eq03.03}
\Psi_{OS}^{(\kappa)}(s)\coloneqq  \sqrt{n^{(E \kappa)}} \left(   
										\int_0^s S_{PFS,0}(u) d\hat{\Lambda}_{02}^{(E \kappa)}(u) 
										+ \int_0^s [S_{OS,0}(u) - S_{PFS,0}(u)] d\hat{\Lambda}_{12}^{(E \kappa)}(u) 
										+ S_{OS,0}(u) - 1
										\right)
\end{align}
As detailed in Appendix \ref{sec:appendix_distributional_properties}, if null hypothesis $H_0$ from \eqref{eq03.01}) holds true, $\Psi_{OS}(s)$ is by Eqs. \eqref{eq:na_estimators} and \eqref{eq:os_identity} an $({\cal{F}}_{s})_{s \geq 0}$--martingale with mean zero and quadratic covariation process
\begin{align}\label{eq03.04}
[\Psi_{OS}^{(\kappa)}](s)= n^{(E \kappa)} \sum_{i \in \aleph^{(E \kappa)}}     
										\int_0^s \left( S_{PFS,0}(u)^2 \frac{I_{02,i}(u)}{Y_{02}^{(E \kappa)}(u)^2} 
									+ [S_{OS,0}(u) - S_{PFS,0}(u)]^2 \frac{I_{12,i}(u)}{Y_{12}^{(E \kappa)}(u)^2} \right) dN_i^{OS}(u). 							
\end{align}
The process $\Psi_{OS}^{(\kappa)}$ suffices the assumptions of Lemma \ref{lemma:rebolledo_condition_iii} and $[\Psi_{OS}^{(\kappa)}](s) \to \sigma_{OS}(s)^2$ in probability as $n^{(E \kappa)} \to \infty$ for some non--decreasing deterministic function $\sigma_{OS}(s)^2$  (see Appendix \ref{sec:appendix_distributional_properties}). By a central limit theorem for local martingales \citep{bib5}, if null hypothesis $H_0$ is true, $\Psi_{OS}^{(\kappa)}(s)$ thus converges in distribution to a mean zero, $({\cal{F}}_{s})_{s \geq 0}$--adapted  Gaussian process with independent increments whose variance function $\sigma_{OS}(s)^2$ can be consistently estimated in stage $\kappa$ by $\hat{\sigma}_{OS}^{(\kappa)}(s)^2\coloneqq [\Psi_{OS}^{(\kappa)}](s)$ in the large sample limit $n^{(E \kappa)} \to \infty$. The stage-wise processes $\Psi_{OS}^{(\kappa)}(s)$ are thus predestined to construct adaptive tests of $H_0$ as will be done in Section \ref{Sec03.07}, as they are adapted to the filtration $({\cal{F}}_{s})_{s \geq 0}$ generated jointly by PFS and OS events.

\subsection{Distribution of $\Psi_{OS}^{(\kappa)}(s)$ under the contiguous alternatives}\label{Sec03.03}

For sample size calculation we are interested in the distribution of $\Psi_{OS}(s)$ under the contiguous alternatives $H_1$ from Eq. \eqref{eq:contiguous_alternatives_both}. As detailed in Appendix \ref{sec:appendix_distributional_properties}, if $H_1$ is true, in the large sample limit $n^{(E \kappa)} \to \infty$ the process $\Psi_{OS}^{(\kappa)}(s)$ from \eqref{eq03.03} converges in distribution to a Gaussian process with independent increments, drift function
\begin{align}\label{eq:os_drift_intersection}
\begin{split}
\mu_{OS}(s) &\coloneqq   \gamma_{OS} \cdot \mu_{OS,A}(s)  +  \gamma_{PFS} \cdot \mu_{OS,B}(s) 
\end{split}
\end{align}
with nuisance parameters
\begin{align*}
\begin{split}
\mu_{OS,A}(s) &\coloneqq  S_{OS,0}(s) \cdot \log(S_{OS,0}(s)) + \int_0^s S_{OS,0}(u) \cdot \log(S_{OS,0}(u)) d\Lambda_{12}^{(E )}(u) \\
\mu_{OS,B}(s) &\coloneqq  \int_0^s S_{PFS,0}(u) \cdot \log(S_{PFS,0}(u)) d(\Lambda_{02}^{(E)}(u) - \Lambda_{12}^{(E)}(u))
\end{split}
\end{align*}
and variance function $\sigma_{OS}(s)^2$ that can be consistently estimated in stage $\kappa$ by $\hat{\sigma}_{OS}^{(E \kappa)}(s)^2 \coloneqq  [\Psi_{OS}^{(\kappa)}](s)$ from Eq. \eqref{eq03.04}.
The nuisance parameters $\mu_{OS,A}(s)$ and $\mu_{OS,B}(s)$ may be estimated in stage $\kappa$ via
\begin{align*}
\begin{split}
\hat{\mu}_{OS,A}^{(\kappa)}(s) &\coloneqq  S_{OS,0}(s) \cdot \log(S_{OS,0}(s)) + \int_0^s S_{OS,0}(u) \cdot \log(S_{OS,0}(u)) d\hat{\Lambda}_{12}^{(E \kappa)}(u) \\
\hat{\mu}_{OS,B}^{(\kappa)}(s) &\coloneqq  \int_0^s S_{PFS,0}(u) \cdot \log(S_{PFS,0}(u)) d(\hat{\Lambda}_{02}^{(E \kappa)}(u) - \hat{\Lambda}_{12}^{(E \kappa)}(u))
\end{split}
\end{align*}
where $d\hat{\Lambda}_{ij}^{(E \kappa)}$ is the stage-$\kappa$ estimate for $d{\Lambda}_{ij}^{(E)}$ from \eqref{eq:na_estimators}. In particular, since $E(\Psi_{OS}^{(\kappa)}(s)) = \mu_{OS}(s)$, the hazard ratio $\omega_{OS}$ w.r.t. OS may be estimated at stage $\kappa$ by 
\begin{align}\label{eq03.07}
  \hat{\omega}_{OS}^{(\kappa)} \coloneqq 
1 - \frac{    \frac{\Psi_{OS}^{(\kappa)}(s) }{ \sqrt{n^{(E \kappa)}}}       + \hat{\mu}_{OS,B}^{(\kappa)}(s) \cdot \hat{\omega}^{\kappa}_{PFS}  }{  \hat{\mu}_{OS,A}^{(\kappa)}(s)  }
\end{align}
where $\hat{\omega}^{\kappa}_{PFS}$ is the stage-$\kappa$ estimate for the true hazard ratio with respect to PFS from \eqref{eq03.51}. Notice that the parameter $n^{(E \kappa)}$ cancels out in \eqref{eq03.07} due to the chosen normalization of $\Psi_{OS}^{(\kappa)}$.

\subsection{The test statistic for the PFS component of $H_0$}\label{Sec03.04}
Unlike in the case of OS, the one--sample log--rank statistic $LR_{PFS}(s)\coloneqq  \sum_i [N_i^{PFS}(s) - \Lambda_{PFS}(T_i^1 \wedge T_i^2 \wedge C_i \wedge s)]$ for PFS is a valid candidate to test the PFS component of null hypothesis $H_0$ also in an adaptive setting, because it is a mean--zero martingale with respect to the filtration $({\cal{F}}_{s})_{s \geq 0}$ generated jointly by OS- and PFS-events when $H_0$ holds true. However, for symmetry reasons, we will use alternative stage-wise statistics $\Psi_{PFS}^{(\kappa)}(s)$ that is structurally similar to the statistics $\Psi_{OS}^{(\kappa)}(s)$ introduced in Chapter \ref{Sec03.04}.

Starting point is the identity  
\begin{align}\label{eq03.08}
S_{PFS}^{(E)}(s) = 1 - \int_0^t S_{PFS}^{(E)}(u) \alpha_{0*}^{(E)}(u) du
\end{align}
that immediately follows from the relations $S_{PFS}^{(E)}(s)=e^{- \Lambda_{PFS}^{(E)}(s)}$ and $\Lambda_{PFS}^{(E)}(s) = \int_0^s \alpha_{0*}^{(E)}(u) du$. Identity \eqref{eq03.08} motivates to introduce the stage-wise stochastic processes
\begin{align}\label{eq03.09}
\Psi_{PFS}^{(\kappa)}(s)\coloneqq  \sqrt{n^{(E \kappa)}} \left( \int_0^s S_{PFS,0}(u) d\hat{\Lambda}_{0*}^{(E \kappa)}(u) + S_{PFS,0}(s) - 1 \right).
\end{align}
When null hypothesis $H_0$ from \eqref{eq03.01} holds true, $\Psi_{PFS}^{(\kappa)}(s)$ is an $({\cal{F}}_{s})_{s \geq 0}$--martingale with mean zero by \eqref{eq03.08} and quadratic covariation process
\begin{align}\label{eq03.10}
[\Psi_{PFS}^{(\kappa)}](s)= n^{(E\kappa)} \sum_{i \in \aleph^{(E \kappa)}}    
										\int_0^s S_{PFS,0}(u)^2 \frac{I_{02,i}(u)}{Y_{02}^{(E \kappa)}(u)^2} dN_i^{PFS}(u). 							
\end{align}
The process $\Psi_{PFS}^{(\kappa)}$ suffices the assumptions of Lemma \ref{lemma:rebolledo_condition_iii} and $[\Psi_{PFS}^{(\kappa)}](s) \to \sigma_{PFS}(s)^2$ in probability as $n^{(E \kappa)} \to \infty$ for some non--decreasing deterministic function $\sigma_{PFS}(s)^2$ by a law of large numbers (c.f. Appendix \ref{sec:appendix_distributional_properties}). By a central limit theorem for local martingales \citep{bib5}, if null hypothesis $H_0$ is true and as sample size increases, the stage-wise processes $\Psi_{PFS}^{(\kappa)}(s)$ thus converge in distribution to a mean zero, $({\cal{F}}_{s})_{s \geq 0}$--adapted  Gaussian process with independent increments whose variance function $\sigma_{PFS}(s)^2$ can be consistently estimated in stage $\kappa$ by $\hat{\sigma}_{PFS}^{(\kappa)}(s)^2\coloneqq [\Psi_{PFS}^{(\kappa)}](s)$.

\subsection{Distribution of $\Psi_{PFS}^{(\kappa)}(s)$ under the contiguous alternatives}\label{Sec03.05}

For sample size calculation we also need the distribution of $\Psi_{PFS}^{(\kappa)}(s)$ under the contiguous alternatives $H_1$. As detailed in Appendix \ref{sec:appendix_distributional_properties}, if $H_1$ holds true and as sample size $n^{(E \kappa)}$ increases, the process $\Psi_{PFS}^{(\kappa)}(s)$ from \eqref{eq03.09} converges in distribution to a Gaussian process with independent increments, drift function
\begin{align}\label{eq03.50}
\begin{split}
\mu_{PFS}(s) &\coloneqq   \gamma_{PFS} \cdot \left( 1 - S_{PFS,0}(s) \right) 
\end{split}
\end{align}
and variance function $\sigma_{PFS}(s)^2$ that may be estimated consistently by $\hat{\sigma}_{PFS}^{(\kappa)}(s)^2\coloneqq [\Psi_{PFS}^{(\kappa)}](s)$ from \eqref{eq03.10}. In particular, since $E(\Psi_{PFS}^{(\kappa)}(s)) = \mu_{PFS}(s)$, the hazard ratio $\omega_{PFS}$ w.r.t. PFS may be estimated in stage $\kappa$ by 
\begin{align}\label{eq03.51}
  \hat{\omega}_{PFS}^{(\kappa)} \coloneqq 
1 - \frac{   \Psi_{PFS}^{(\kappa)}(s) }{ \sqrt{ n^{(E \kappa)} } \left( 1 - S_{PFS,0}(s) \right)  }.
\end{align}
Notice that the parameter $n^{(E \kappa)}$ cancels out in \eqref{eq03.51} due to the chosen normalization of $\Psi_{PFS}^{(\kappa)}$.

\subsection{Correlation between $\Psi_{OS}^{(\kappa)}(s)$ and $\Psi_{PFS}^{(\kappa)}(s)$}\label{Sec03.06}

Assume that either null hypothesis $H_0$ or the contiguous alternatives $H_1$ (see Section \ref{Sec03.01}) hold true. Then there is even the stronger statement that, as $n \to \infty$, the bivariate process $\mathbf{\Psi}^{(\kappa)}(s)\coloneqq (\Psi_{OS}^{(\kappa)}(s),\Psi_{PFS}^{(\kappa)}(s))$ converges in distribution to a bivariate Gaussian process $\mathbf{\Psi}^{\infty}(s)$ with independent increments. 
Thus, for each fixed $0 \leq s_1 \leq s_2$, $\mathbf{\Psi}^{\infty}(s_1) \sim N(\mathbf{\mu}(s_1),\mathbf{\Sigma}(s_1))$ and $\mathbf{\Psi}^{\infty}(s_2) - \mathbf{\Psi}^{\infty}(s_1) \sim N(\mathbf{\mu}(s_2)-\mathbf{\mu}(s_1),\mathbf{\Sigma}(s_2)-\mathbf{\Sigma}(s_1))$ are independent and normally distributed with
\begin{equation*}
{{\mathbf{\mu}(s)}} \coloneqq  
\left(\begin{array}{c}{\mu_{OS}(s)}  \\ {\mu_{PFS}(s)} \end{array}\right) \ 
{\mathrm{ and }} \ 
\mathbf{\Sigma}(s) \coloneqq 
\left(\begin{array}{cc}
\sigma^{2}_{OS}(s) 							   &    \rho(s) \cdot \sigma_{OS}(s) \cdot \sigma_{PFS}(s)\\
\rho(s) \cdot \sigma_{OS}(s) \cdot \sigma_{PFS}(s) &    \sigma^{2}_{PFS}(s)
\end{array}\right).
\end{equation*}
The drift components $\mu_{OS}(s)$ and $\mu_{PFS}(s)$ results according to Eqs. \eqref{eq:os_drift_intersection} and \eqref{eq03.50}, respectively. In particular, we have zero drift $\mathbf{\mu}(s)=0$ when null hypothesis $H_0$ holds true. The variances $\sigma^{2}_{OS}(s)$ and $\sigma^{2}_{PFS}(s)$ may be estimated consistently in stage $\kappa$ by $\hat{\sigma}_{OS}^{(\kappa)}(s)^2\coloneqq [\Psi_{OS}^{(\kappa)}](s)$ and $\hat{\sigma}_{PFS}^{(\kappa)}(s)^2\coloneqq [\Psi_{PFS}^{(\kappa)}](s)$ from Eqs. \eqref{eq03.04} and \eqref{eq03.10}. The correlation function $\rho(s)$ may be estimated consistently in stage $\kappa$ by
\begin{align}\label{eq03.60}
\begin{split}
\hat{\rho}^{(\kappa)}(s) &\coloneqq  \frac{ n^{(E \kappa)} }{ \hat{\sigma}_{OS}^{(\kappa)}(s) \cdot \hat{\sigma}_{PFS}^{(\kappa)}(s) } \cdot  \sum_{i \in \aleph^{(E \kappa)} } \int_0^s S_{PFS,0}(u)^2 \frac{I_{02,i}(u)}{Y_{02}^{(E \kappa)}(u)^2} dN_i^{OS}(u)
\end{split}
\end{align}
Notice that the parameter $n^{(E \kappa)}$ cancels out in \eqref{eq03.60} due to the chosen normalization of $\hat{\sigma}_{OS}^{(\kappa)}(s)$ and $\hat{\sigma}_{PFS}^{(E \kappa)}(s)^2$.

\subsection{A two--stage adaptive test of $H_0$}\label{Sec03.07}

In many oncological diseases, the early occurrence of progression is associated with an unfavorable course of disease. The occurrence of early progression can then act as a surrogate parameter for later deaths in the long-term course. For a trial investigating long-term progression-free and overall survival, it therefore makes sense to perform an interim analysis to decide early on the continuation and, if necessary, sample size adjustment of the trial based on the observed frequency of early progressions. To this end, we propose a two-stage adaptive test of $H_0$ from Eq. \eqref{eq03.01}, which allows data-dependent sample size adjustment of the trial based on the observed short-term PFS data.

In a non--adaptive setting, a test of null hypothesis $H_0$ is commonly based on the union--intersection--test $p$--value $\tilde{P}\coloneqq 2 \min(\tilde{P}_{OS},\tilde{P}_{PFS})$ where $\tilde{P}_{OS}$ and $\tilde{P}_{PFS}$ are the $p$--values of the classical one--sample log--rank statistic $L_{OS}$ and $L_{PFS}$ for OS and PFS, respectively. This approach is not adequate in our adaptive setting, because $L_{OS}$ is not adapted to the joint filtration $({\cal{F}}_s)_{s \geq 0}$ of PFS and OS events, which implies that PFS interim data cannot be used to inform design changes. Instead, we will use stage-wise union--intersection--test $p$--values derived from the statistics $\Psi_{OS}^{(\kappa)}$ and $\Psi_{PFS}^{(\kappa)}$. 

Let the trial start at calendar time zero. At calender time $t_1>0$, an interim analysis will be performed to determine short-term PFS and OS within the first $s_1$ months since start of therapy, e.g. $s_1 =$ 6-months after start of treatment. Of course, we require $s_1 < t_1$. Patients recruited prior to calender time $a_1 \coloneqq  t_1 - s_1$ define the set $\aleph^{(E1)}$ of stage $1$ patients, and will be part of the interim analysis. Assume there are $n^{(E1)}$ stage $1$ patients. To quantify the short-term PFS and OS, the interim analysis calculates the statistics 
\begin{align*}
\tilde{Z}_{x,11} & \coloneqq \frac{ \Psi_{x}^{(1)}(s_1) }{ \sqrt{ [\Psi_{x}^{(1)}](s_1) } } 
\end{align*}
for $x=PFS$ or $x = OS$. Herein, $\Psi_{OS}^{(1)}$ and $\Psi_{PFS}^{(1)}$ denote the processes from Eqs. \eqref{eq03.03} and \eqref{eq03.09} calculated within the set of stage $1$ patients, respectively. The statistics $Z_{PFS,11}$ and $Z_{OS,11}$ are normally distributed with a drift proportional to the true hazard ratio with respect to PFS and OS, respectively (c.f. Section \ref{Sec03.06}). $Z_{PFS,11}$ and $Z_{OS,11}$ may thus be used as interim measures for the short-term treatment effect w.r.t. PFS and OS. Therefore, the observed value of $Z_{x,11}$ shall be used to decide whether the trial stops at the interim analysis or goes on to a second stage.

In case that the trial goes on, assume that $n^{(E2)}$ new patients are recruited during further $A_2$ months. These $n^{(E2)}$ patients recruited between calendar time $a_1$ and $a_1 + A_2$ define the set $\aleph^{(E2)}$ of stage $2$ patients. The final analysis will be at calendar time $a_1 + A_2 +f$. The adaptive component of the design is that $A_2$ (and in turn $n^{(E2)}$) may be chosen as a function of $Z_{PFS,11}$ and $Z_{OS,11}$, and is thus random. In principle, we may use for sample size recalculation even the complete sigma algebra ${\cal{F}}_{s_1}^{(1)}\coloneqq  \sigma( \cup_{ i \in \aleph^{(E1)}} {\cal{F}}_{i,s_1}  )$ of short--term PFS and OS data observed in stage $1$ patients within their first $s_1$ months of follow--up. Now, let $\Psi_{OS}^{(2)}$ and $\Psi_{PFS}^{(2)}$ be the processes from Eqs. \eqref{eq03.03} and \eqref{eq03.09} calculated within the set $\aleph^{(E2)}$ of stage 2 patients, respectively. Also set $S_2 \coloneqq  a_1 + A_2 +f$. At the time of final analysis, the increment statistics in stage $1$ patients
\begin{align*}
\begin{split}
\tilde{Z}_{x,12} & \coloneqq  \frac{  \Psi_{x}^{(1)}(S_2) - \Psi_{x}^{(1)}(s_1)  }{  \sqrt{   [\Psi_{x}^{(1)}](S_2) - [\Psi_{x}^{(1)}](s_1)   }    } 
\end{split}
\end{align*}
as well as the overall statistic in stage $2$ patients
\begin{align*}
\begin{split}
\tilde{Z}_{x,2} & \coloneqq  \frac{  \Psi_{x}^{(2)}(S_2)  }{  \sqrt{   [\Psi_{x}^{(2)}](S_2)   }    } \\
\end{split}
\end{align*}
will be calculated for $x= PFS$ or $x=OS$. For arbitrary weights $w_1$, $w_2 >0$ with $w_1^2+w_2^2=1$, these may be combined to cause--specific stage--wise statistics 
\begin{align*}
\begin{split}
Z_{x,1} & \coloneqq  \tilde{Z}_{x,11} \quad  Z_{x,2}\coloneqq  w_1 \tilde{Z}_{x,12} + w_2 \tilde{Z}_{x,2}, 
\end{split}
\end{align*}
for $x= PFS$ or $x=OS$. The adaptive test of null hypothesis $H_0$ will finally be based on the stage--wise $p$--values
\begin{align*}
\begin{split}
P_{j} &\coloneqq  \min(P_{PFS,j}/\eta_{PFS,j},P_{OS,j}/\eta_{OS,j}) 
\end{split}
\end{align*}
with cause--specific stage--wise $p$--values $P_{x,j} \coloneqq  \Phi(Z_{x,j})$ for $x=OS, PFS$ and $j=1,2$, and prefixed weights $\eta_{x,j}>0$ such that $\eta_{PFS,j} + \eta_{OS,j} \leq 1$ for $j=1,2$. 

The two--stage adaptive test of $H_0$ with significance level $\alpha$ is defined by the rejection region $\{ P_1 \leq \alpha_1 \} \cup \{ \alpha_1 < P_1 \leq \alpha_0, P_2 \leq \alpha(P_1) \}$ for some stage--one rejection bound $0 < \alpha_1 < \alpha$, some binding stopping for futility bound $\alpha_0 \in (\alpha_1,1]$, and some non-increasing conditional error function $\alpha(\cdot):(\alpha_1,\alpha_0] \to [0,1]$ with $\int_{\alpha_1}^{\alpha_0} \alpha(x) dx = \alpha - \alpha_1$. 

Notice that the weights $w_j$, $\eta_{x,j}$ and design parameters  $\alpha_1$, $\alpha_0$ and $\alpha(\cdot)$ have to be fixed in advance and must remain unchanged during the trial.  At the interim analysis, however, the parameter $S_2 \equiv a_1 + A_2 + f$ (i.e. the length of accrual $a\coloneqq a_1 + A_2$ and of follow-up $f$ and thus the stage-two sample size $n^{(E2)}$) may be modified in the light of all short--term PFS and OS data ${\cal{F}}_{s_1}^{(1)}$ observed in stage $1$ patients within their first $s_1$ months of follow--up. 
The reason is that for sufficiently large sample size $Z_{OS,2}$ and $Z_{PFS,2}$ are statistically independent from $Z_{OS,1}$, $Z_{PFS,1}$ (and even ${\cal{F}}_{s_1}^{(1)}$) due to the asymptotic independent increments structure of $(\Psi_{OS}^{(1)},\Psi_{PFS}^{(1)})$. This ensures strict type one error rate control of the adaptive test despite of data-dependent sample size recalculation based on $Z_{OS,1}$, $Z_{PFS,1}$ and even ${\cal{F}}_{s_1}^{(1)}$.

\section{Adaptive Testing of Hypotheses on OS in Single-Arm Trials}\label{sec:single_arm_os}

\subsection{The Null Hypothesis}\label{Sec04.01}
We consider a single-arem clinical trial in which patients are recruited in successive stages $\kappa = 1, \ldots, K$ in order to be treatment with some experimental treatment $x=E$, say. In this setting, we consider testing the two-sided null-hypothesis that the overall survival $S_{OS}^{(E)}(s)$ of the patients coincides with a prefixed reference survival function $S_{OS,0}(s)$, i.e.
\begin{align*}
H_0: S_{OS}^{(E)}(s) = S_{OS,0}(s). 
\end{align*}
We aim for an adaptive test of $H_0$ that allows for data-dependent sample size adjustments based on PFS and OS interim data simultaneously. On first sight, the one--sample log--rank statistic $LR_{OS}(s)\coloneqq  \sum_i [N_i^{OS}(s) - \Lambda_{OS,0}(T_i^2 \wedge C_i \wedge s)]$, $\Lambda_{OS,0} \coloneqq  \log(S_{OS,0})$, appears as natural candidate to test null hypothesis $H_0$, as it is a mean--zero martingale when $H_0$ holds true. The martingale property of $LR_{OS}(s)$, however, only holds with respect to the filtration generated exclusively by OS-events. This is not sufficient for our purposes, as we would like to make sample size adjustments based on interim information regarding OS and PFS both. In Section \ref{Sec04.02}, we will thus construct an alternative statistic $\Psi(s)$ that ensures the martingale property with respect to the filtration $({\cal{F}}_{s})_{s \geq 0}$ generated jointly by OS- and PFS-events.

For purposes of power calculation we additionally introduce the (stage-wise) \emph{contiguous alternatives}
\begin{align}\label{eq:contiguous_alternatives_os}
H_{1,n}: S_{OS}^{(E)}(s) = S_{OS,0}(s)^{\omega_{OS}} \text{ with hazard ratio }\omega_{OS}= \exp(- {\gamma_{OS}}/\sqrt{n^{(E \kappa)}}).
\end{align} 
for $\gamma_{OS} > 0$. Whereas ${\omega_{OS}}$ is the clinically relevant hazard ratio, the parameter $\gamma_{OS} > 0$ is the quantity that will appear in the normal approximation of our test statistic.

\subsection{The Test Statistic}\label{Sec04.02}

Starting point is the following identity for the overall survival function under experimental treatment $E$ that holds true under the Markov assumption (see Appendix \ref{sec:appendix_distributional_properties})
\begin{align*}
S_{OS}^{(E)}(s) = 1 - \int_0^t S_{PFS}^{(E)}(u) \alpha_{02}^{(E)}(u) du - \int_0^t [S_{OS}^{(E)}(u) - S_{PFS}^{(E)}(u)] \alpha_{12}^{(E)}(u) du.
\end{align*}
Identity \eqref{eq:os_identity} motivates to introduce the stage-wise statistics
\begin{align}\label{eq04.04}
\Psi^{(\kappa)}(s)\coloneqq  \sqrt{n^{(E \kappa)}} \left(   
										\int_0^s \hat{S}_{PFS}^{(E \kappa)}(u) d\hat{\Lambda}_{02}^{(E \kappa)}(u) 
										+ \int_0^s [S_{OS,0}(u) - \hat{S}_{PFS}^{(E \kappa)}(u)] d\hat{\Lambda}_{12}^{(E \kappa)}(u) 
										+ S_{OS,0}(u) - 1
										\right)
\end{align}
where $\hat{S}_{PFS}^{(E \kappa)}(s)\coloneqq  \exp(- \hat{\Lambda}_{0*}^{(E \kappa)}(s-))\coloneqq \lim_{\epsilon \to 0}\exp(- \hat{\Lambda}_{0*}^{(E \kappa)}(s-\epsilon))$ is a consistent estimator of the true progression-free survival function $S_{PFS}^{(E)}(s)$ under experimental treatment $E$. 
As detailed in Appendix \ref{sec:appendix_distributional_properties}, $\Psi^{(\kappa)}(s)$ is approximately normally distributed with mean $\mu(s)\coloneqq  \gamma_{OS} \cdot S_{OS,0}(s) \log(S_{OS,0}(s))$ under the contiguous alternatives \eqref{eq:contiguous_alternatives_os}. In particular, if null hypothesis $H_0$ from \eqref{eq03.01} holds true, $\Psi^{(\kappa)}(s)$ has mean zero. In order to formulate the statistical test, we will need some further notation.
First, we define 
\begin{align*}
\hat{F}_{j2}^{(E \kappa)}(s)\coloneqq \int_0^s \exp(-\hat{\Lambda}_{0*}^{(E \kappa)}(u-))d\hat{\Lambda}_{j2}^{(E \kappa)}(u), \qquad \text{for } j=0,1, \text{ and } \kappa=1, \ldots, K.
\end{align*}
We also introduce the stage-wise matrices
\begin{equation*}
\hat{V}^{(E \kappa)}(s) \coloneqq 
[\hat{\Theta}^{(E \kappa)}](s) \coloneqq 
\begin{pmatrix} [\hat{\Theta}^{(E \kappa)}]_{11}(s) & 0  & [\hat{\Theta}^{(E \kappa)}]_{13}(s) & [\hat{\Theta}^{(E \kappa)}]_{14}(s) \\ 0 & [\hat{\Theta}^{(E \kappa)}]_{22}(s) & 0 & 0 \\ [\hat{\Theta}^{(E \kappa)}]_{31}(s) & 0  & [\hat{\Theta}^{(E \kappa)}]_{33}(s) & [\hat{\Theta}^{(E \kappa)}]_{34}(s)  \\    [\hat{\Theta}^{(E \kappa)}]_{41}(s) & 0  & [\hat{\Theta}^{(E \kappa)}]_{43}(s) & [\hat{\Theta}^{(E \kappa)}]_{44}(s)   
\end{pmatrix}
\end{equation*}
with non-zero components
\begin{align}\label{eq:covariance_matrix_estimate_perstage}
\begin{split}
[\hat{\Theta}^{(E \kappa)}]_{11}(s) 																						&= n^{(E \kappa)} \cdot \sum_{i \in \aleph^{(E \kappa)}} \int_0^s \hat{S}_{PFS}^{(E \kappa)}(u)^2 																																	\frac{I_{02,i}(u)}{Y_{02}^{(E \kappa)}(u)^2} dN_i^{OS}(u)  \\
[\hat{\Theta}^{(E \kappa)}]_{13}(s) \equiv [\hat{\Theta}^{(E \kappa)}]_{31}(s)  &= n^{(E \kappa)} \cdot \sum_{i \in \aleph^{(E \kappa)}} \int_0^s \hat{S}_{PFS}^{(E \kappa)}(u) 																																		\frac{I_{02,i}(u)}{Y_{02}^{(E \kappa)}(u)^2} dN_i^{OS}(u)  \\
[\hat{\Theta}^{(E \kappa)}]_{14}(s) \equiv [\hat{\Theta}^{(E \kappa)}]_{41}(s)  &= n^{(E \kappa)} \cdot \sum_{i \in \aleph^{(E \kappa)}} \int_0^s \hat{S}_{PFS}^{(E \kappa)}(u) \cdot [\hat{F}_{02}^{(E \kappa)}(u)-\hat{F}_{12}^{(E \kappa)}(u)]	\frac{I_{02,i}(u)}{Y_{02}^{(E \kappa)}(u)^2} dN_i^{OS}(u)  \\
[\hat{\Theta}^{(E \kappa)}]_{22}(s) 																						&= n^{(E \kappa)} \cdot \sum_{i \in \aleph^{(E \kappa)}} \int_0^s [{S}_{OS,0}(u)-\hat{S}_{PFS}^{(E \kappa)}(u)]^2 																									\frac{I_{12,i}(u)}{Y_{12}^{(E \kappa)}(u)^2} dN_i^{OS}(u)  \\
[\hat{\Theta}^{(E \kappa)}]_{33}(s) 																						&= n^{(E \kappa)} \cdot \sum_{i \in \aleph^{(E \kappa)}} \int_0^s                          																																				\frac{I_{02,i}(u)}{Y_{02}^{(E \kappa)}(u)^2} dN_i^{PFS}(u) \\
[\hat{\Theta}^{(E \kappa)}]_{34}(s) \equiv [\hat{\Theta}^{(E \kappa)}]_{43}(s)  &= n^{(E \kappa)} \cdot \sum_{i \in \aleph^{(E \kappa)}} \int_0^s [\hat{F}_{02}^{(E \kappa)}(u)-\hat{F}_{12}^{(E \kappa)}(u)] 	    																\frac{I_{02,i}(u)}{Y_{02}^{(E \kappa)}(u)^2} dN_i^{PFS}(u) \\
[\hat{\Theta}^{(E \kappa)}]_{44}(s) 																						&= n^{(E \kappa)} \cdot \sum_{i \in \aleph^{(E \kappa)}} \int_0^s [\hat{F}_{02}^{(E \kappa)}(u)-\hat{F}_{12}^{(E \kappa)}(u)]^2  																	\frac{I_{02,i}(u)}{Y_{02}^{(E \kappa)}(u)^2} dN_i^{PFS}(u) 
\end{split}
\end{align}
with $\hat{S}_{PFS}^{(E \kappa)}(u)\coloneqq \exp(-\hat{\Lambda}_{0*}^{(E \kappa)}(u-))\coloneqq \lim_{\epsilon \to 0}\exp(- \hat{\Lambda}_{0*}^{(E \kappa)}(s-\epsilon))$. Finally, we also need
\begin{align}\label{eq04.08}
\begin{split}
\hat{\sigma}^{(\kappa)}(s)^2					&\coloneqq \hat{L}^{(E \kappa)}(s)   \cdot \hat{V}^{(E \kappa)}(s) \cdot \hat{L}^{(E \kappa)}(s)^{\textsc{T}} \\
\hat{\varsigma}^{(\kappa)}(s_2,s_1)^2	&\coloneqq \hat{L}^{(E \kappa)}(s_2) \cdot [\hat{V}^{(E \kappa)}(s_2) - \hat{V}^{(E \kappa)}(s_1)] \cdot \hat{L}^{(E \kappa)}(s_2)^{\textsc{T}} 
\end{split}
\end{align}
with $\hat{L}^{(E \kappa)}(s)\coloneqq (1,1,\hat{F}_{12}^{(E \kappa)}(s)-\hat{F}_{02}^{(E \kappa)}(s),1)$.


\subsection{A Two-Stage Adaptive Test of $H_0$}\label{Sec04.04}

In many oncological diseases, the early occurrence of progression is associated with an unfavorable course. The occurrence of early progression can then act as a surrogate for later deaths in the long-term course. For a trial investigating long-term progression-free and overall survival, it therefore makes sense to perform an interim analysis to decide early on the continuation and, if necessary, sample size adjustment of the trial based on the observed frequency of early progressions. To this end, we propose a two-stage adaptive test of $H_0: S_{OS}^{(E)}(s) = S_{OS,0}(s)$, which allows data-dependent sample size adjustment of the trial based on the observed short-term PFS data.

In a non--adaptive setting, a test of null hypothesis $H_0$ is commonly based on the classical one--sample log--rank statistic $L_{OS}$ for OS. This is not appropriate in the adaptive setting, because $L_{OS}$ is not adapted to the joint filtration $({\cal{F}}_s)_{s \geq 0}$ of PFS- and OS-events, which in turn means that PFS-data must not be used for data-dependent design changes, contrary to what is desired (c.f. \cite{BauerPosch04}). Instead, we will use an approach based on the stage-wise statistics $\Psi^{(\kappa)}$ introduced in \eqref{eq04.04} which are adapted to $({\cal{F}}_s)_{s \geq 0}$, thus allowing simultaneous use of PFS- and OS-data for data-dependent design changes.

Let the trial start at calendar time zero. At calender time $t_1>0$, an interim analysis will be performed to determine short-term PFS and OS within the first $s_1$ months after entry (e.g. $s_1 =$ 6-months after entry). Of course, we have to require $s_1 < t_1$. Patients recruited prior to calender time $a_1 \coloneqq  t_1 - s_1$ define the set $\aleph^{(E1)}$ of stage $1$ patients, and will be part of the interim analysis. Let there be $n^{(E1)}$ stage-$1$ patients. 
At the interim analysis, we derive the interim statistics
\begin{align*}
\begin{split}
\hat{Z}_{11} & \coloneqq  \frac{  \Psi^{(1)}(s_1)  }{  \hat{\sigma}^{(1)}(s_1)   } \qquad \text{and} \qquad  \hat{\Lambda}_{0*}^{(E1)}(s_1),
\end{split}
\end{align*}
where $\hat{\sigma}^{(1)}(s_1)^2$ is an estimate of the variance of $\Psi^{(1)}(s_1)$ calculated within the set of stage $1$ patients acc. to Eq. \eqref{eq04.08}. The statistics $\hat{Z}_{11}$ is normally distributed with mean $\gamma_{OS} \cdot S_{OS,0}(s) \log(S_{OS,0}(s))$ which is proportional to the true hazard ratio w.r.t. OS (c.f. Appendix \ref{sec:appendix_distributional_properties}, Eq.\eqref{A6.eq102}). $\hat{Z}_{11}$ is thus an interim measure for the treatment effect w.r.t. OS. Analogously, $\hat{\Lambda}_{0*}^{(E 1)}(s_1)$ from Eq. \eqref{eq:na_estimators} estimates the cumulative hazard ${\Lambda}_{0*}^{(E)}(s_1)$ w.r.t. PFS under experimental treatment $E$ at time $s_1$, and thus quantifies the short-term PFS. The observed values of $\hat{Z}_{11}$ and $\hat{\Lambda}_{0*}^{(E 1)}(s_1)$ will be used to decide whether the trial stops at the interim analysis or goes on to a second stage.

In case that the trial goes on, let the recruitment phase be extended for an additional $A_2$ months, followed by a follow-up phase of $f$ months. 
Accordingly, the final analysis of the patients will be at calendar time $a_1 + A_2 + f$ which is equally the maximum in-study time $S_2 \coloneqq  a_1 + A_2 + f$ of the first patient.
These new patients recruited between calendar time $a_1$ and $a_1 + A_2$ form the set $\aleph^{(E2)}$ of stage $2$ patients.  Let there be $n^{(E 2)}$ stage-$2$ patients.  

The adaptive component of the design is that $S_2$ (and in turn $A_2$ and $n^{(E2)}$) may be chosen in dependence of $\hat{Z}_{11}$ and $\hat{\Lambda}_{0*}^{(E 1)}(s_1)$, and are thus random. From now on, we will write $S_2\coloneqq \xi(\hat{Z}_{11}, \hat{\Lambda}_{0*}^{(E 1)}(s_1) - {\Lambda}_{0*,0}(s_1) )$ instead of $s_2$ for some continuous function $\xi:R \times [0,\infty) \to (s_1,\infty)$, $(x_1,x_2) \mapsto \xi(x_1,x_2)$ to indicate this randomness. The function $\xi$ is called the \emph{adaptation rule} that we assume is established in advance of the interim analysis. An example for an adaption rule will be given at the end of this section.

At the time of final analysis, the increment statistics in stage $1$ patients
\begin{align*}
\begin{split}
\hat{Z}_{12} & \coloneqq  \frac{  \Psi^{(1)}(S_2) - \Psi^{(1)}(s_1)  }{  \hat{\varsigma}^{(1)}(S_2,s_1)   }     
\end{split}
\end{align*}
as well as the overall statistic in stage $2$ patients
\begin{align*}
\begin{split}
\hat{Z}_{2} & \coloneqq  \frac{  \Psi^{(2)}(S_2)  }{\hat{\sigma}^{(2)}(S_2)} \\
\end{split}
\end{align*}
is calculated. For arbitrary but prefixed weights $w_1$, $w_2$ with $w_1^2+w_2^2=1$, these are combined to stage--wise statistics 
\begin{align*}
\begin{split}
Z_{1} & \coloneqq  \hat{Z}_{11}, \quad \text{and} \quad  Z_{2}\coloneqq  w_1 \hat{Z}_{12} + w_2 \hat{Z}_{2}. 
\end{split}
\end{align*}
The adaptive test of null hypothesis $H_0:S_{OS}^{(E)}(s)=S_{OS,0}(s)$ is defined by the rejection region ${\cal{R}}\coloneqq  \{  Z_{1} \geq c_1\} \cup \{ c_0 \leq Z_1 < c_1, Z_2 \geq c_2 \}$ with stage-wise rejection bounds $c_1$ and $c_2$, and stopping for futility bound $c_0$. The critical bounds $c_1,c_2$ are chosen such that the stage 1 type one error rate $P_{H_0}(\{ Z_{11} \geq c_1 \})$ equals $\alpha_1$ and such that the overall type one error rate $P_{H_0}({\cal{R}})$ equals the desired significance level $\alpha>\alpha_1$. 
As shown in Appendix \ref{sec:appendix_distributional_properties}, this means choosing
\begin{align*}
\begin{split}
c_1\coloneqq  \Phi^{-1}(1- \alpha_1),
\end{split}
\end{align*}
and choosing $c_2$ such that
\begin{align}\label{eq04.15}
\begin{split}
\alpha - \alpha_1  
= 	\int_{c_0}^{c_1} dz \int_{-\infty}^{\infty} dm f(z,m) \Phi\left( -{c}_2 + w_1 \frac{[\Delta \hat{F}(\hat{s}_2(z,m)) - \Delta \hat{F}(s_1) ] \cdot m ]}{ \hat{{\varsigma}}(\hat{s}_2(z,m),s_1)} \right) 
\end{split}
\end{align}
where $m, z \in R$ are real numbers, 
\begin{align}\label{eq04.16}
\begin{split}
\hat{s}_2(z,m) \coloneqq  \xi(x_1,x_2)|_{x_1=z, x_2 = \frac{m}{\sqrt{n^{(E 1)}}} - \hat{\Lambda}_{0*}^{(E 1)}(s_1) } 
\equiv \xi\left( z, \frac{m}{\sqrt{n^{(E 1)}}} - \hat{\Lambda}_{0*}^{(E 1)}(s_1) \right)
\end{split}
\end{align}
for the prefixed adaptation rule $\xi(x_1,x_2)$, 
\begin{align*}
\begin{split}
f(z,m) \coloneqq  \frac{ \exp\left( - \frac{1}{2} (z,m) \cdot \Sigma^{-1} \cdot (z,m)^{\textsc{T}} \right) }{ 2 \pi  \sqrt{\det\Sigma } }   
\end{split}
\end{align*}
is the density of a bivariate normal distribution with mean zero and covariance matrix
\begin{align*}
\begin{split}
\Sigma \coloneqq  \begin{pmatrix} 1 & \frac{\hat{L}^{(E 1)}(s_1) \cdot \hat{V}_{\cdot 3}^{(E 1)}(s_1)}{\hat{\sigma}^{(E 1)}(s_1)^{-1}} \\ \frac{\hat{L}^{(E 1)}(s_1) \cdot \hat{V}_{\cdot 3}^{(E 1)}(s_1)}{\hat{\sigma}^{(E 1)}(s_1)^{-1}} & \hat{V}_{33}^{(E 1)}(s_1)
\end{pmatrix}
\end{split}
\end{align*}
with $\hat{V}_{\cdot 3}^{(E 1)}(s_1)$ denoting the third column of the matrix $\hat{V}^{(E 1)}(s_1)$, and with $\hat{V}_{3 3}^{(E 1)}(s_1)$ being the entry in row three and column three of $\hat{V}^{(E 1)}(s_1)$, and 
\begin{align*}
\begin{split}
\Delta \hat{F}(s)\coloneqq  \hat{F}_{12}^{(E 1)}(s) - \hat{F}_{02}^{(E 1)}(s) \qquad \text{for any } s \geq 0.
\end{split}
\end{align*}
Notice that the critical bound $c_2$ only becomes known at the interim analysis, after the adaptation rule $\xi(x_1,x_2)$ has been established and the quantities $\hat{\varsigma}(\cdot,\cdot)$, $\Sigma$, $\hat{\Lambda}_{0*}^{(E 1)}(s_1)$ and $\Delta \hat{F}(\cdot)$ become evaluable.

To give an example of a design modification, let the adaptive decision rule be such that enrollment after the interim analysis is continued for another 36 months only if (i) the observed $s_1$-months PFS in the experimental arm is at least 80\% and (ii) the treatment effect w.r.t OS is in favour of the experimental treatment (i.e. $Z_{11} \leq 0$), and otherwise if (i) or (ii) is not fulfilled accrual to the study is stopped at the interim analysis. In formulas, the modified time of final analysis would then be
\begin{align*}
S_2 \coloneqq  a_1 + 36 \cdot I\left( \hat{\Lambda}_{0*}^{(E 1)}(s_1) \leq - \log(0.8) \right) \cdot I(Z_{11} \leq 0) + f.
\end{align*}
corresponding to the adaptation rule 
\begin{align*}
\begin{split}
\xi(x_1,x_2) = a_1 + 36 \cdot I(x_2 \leq - \log(0.8) ) \cdot I(x_1 \leq 0) + f, 
\end{split}
\end{align*}
because this choice for $\xi(\cdot,\cdot)$ yields $S_2 = \xi(x_1,x_2)|_{x_1=Z_{11}, x_2 = \hat{\Lambda}_{0*}^{(E 1)}(s_1)  } \equiv \xi(Z_{11}, \hat{\Lambda}_{0*}^{(E 1)}(s_1) ) $.
Analogously, the function $\hat{s}_2(z,m)$ defined in Eq. \eqref{eq04.16} is 
\begin{align*}
\begin{split}
\hat{s}_2(z,m) &\coloneqq  \xi(x_1,x_2)|_{x_1=z, x_2 = \frac{m}{\sqrt{n^{(E 1)}}} - \hat{\Lambda}_{0*}^{(E 1)}(s_1)  }   \\
& \equiv   a_1 + 36 \cdot I\left( \frac{m}{\sqrt{n^{(E 1)}}} - \hat{\Lambda}_{0*}^{(E 1)}(s_1)  \leq - \log(0.8) \right) \cdot I(z \leq 0) + f.
\end{split}
\end{align*}
With these expressions, the integral in Eq. \eqref{eq04.15} is evaluated to determine the critical bound $c_2$. The power of the adaptive test is obtained by calculating the probability of the event ${\cal{R}}$ under the contiguous alternatives from Eq. \eqref{eq:contiguous_alternatives_os}. For this purpose we need the distribution of $\Psi^{(\kappa)}$ under the contiguous alternatives which is provided in Appendix \ref{sec:appendix_distributional_properties}.

\section{Adaptive Testing of Hypotheses on OS in Randomized Clinical Trials}\label{sec:rct_os}

\subsection{The Null Hypothesis}\label{Sec05.01}
We consider a clinical trial in which patients are recruited in successive stages $\kappa = 1, \ldots, K$ and are randomly assigned to experimental treatment $E$ or control treatment $C$ in proportion $r^{(\kappa)}\coloneqq n^{(E \kappa)}/n^{(C\kappa)}$ at each stage. Here, $n^{(x\kappa)}$ is the number of patients in treatment arm $x=E,C$ at stage $\kappa$. In this setting, we consider testing the two-sided null-hypothesis of no difference in overall survival between the treatment groups
\begin{align}\label{eq05.01}
H_0: S_{OS}^{(E)}(s) = S_{OS}^{(C)}(s). 
\end{align}
We aim for an adaptive test of $H_0$ that allows for data-dependent sample size adjustments based on all PFS and OS interim data simultaneously. On first sight, the two--sample log--rank statistic $LR_{OS}(s)$ appears as natural candidate to test null hypothesis $H_0$, as it is a mean--zero martingale when $H_0$ holds true. The martingale property of $LR_{OS}(s)$, however, only holds with respect to the filtration generated exclusively by OS-events. This is not sufficient for our purposes, as we would like to make sample size adjustments based on interim information regarding OS and PFS both. In Section \ref{Sec05.02}, we will thus propose alternative stage-wise statistics $\Delta\Psi^{(\kappa)}(s)$ that ensure the martingale property with respect to the filtration $({\cal{F}}_{s})_{s \geq 0}$ generated jointly by OS- and PFS-events.

For purposes of power calculation we additionally introduce the (stage-wise) \emph{contiguous alternatives}
\begin{align}\label{eq05.02}
H_1: S_{OS}^{(E)}(s) = [S_{OS}^{(C)}(s)]^{\omega_{OS}}
\end{align} 
with hazard ratio ${\omega_{OS}}= \exp(- {\gamma_{OS}}/\sqrt{n^{(\kappa)}})$ for some $\gamma_{OS} > 0$ and with $n^{(\kappa)}\coloneqq n^{(E \kappa)}+n^{(C\kappa)}$. Whereas ${\omega_{OS}}$ is the clinically relevant hazard ratio, the parameter $\gamma_{OS} > 0$ is the quantity that will appear in the normal approximation of our test statistic.

\subsection{The Test Statistic}\label{Sec05.02}

Starting point is the following identity for the overall survival function in treatment group $x=E,C$ that holds true under the Markov assumption (see Appendix \ref{subsec:appendix_os_identity})
\begin{align}\label{eq05.03}
S_{OS}^{(x)}(t) = e^{-\Lambda_{12}^{(x)}(t)} \cdot \left\{ 1 + \int_0^t  e^{- [\Lambda_{0*}^{(x)}(u)-\Lambda_{12}^{(x)}(u)]} d[\Lambda_{12}^{(x)}(u)-\Lambda_{02}^{(x)}(u)]     \right\}. 
\end{align}
Identity \eqref{eq05.03} motivates to introduce the stage- and group-wise stochastic processes
\begin{align}\label{eq05.04}
\Psi^{(x \kappa)}(s)\coloneqq  e^{-\hat{\Lambda}_{12}^{(x \kappa)}(s)} \cdot \left( c  +  \int_0^s  e^{-[\hat{\Lambda}_{0*}^{(x \kappa)}(u-)-\hat{\Lambda}_{12}^{(x \kappa)}(u-)]} d[\hat{\Lambda}_{12}^{(x \kappa)}(u)-\hat{\Lambda}_{02}^{(x \kappa)}(u)]    \right),
\end{align}
calculated from patients from treatment group $x=E,C$ and stage $\kappa$. These are used to define the stage-wise statistics 
\begin{align}\label{eq05.05}
\Delta\Psi^{(\kappa)}(s) \coloneqq  \Psi^{(C \kappa)}(s) - \Psi^{(E \kappa)}(s), \qquad \kappa = 1, \ldots, K,
\end{align}
for adaptively testing null hypothesis $H_0:S_{OS}^{(E)}(s) = S_{OS}^{(C)}(s)$ as detailed below in Section \ref{Sec05.04}.

\subsection{Specific Notation}\label{Sec05.03}

Notice that $\Psi^{(x \kappa)}(s)$ has the general algebraic structure 
\begin{align}\label{eq05.06}
\Psi^{(x \kappa)}(s)\coloneqq  e^{-\hat{\Lambda}_3^{(x \kappa)}(s)} \cdot \left( c  +  \int_0^s  e^{-\hat{\Lambda}_1^{(x \kappa)}(u-)} d\hat{\Lambda}_2^{(x \kappa)}(u)    \right),
\end{align}
with parameter values $c\coloneqq 1$, and
\begin{align*}
\begin{split}
\hat{\Lambda}_1^{(x \kappa)}(s) &\coloneqq  \hat{\Lambda}_{0*}^{(x \kappa)}(s) - \hat{\Lambda}_{12}^{(x\kappa)}(s), \\
\hat{\Lambda}_2^{(x \kappa)}(s) &\coloneqq  \hat{\Lambda}_{12}^{(x\kappa)}(s)  - \hat{\Lambda}_{02}^{(x\kappa)}(s), \\
\hat{\Lambda}_3^{(x \kappa)}(s) &\coloneqq  \hat{\Lambda}_{12}^{(x\kappa)}(s),
\end{split}
\end{align*}
whose distribution properties have been studied in detail in Appendix \ref{sec:appendix_distributional_properties}.
For the implementation of the statistical test, we will need some further notation. First, let
\begin{align*}
\hat{F}^{(x \kappa)}(s) &\coloneqq  \int_0^s  e^{-\hat{\Lambda}_1^{(x \kappa)}(u-)} d\hat{\Lambda}_2^{(x \kappa)}(u).
\end{align*}
Then, we will also need the pairwise covariation structure of the compensating $({\cal{F}}_{s})_{s \geq 0}$-martingales
\begin{align}\label{eq05.08}
\begin{split}
M_1^{(x \kappa)}(s) &\coloneqq  M_{0*}^{(x \kappa)}(s) - M_{12}^{(x \kappa)}(s), \\
M_2^{(x \kappa)}(s) &\coloneqq  M_{12}^{(x \kappa)}(s) - M_{02}^{(x \kappa)}(s), \\
M_3^{(x \kappa)}(s) &\coloneqq  M_{12}^{(x \kappa)}(s).
\end{split}
\end{align}
Using Eq. \eqref{eq:na_martingale_covariation}, this covariance structure reads as follows in differential form:
\begin{align*}
\begin{split}
d[M_1^{(x \kappa)}](s) 									   &=   n^{(x \kappa)} \sum_{i \in \aleph^{(x \kappa)}} \left(   \frac{I_{02,i}(s)}{Y_{02}^{(x \kappa)}(s)^2} dN_i^{PFS}(s) +   \frac{I_{12,i}(s)}{Y_{12}^{(x \kappa)}(s)^2} dN_i^{OS}(s)   \right) \\
d[M_1^{(x \kappa)},M_2^{(x \kappa)}](s)    &= - n^{(x \kappa)} \sum_{i \in \aleph^{(x \kappa)}} \left(   \frac{I_{12,i}(s)}{Y_{12}^{(x \kappa)}(s)^2} dN_i^{OS}(s)  +   \frac{I_{02,i}(s)}{Y_{02}^{(x \kappa)}(s)^2} dN_i^{OS}(s)   \right) \\
d[M_1^{(x \kappa)},M_3^{(x \kappa)}](s)    &= - n^{(x \kappa)} \sum_{i \in \aleph^{(x \kappa)}} \left(   \frac{I_{12,i}(s)}{Y_{12}^{(x \kappa)}(s)^2} dN_i^{OS}(s)      \right) \\
d[M_1^{(x \kappa)},M_{0*}^{(x \kappa)}](s) &=   n^{(x \kappa)} \sum_{i \in \aleph^{(x \kappa)}} \left(   \frac{I_{02,i}(s)}{Y_{02}^{(x \kappa)}(s)^2} dN_i^{PFS}(s)     \right) \\
d[M_2^{(x \kappa)}](s)    						     &=   n^{(x \kappa)} \sum_{i \in \aleph^{(x \kappa)}} \left(   \frac{I_{12,i}(s)}{Y_{12}^{(x \kappa)}(s)^2} dN_i^{OS}(s)  +   \frac{I_{02,i}(s)}{Y_{02}^{(x \kappa)}(s)^2} dN_i^{OS}(s)   \right) \\
d[M_2^{(x \kappa)},M_3^{(x \kappa)}](s)    &=   n^{(x \kappa)} \sum_{i \in \aleph^{(x \kappa)}} \left(   \frac{I_{12,i}(s)}{Y_{12}^{(x \kappa)}(s)^2} dN_i^{OS}(s)      \right) \\
d[M_2^{(x \kappa)},M_{0*}^{(x \kappa)}](s) &= - n^{(x \kappa)} \sum_{i \in \aleph^{(x \kappa)}} \left(   \frac{I_{02,i}(s)}{Y_{02}^{(x \kappa)}(s)^2} dN_i^{OS}(s)     \right) \\
d[M_3^{(x \kappa)}](s)    								 &=   n^{(x \kappa)} \sum_{i \in \aleph^{(x \kappa)}} \left(   \frac{I_{12,i}(s)}{Y_{12}^{(x \kappa)}(s)^2} dN_i^{OS}(s)      \right) \\
d[M_3^{(x \kappa)},M_{0*}^{(x \kappa)}](s) &=   0 \\
d[M_{0*}^{(x \kappa)}](s) 								 &=   n^{(x \kappa)} \sum_{i \in \aleph^{(x \kappa)}} \left(   \frac{I_{02,i}(s)}{Y_{02}^{(x \kappa)}(s)^2} dN_i^{PFS}(s)     \right) 
\end{split}
\end{align*}
On this basis, we finally introduce the $4 \times 4$ matrix $\hat{V}^{(x \kappa)}(s) \coloneqq  \left( \hat{[\Theta]}_{ij}(s) \right)_{i,j=1, \ldots, 4}$ with components
\begin{align}\label{eq05.11}
\begin{split}
\hat{[\Theta]}_{11}(s) 																&= \int_0^s \left( \hat{F}^{(x \kappa)}(u) \right)^2                                      d[M_1^{(x \kappa)}](u)  \\
\hat{[\Theta]}_{12}(s) \equiv \hat{[\Theta]}_{21}(s)  &= \int_0^s        \hat{F}^{(x \kappa)}(u) \cdot e^{- \hat{\Lambda}_1^{(x \kappa)}(u-)}   d[M_1^{(x \kappa)},M_2^{(x \kappa)}](u)  \\
\hat{[\Theta]}_{13}(s) \equiv \hat{[\Theta]}_{31}(s)  &= \int_0^s        \hat{F}^{(x \kappa)}(u)                                                d[M_1^{(x \kappa)},M_{0*}^{(x \kappa)}](u)  \\
\hat{[\Theta]}_{14}(s) \equiv \hat{[\Theta]}_{41}(s)  &= \int_0^s        \hat{F}^{(x \kappa)}(u)                                                d[M_1^{(x \kappa)},M_{3}^{(x \kappa)}](u)  \\
\hat{[\Theta]}_{22}(s) 																&= \int_0^s        e^{-2 \hat{\Lambda}_1^{(x \kappa)}(u-)}   															 d[M_2^{(x \kappa)}](u) \\
\hat{[\Theta]}_{23}(s) \equiv \hat{[\Theta]}_{32}(s)  &= \int_0^s        e^{- \hat{\Lambda}_1^{(x \kappa)}(u-)}                                  d[M_2^{(x \kappa)},M_{0*}^{(x \kappa)}](u)  \\
\hat{[\Theta]}_{24}(s) \equiv \hat{[\Theta]}_{42}(s)  &= \int_0^s        e^{- \hat{\Lambda}_1^{(x \kappa)}(u-)}                                  d[M_2^{(x \kappa)},M_{3}^{(x \kappa)}](u)  \\
\hat{[\Theta]}_{33}(s) 																&= [M_{0*}^{(x \kappa)}](s) \\
\hat{[\Theta]}_{34}(s) \equiv \hat{[\Theta]}_{43}(s)  &= 0 \\
\hat{[\Theta]}_{44}(s) 																&= [M_3^{(x \kappa)}](s),
\end{split}
\end{align}
as well as further derived quantities
\begin{align}\label{eq:covariance_estimate_two_arm_os}
\begin{split}
\hat{\mu}^{(x \kappa)}(s) &\coloneqq  e^{-\hat{\Lambda}_3^{(x \kappa)}(u)}[\hat{F}^{(x \kappa)}(s) + c], \\
\hat{L}^{(x \kappa)}(s)   &\coloneqq  e^{-\hat{\Lambda}_3^{(x \kappa)}(u)} \cdot (1,1,-\hat{F}^{(x \kappa)}(s),-c) \\
\hat{\cal{L}}^{(C \kappa)}(s) &\coloneqq  \sqrt{ 1 + r^{(\kappa)}  } \cdot \hat{L}^{(C \kappa)}(s), \\
\hat{\cal{L}}^{(E \kappa)}(s) &\coloneqq  \sqrt{ 1 + 1/r^{(\kappa)}  } \cdot \hat{L}^{(E \kappa)}(s), \\
\hat{\sigma}^{(\kappa)}(s)^2 &\coloneqq \sum_{x=C,E} \hat{\cal{L}}^{(x \kappa)}(s) \cdot \hat{V}^{(x \kappa)}(s) \cdot \hat{\cal{L}}^{(x \kappa)}(s)^{\textsc{T}}, \\
\hat{\varsigma}^{(\kappa)}(s_2,s_1)^2 &\coloneqq \sum_{x=C,E} \hat{\cal{L}}^{(x \kappa)}(s_2) \cdot \left( \hat{V}^{(x \kappa)}(s_2) - \hat{V}^{(x \kappa)}(s_1) \right) \cdot \hat{\cal{L}}^{(x \kappa)}(s_2)^{\textsc{T}}.
\end{split}
\end{align}

\subsection{A Two-Stage Adaptive Test of $H_0$}\label{Sec05.04}

We propose a two-stage adaptive test of $H_0: S_{OS}^{(E)}(s) = S_{OS}^{(C)}(s)$ from Eq. \eqref{eq05.01}, which allows data-dependent sample size adjustment of the trial based on the observed short-term PFS data. In a non--adaptive setting, a test of null hypothesis $H_0$ is commonly based on the classical two--sample log--rank statistic $L_{OS}$ for OS. This is not appropriate in the adaptive setting, because $L_{OS}$ is not adapted to the joint filtration $({\cal{F}}_s)_{s \geq 0}$ of PFS and OS events, which in turn means that PFS-data must not be used for data-dependent design changes, contrary to what is desired (c.f. \cite{BauerPosch04}). Instead, we will use an approach based on the stage-wise statistics $\Delta \Psi^{(\kappa)}$ introduced in \eqref{eq05.05} which are adapted to $({\cal{F}}_s)_{s \geq 0}$, thus allowing simultaneous use of PFS- and OS-data for data-dependent design changes.

Let the trial start at calendar time zero. At calender time $t_1>0$, say, an interim analysis is performed to determine the short-term PFS and OS within the first $s_1$ months since start of therapy (e.g. think of $s_1 =$ 6-months after start of treatment). Of course, we have to require $s_1 < t_1$. Patients recruited prior to calender time $a_1 \coloneqq  t_1 - s_1$ define the set $\aleph^{(1)}$ of stage $1$ patients, and will be part of the interim analysis. Assume there are $n^{(1)} \coloneqq  n^{(E1)} + n^{(C1)}$ stage-1 patients, with $n^{(E1)}$ and $n^{(C1)}$ of them being randomly allocated to experimental and control treatment with stage-1 randomization ratio $r^{(1)}\coloneqq n^{(E1)}/n^{(C1)}$, respectively. At the interim analysis, we derive the interim statistics
\begin{align*}
\begin{split}
\hat{Z}_{11} & \coloneqq  \frac{ \Delta \Psi^{(1)}(s_1)  }{  \hat{\sigma}^{(1)}(s_1)   }, \quad \hat{\Lambda}_{0*}^{(E1)}(s_1), \quad \text{and} \quad  \hat{\Lambda}_{0*}^{(C1)}(s_1).
\end{split}
\end{align*}
Notice that $\hat{\sigma}^{(1)}(s_1)^2$ is an estimate of the variance of $\Psi^{(1)}(s_1)$ calculated within the set of stage $1$ patients. The statistics $\hat{Z}_{11}$ is normally distributed with mean $\gamma_{OS} \cdot \mu^{(C)}(s_1) \log(\mu^{(C)}(s_1))$, which is proportional to the true hazard ratio $\gamma_{OS}$ w.r.t. OS and otherwise only depends on some control group baseline parameter $\mu^{(C)}(s_1)\coloneqq  e^{- \Lambda_3^{(C)}(s_1)} [1 + \int_0^{s_1} e^{- \Lambda_1^{(C)}(u)} d\Lambda_2^{(C)}(u) ]$ (c.f. Appendix \ref{sec:appendix_distributional_properties}, Eq. \eqref{A8.Eq07bb}). $\hat{Z}_{11}$ thus is an interim measure for the treatment effect w.r.t. OS. Analogously, $\hat{\Lambda}_{0*}^{(x 1)}(s_1)$ estimates the cumulative hazard ${\Lambda}_{0*}^{(x)}(s_1)$ w.r.t. PFS in treatment group $x=E,C$ at time $s_1$, and thus quantifies the treatment group specific short-term PFS. The observed values of $\hat{Z}_{11}$ and $\hat{\Lambda}_{0*}^{(x 1)}(s_1)$ for $x=E,C$ will be used to decide whether the trial stops at the interim analysis or goes on to a second stage.

In case that the trial goes on, let the recruitment phase be extended for an additional $A_2$ months, followed by a follow-up phase of $f$ months. 
Accordingly, final analysis of all patients will be at calendar time $a_1 + A_2 + f$ which is equally the maximum in-study time $S_2 \coloneqq  a_1 + A_2 + f$ of the first patient.
These new patients recruited between calendar time $a_1$ and $a_1 + A_2$ form the set $\aleph^{(E2)}$ of stage $2$ patients. Assume there are $n^{(2)} \coloneqq  n^{(E2)} + n^{(C2)}$ stage-2 patients, with $n^{(E2)}$ and $n^{(C2)}$ of them being randomly allocated to experimental and control treatment with stage-2 randomization ratio $r^{(2)}\coloneqq n^{(E2)}/n^{(C2)}$, respectively. 

The adaptive component of the design is that $S_2$ (and in turn $A_2$ and $n^{(E2)}$) may be chosen in dependence of $\hat{Z}_{11}$ and $\hat{\Lambda}_{0*}^{(C1)}(s_1)-\hat{\Lambda}_{0*}^{(E1)}(s_1)$, and are thus random. From now on, we will write $S_2\coloneqq \xi(\hat{Z}_{11}, \hat{\Lambda}_{0*}^{(C 1)}(s_1) - \hat{\Lambda}_{0*}^{(E 1)}(s_1))$ instead of $s_2$ for some continuous function $\xi: \mathbb{R} \times \mathbb{R} \to (s_1,\infty)$, $(x_1,x_2) \mapsto \xi(x_1,x_2)$ to indicate this randomness, where $R$ denotes the set of real number. The function $\xi$ is called the \emph{adaptation rule} that we assume is established in advance of the interim analysis. An example for an adaptation rule will be given at the end of the section.

At the time of final analysis, the increment statistics in stage $1$ patients
\begin{align*}
\begin{split}
\hat{Z}_{12} & \coloneqq  \frac{  \Delta \Psi^{(1)}(S_2) - \Delta \Psi^{(1)}(s_1)  }{  \hat{{\varsigma}}^{(1)}(S_2,s_1)   }     
\end{split}
\end{align*}
as well as the overall statistic in stage $2$ patients
\begin{align*}
\begin{split}
\hat{Z}_{2} & \coloneqq  \frac{  \Delta \Psi^{(2)}(S_2)  }{\hat{\sigma}^{(2)}(S_2)} \\
\end{split}
\end{align*}
is calculated. For arbitrary but prefixed weights $w_1$, $w_2$ with $w_1^2+w_2^2=1$, these are combined to stage--wise statistics 
\begin{align*}
\begin{split}
Z_{1} & \coloneqq  \hat{Z}_{11}, \quad \text{and} \quad  Z_{2}\coloneqq  w_1 \hat{Z}_{12} + w_2 \hat{Z}_{2}. 
\end{split}
\end{align*}
The adaptive test of null hypothesis $H_0:S_{OS}^{(E)}(s)=S_{OS}^{(C)}(s)$ is defined by the rejection region ${\cal{R}}\coloneqq  \{  Z_{1} \geq c_1\} \cup \{ c_0 \leq Z_1 < c_1, Z_2 \geq c_2 \}$ with stage-wise rejection bounds $c_1$ and $c_2$, and stopping for futility bound $c_0$. The critical bounds $c_1,c_2$ are chosen such that the stage 1 type one error rate $P_{H_0}(\{ Z_{1} \geq c_1 \})$ equals $\alpha_1$, and such that the overall type one error rate $P_{H_0}({\cal{R}})$ equals the desired significance level $\alpha>\alpha_1$. 
As proven in Appendix \ref{sec:appendix_distributional_properties}, this means choosing
\begin{align*}
\begin{split}
c_1\coloneqq  \Phi^{-1}(1- \alpha_1),
\end{split}
\end{align*}
and choosing $c_2$ such that
\begin{align}\label{eq05.46}
\begin{split}
&\alpha - \alpha_1  
= 	\int_{c_0}^{c_1} dz \int_{\mathbb{R}^4} d\theta_E \int_{\mathbb{R}^4} d\theta_C  f(z,\theta_E,\theta_C) \\ 
& \qquad \qquad		\Phi\left( -c_2 + w_1 \frac{   \sum_{x=C,E} (-1)^{I(x=E)} \cdot \theta_x \cdot \left( \hat{\cal{L}}^{(x \kappa)}(\hat{s}_2(z, \theta_E, \theta_C)) - \hat{\cal{L}}^{(x \kappa)}(s_1) \right)    }{  \hat{\varsigma}^{(1)}(\hat{s}_2(z,\theta_E, \theta_C),s_1) }  \right),
\end{split}
\end{align}
where $z \in R$ is a real number, $\theta_E,\theta_C \in \mathbb{R}^4$ are four-dimensional real numbers,
\begin{align*}
\begin{split}
\hat{s}_2(z,\theta_E,\theta_C) \coloneqq   
\xi(x_1,x_2)|_{x_1=z, x_2 = \frac{ \text{pr}_3(\theta_C) }{\sqrt{n^{(C 1)}}}  -  \frac{ \text{pr}_3(\theta_E) }{\sqrt{n^{(E 1)}}}  + \hat{\Lambda}_{0*}^{(C1)}(s_1) - \hat{\Lambda}_{0*}^{(E1)}(s_1) }
\end{split}
\end{align*} 
with  $\xi(\cdot,\cdot)$ being the prespecified adaptation rule and $\text{pr}_3: \mathbb{R}^4 \to \mathbb{R}, (x_1,x_2,x_3,x_4)^t \mapsto x_3$ denoting projection onto the third component, 
and where
\begin{align*}
\begin{split}
f(z,\theta_E,\theta_C) \coloneqq  \frac{   \exp\left( - \frac{1}{2}  (z,\theta_E,\theta_C) \cdot  \Sigma^{-1} \cdot (z,\theta_E,\theta_C)^t    \right)   }{  \sqrt{(2 \pi)^{9} \cdot \det(\Sigma)}  } 
\end{split}
\end{align*}
is the density of a nine-dimensional multivariate normal distribution with mean zero and covariance matrix
\begin{align*}
\begin{split}
\Sigma \coloneqq   
  \begin{pmatrix} 
	1 																													   &  \frac{\hat{\cal{L}}^{(E 1)}(s_1) \cdot \hat{V}^{(E 1)}(s_1)}{\hat{\sigma}^{(E 1)}(s_1)}  
																																 &  \frac{\hat{\cal{L}}^{(C 1)}(s_1) \cdot \hat{V}^{(C 1)}(s_1)}{\hat{\sigma}^{(C 1)}(s_1)}     \\
	\frac{ \hat{V}^{(E 1)}(s_1) \cdot \hat{\cal{L}}^{(E 1)}(s_1)^{\textsc{T}} }{\hat{\sigma}^{(E 1)}(s_1)} 	 &                                  \hat{V}^{(E 1)}(s_1)       &  0_{4 \times 4}   \\
	\frac{ \hat{V}^{(C 1)}(s_1) \cdot \hat{\cal{L}}^{(C 1)}(s_1)^{\textsc{T}} }{\hat{\sigma}^{(C 1)}(s_1)}    &                                  0_{4 \times 4}             &  \hat{V}^{(C 1)}(s_1)   \\
  \end{pmatrix} .
\end{split}
\end{align*}
Here $ 0_{4 \times 4}$ is a  $4 \times 4$ matrix with zero entries, only. Notice that the critical bound $c_2$ only becomes known at the interim analysis, after the adaptation rule $\xi(x_1,x_2)$ has been established and the quantities $\hat{\varsigma}^{(1)}(\cdot,\cdot)$, $\Sigma$, $\hat{\Lambda}_{0*}^{(x 1)}(s_1)$ and ${\cal{L}}^{(x 1)}(\cdot)$ become evaluable.

To give an example of a design modification, let the adaptive decision rule be such that enrollment after the interim analysis is continued for another 36 months only if (i) the $s_1$-months PFS in the experimental treatment group is at least $q$\% higher then in the control group (i.e. $S^{(E)}(s_1) \geq (1 + q) \cdot S^{(C)}(s_1)$) and (ii) the treatment effect w.r.t OS is in favour of the experimental treatment (i.e. $Z_{11} \leq 0$), and otherwise if (i) or (ii) is not fulfilled accrual to the study is stopped at the interim analysis. In formulas, the modified time of final analysis would then be
\begin{align*}
\begin{split}
S_2 \coloneqq  a_1 + 36 \cdot I\left( \hat{\Lambda}_{0*}^{(C 1)}(s_1) - \hat{\Lambda}_{0*}^{(E 1)}(s_1) \geq \log(1+q) \right) \cdot I(Z_{11} \leq 0) + f
\end{split}
\end{align*}
corresponding to the adaptation rule 
\begin{align*}
\begin{split}
\xi(x_1,x_2) = a_1 + 36 \cdot I(x_2 \geq \log(1+q) ) \cdot I(x_1 \leq 0) + f, 
\end{split}
\end{align*}
because this choice for $\xi(\cdot,\cdot)$ yields $S_2 = \xi(x_1,x_2)|_{x_1=Z_{11}, x_2 = \hat{\Lambda}_{0*}^{(C 1)}(s_1) - \hat{\Lambda}_{0*}^{(E 1)}(s_1)  }$. 
Analogously, the function $\hat{s}_2(z,m)$ defined in Eq. \eqref{eq04.16} is 
\begin{align*}
\begin{split}
\hat{s}_2(z,m) \coloneqq & \xi(x_1,x_2)|_{x_1=z, x_2 = \frac{ \text{pr}_3(\theta_C) }{\sqrt{n^{(C 1)}}}  -  \frac{ \text{pr}_3(\theta_E) }{\sqrt{n^{(E 1)}}} + \hat{\Lambda}_{0*}^{(C 1)}(s_1) - \hat{\Lambda}_{0*}^{(E 1)}(s_1) }   \\
= & a_1 + 36 \cdot I\left( \frac{ \text{pr}_3(\theta_C) }{\sqrt{n^{(C 1)}}}  -  \frac{ \text{pr}_3(\theta_E) }{\sqrt{n^{(E 1)}}} + \hat{\Lambda}_{0*}^{(C 1)}(s_1) - \hat{\Lambda}_{0*}^{(E 1)}(s_1) \geq \log(1+q) \right) \cdot I(z \leq 0) + f.
\end{split}
\end{align*}
With these expressions, the integral in Eq. \eqref{eq05.46} is evaluated to determine the critical bound $c_2$.\\
The power of the adaptive test is obtained by calculating the probability of the event ${\cal{R}}$ under the contiguous alternatives from Eq. \eqref{eq05.02}. For this purpose we need the distribution of $\Delta \Psi^{(\kappa)}$ under the contiguous alternatives which is provided in Appendix \ref{sec:appendix_distributional_properties}.

\section{Adaptive Testing of Hypotheses on the Gap between PFS and OS in Randomized Clinical Trials}\label{sec:rct_pfs_os_gap}

\subsection{The Null Hypothesis}\label{Sec06.01}
We consider a clinical trial in which patients are recruited in successive stages $\kappa = 1, \ldots, K$ and are randomly assigned to experimental treatment $E$ or control treatment $C$ in proportion $r^{(\kappa)}\coloneqq n^{(E \kappa)}/n^{(C\kappa)}$ at each stage. Here, $n^{(x\kappa)}$ is the number of patients in treatment arm $x=E,C$ at stage $\kappa$. In this setting, we consider testing the two-sided null-hypothesis of no difference in the gap between overall survival (OS) and progression-free survival (PFS)
\begin{align*}
H_0: S_{OS}^{(E)}(s) - S_{PFS}^{(E)}(s) = S_{OS}^{(C)}(s) - S_{PFS}^{(C)}(s). 
\end{align*}
In Section \ref{Sec06.02}, we propose stage-wise statistics $\Delta\Psi^{(\kappa)}(s)$ that ensure the martingale property with respect to the filtration $({\cal{F}}_{s})_{s \geq 0}$ generated jointly by OS- and PFS-events, and thus are suitable for the desired adaptive test.

For purposes of power calculation we additionally introduce the (stage-wise) \emph{contiguous alternatives}
\begin{align*}
H_1: S_{OS}^{(E)}(s) - S_{PFS}^{(E)}(s) = [S_{OS}^{(C)}(s) - S_{PFS}^{(C)}(s)]^{\omega}
\end{align*} 
with hazard ratio ${\omega}= \exp(- {\gamma}/\sqrt{n^{(\kappa)}})$ for some $\gamma > 0$ and with $n^{(k)}\coloneqq n^{(E \kappa)}+n^{(C\kappa)}$. Whereas ${\omega}$ is the clinically relevant hazard ratio, the parameter $\gamma > 0$ is the quantity that will appear in the normal approximation of our test statistic.

\subsection{The Test Statistic}\label{Sec06.02}

Starting point is the following identity for the overall survival function in treatment group $x=E,C$ that holds true under the Markov assumption (see Appendix \ref{subsec:appendix_os_identity})
\begin{align*}
S_{OS}^{(x)}(t) = e^{-\Lambda_{12}^{(x)}(t)} \cdot \int_0^t  e^{- [\Lambda_{0*}^{(x)}(u)-\Lambda_{12}^{(x)}(u)]} d\Lambda_{01}^{(x)}(u).    
\end{align*}
Identity \eqref{eq05.03} motivates to introduce the stage- and treatment-wise stochastic processes
\begin{align}\label{eq06.04}
\Psi^{(x \kappa)}(s)\coloneqq  e^{-\hat{\Lambda}_{12}^{(x \kappa)}(s)} \cdot \int_0^s  e^{-[\hat{\Lambda}_{0*}^{(x \kappa)}(u-)-\hat{\Lambda}_{12}^{(x \kappa)}(u-)]} d\hat{\Lambda}_{01}^{(x \kappa)}(u),
\end{align}
calculated from patients from treatment group $x=E,C$ and stage $\kappa$. These are used to define the stage-wise statistics 
\begin{align}\label{eq06.05}
\Delta\Psi^{(\kappa)}(s) \coloneqq  \Psi^{(C \kappa)}(s) - \Psi^{(E \kappa)}(s), \qquad \kappa = 1, \ldots, K,
\end{align}
for adaptively testing null hypothesis $H_0:S_{OS}^{(E)}(s) = S_{OS}^{(C)}(s)$ as detailed below in Section \ref{Sec06.04}.

\subsection{Specific Notation}\label{Sec06.03}

Notice that $\Psi^{(x \kappa)}(s)$ has the general algebraic structure 
\begin{align*}
\Psi^{(x \kappa)}(s)\coloneqq  e^{-\hat{\Lambda}_3^{(x \kappa)}(s)} \cdot \left( c  +  \int_0^s  e^{-\hat{\Lambda}_1^{(x \kappa)}(u-)} d\hat{\Lambda}_2^{(x \kappa)}(u)    \right),
\end{align*}
with parameter values $c\coloneqq 0$, and
\begin{align*}
\hat{\Lambda}_1^{(x \kappa)}(s) &\coloneqq  \hat{\Lambda}_{0*}^{(x \kappa)}(s) - \hat{\Lambda}_{12}^{(x\kappa)}(s) \\
\hat{\Lambda}_2^{(x \kappa)}(s) &\coloneqq  \hat{\Lambda}_{0*}^{(x \kappa)}(s) - \hat{\Lambda}_{02}^{(x\kappa)}(s) \\
\hat{\Lambda}_3^{(x \kappa)}(s) &\coloneqq  \hat{\Lambda}_{12}^{(x \kappa)}(s),
\end{align*}
whose distribution properties have been studied in detail in Appendix \ref{sec:appendix_distributional_properties}.
For the implementation of the statistical test, we will need some further notation. First, let
\begin{align}\label{eq06.07b}
\hat{F}^{(x \kappa)}(s) &\coloneqq  \int_0^s  e^{-\hat{\Lambda}_1^{(x \kappa)}(u-)} d\hat{\Lambda}_2^{(x \kappa)}(u).
\end{align}
Then, we will also need the pairwise covariation structure of the compensating $({\cal{F}}_{s})_{s \geq 0}$-martingales
\begin{align}\label{eq06.08}
\begin{split}
M_1^{(x \kappa)}(s) &\coloneqq  M_{0*}^{(x \kappa)}(s) - M_{12}^{(x \kappa)}(s) \\
M_2^{(x \kappa)}(s) &\coloneqq  M_{0*}^{(x \kappa)}(s) - M_{02}^{(x \kappa)}(s) \\
M_3^{(x \kappa)}(s) &\coloneqq  M_{12}^{(x \kappa)}(s)
\end{split}
\end{align}
Recalling Eq. \eqref{eq:na_martingale_covariation}, this covariance structure reads as follows in differential form: 
\begin{align}\label{eq06.09}
\begin{split}
d[M_1^{(x \kappa)}](s) 									   &=   n^{(x \kappa)} \sum_{i \in \aleph^{(x \kappa)}} \left(   \frac{I_{02,i}(s)}{Y_{02}^{(x \kappa)}(s)^2} dN_i^{PFS}(s) +   \frac{I_{12,i}(s)}{Y_{12}^{(x \kappa)}(s)^2} dN_i^{OS}(s)   \right) \\
d[M_1^{(x \kappa)},M_2^{(x \kappa)}](s)    &=   n^{(x \kappa)} \sum_{i \in \aleph^{(x \kappa)}} \left(   \frac{I_{02,i}(s)}{Y_{02}^{(x \kappa)}(s)^2} dN_i^{PFS}(s)  -   \frac{I_{02,i}(s)}{Y_{02}^{(x \kappa)}(s)^2} dN_i^{OS}(s)   \right) \\
d[M_1^{(x \kappa)},M_3^{(x \kappa)}](s)    &= - n^{(x \kappa)} \sum_{i \in \aleph^{(x \kappa)}} \left(   \frac{I_{12,i}(s)}{Y_{12}^{(x \kappa)}(s)^2} dN_i^{OS}(s)      \right) \\
d[M_1^{(x \kappa)},M_{0*}^{(x \kappa)}](s) &=   n^{(x \kappa)} \sum_{i \in \aleph^{(x \kappa)}} \left(   \frac{I_{02,i}(s)}{Y_{02}^{(x \kappa)}(s)^2} dN_i^{PFS}(s)     \right) \\
d[M_2^{(x \kappa)}](s)    						     &=   n^{(x \kappa)} \sum_{i \in \aleph^{(x \kappa)}} \left(   \frac{I_{02,i}(s)}{Y_{02}^{(x \kappa)}(s)^2} dN_i^{PFS}(s)  -   \frac{I_{02,i}(s)}{Y_{02}^{(x \kappa)}(s)^2} dN_i^{OS}(s)   \right) \\
d[M_2^{(x \kappa)},M_3^{(x \kappa)}](s)    &=   0 \\
d[M_2^{(x \kappa)},M_{0*}^{(x \kappa)}](s) &=   n^{(x \kappa)} \sum_{i \in \aleph^{(x \kappa)}} \left(   \frac{I_{02,i}(s)}{Y_{02}^{(x \kappa)}(s)^2} dN_i^{PFS}(s)  -   \frac{I_{02,i}(s)}{Y_{02}^{(x \kappa)}(s)^2} dN_i^{OS}(s)   \right) \\
d[M_3^{(x \kappa)}](s)    								 &=   n^{(x \kappa)} \sum_{i \in \aleph^{(x \kappa)}} \left(   \frac{I_{12,i}(s)}{Y_{12}^{(x \kappa)}(s)^2} dN_i^{OS}(s)      \right) \\
d[M_3^{(x \kappa)},M_{0*}^{(x \kappa)}](s) &=   0 \\
d[M_{0*}^{(x \kappa)}](s) 								 &=   n^{(x \kappa)} \sum_{i \in \aleph^{(x \kappa)}} \left(   \frac{I_{02,i}(s)}{Y_{02}^{(x \kappa)}(s)^2} dN_i^{PFS}(s)     \right) 
\end{split}
\end{align}
On this basis, we finally define the $4 \times 4$ matrix $\hat{V}^{(x \kappa)}(s) \coloneqq  \left( \hat{[\Theta]}_{ij}(s) \right)_{i,j=1, \ldots, 4}$ as in Eq. \eqref{eq05.11}, but with $\hat{F}^{(x \kappa)}(s)$ from Eq. \eqref{eq06.07b} and the pairwise covariation processes $d[M_i^{(x \kappa)},M_i^{(x \kappa)}](s)$ according to Eq. \eqref{eq06.09}. Analogously, we also introduce the derived quantities $\hat{\mu}^{(x \kappa)}(s)$, $\hat{L}^{(x \kappa)}(s)$, $\hat{\cal{L}}^{(C \kappa)}(s)$, $\hat{\cal{L}}^{(E \kappa)}(s)$, $\hat{\sigma}^{(\kappa)}(s)^2$, and $\hat{\varsigma}^{(\kappa)}(s_2,s_1)^2$ using the same definitions as in Eq. \eqref{eq:covariance_estimate_two_arm_os}.

\subsection{A Two-Stage Adaptive Test of $H_0: S_{OS}^{(E)}(s) - S_{PFS}^{(E)}(s) = S_{OS}^{(C)}(s) - S_{PFS}^{(C)}(s)$}\label{Sec06.04}

We propose a two-stage adaptive test of null hypothesis $H_0$ which allows data-dependent sample size modifications of the trial based on the observed short-term PFS. Notice that the stage-wise statistics $\Delta \Psi^{(\kappa)}$ from Eq. \eqref{eq05.05} for testing $H_0: S_{OS}^{(E)}(s) = S_{OS}^{(C)}(s)$ and the stage-wise statistics $\Delta \Psi^{(\kappa)}$ from Eq. \eqref{eq06.05} for testing $S_{OS}^{(E)}(s) - S_{PFS}^{(E)}(s) = S_{OS}^{(C)}(s) - S_{PFS}^{(C)}(s)$ both have the same algebraic structure given in Eq. \eqref{eq05.06}. The desired adaptive test of $H_0: S_{OS}^{(E)}(s) - S_{PFS}^{(E)}(s) = S_{OS}^{(C)}(s) - S_{PFS}^{(C)}(s)$ can thus be implemented as described in Chapter \ref{Sec05.04}, by replacing all expressions from Chapter \ref{Sec05.04} with the expressions of the same name from Chapters \ref{Sec06.02} and \ref{Sec06.03}.

\section{Adaptive Hypothesis Testing In General Markovian Multi-State Models}\label{Sec07}

A multi-state model is a stochastic process $(X(s))_{s \geq 0}$ with values in a discrete state space ${\cal{S}}\coloneqq \{0,1,\ldots,D\}\}$ for some natural number $D>0$. Here $0$ is the initial state, and $1, \ldots, D$ are further transient or absorbing states. The probabilities of transition between states
\begin{align*}
	P_{gh}(s,t)\coloneqq  \mathbb{P}\left( X(t)=h | X(s)=g, {\cal{F}}_{s-} \right), \quad s<t,
\end{align*}
are quantities of key interest. Here, the $\sigma$-algebra ${\cal{F}}_{s-}$ is the information about the transitions of the process $X$ prior to time $s$. $P_{gh}(s,t)$ is the $gh$-element of the transition probability matrix $\mathbf{P}(s,t)$ that shows the probabilites transitions between the states in the time interval from $s$ to $t$. We additionally introduce the transition intensities 
\begin{align*}
	\begin{split}
	\alpha_{gh}(t)&\coloneqq  \lim_{\Delta t \to 0} \frac{P_{gh}(t,t+\Delta t)}{\Delta t}, \quad \text{for } g \neq h, \text{ and} \\
	\alpha_{gg}(t)&\coloneqq  - \sum_{ h \neq g} \alpha_{gh}(t).
\end{split}
\end{align*}
These may be summarized to a transition intensity matrix {\boldmath${\alpha}$}$(t)$, and a corresponding matrix $\mathbf{A}(t)\coloneqq \int_0^t \mathbf{\alpha}(u) du$ of cumulative transition intensities.

In sequel, we assume that the process $X(t)$ is Markov, i.e. we assume that
\begin{align*}
	\mathbb{P}\left( X(t)=h | X(s)=g, {\cal{F}}_{s-} \right) = \mathbb{P}\left( X(t)=h | X(s)=g \right)
\end{align*}
We consider hypothesis testing on $\mathbf{P}$ when the process $X$ is Markovian. In particular, notice that we have $\mathbf{P}(s,t)=\mathbf{P}(s,u) \cdot \mathbf{P}(u,t)$ in the Markovian case. The matrices $\mathbf{P}$ and $\mathbf{A}$ are related via the matrix product-integral
\begin{align*}
	\mathbf{P}(s,t) = 
	\underset{ u \in (s,t]  }\prod \{  \mathbf{I} + d\mathbf{A}(u)  \}.
\end{align*}
To estimate the transition probability matrix, we need to estimate the matrix of cumulative transition intensities $\mathbf{A}(u)$. Let $N_{gh}(u)$ be the number of individuals who are observed to move from state $g$ to $h$ in the time interval $[0,u]$. Let $Y_g(u)$ be the number of patients who are obversed in stage $g$ just before time $t$, and $J_g(u)\coloneqq  I(Y_g(u)>0)$ for each $g \in {\cal{S}}$. Then we may estimate $\mathbf{A}(u)$ by a matrix $\hat{\mathbf{A}}(u)$ of Nelsen-Aalen estimators, where the $gh$-element is given by   
\begin{align*}
	\begin{split}
			\hat{{A}}_{gh}(t)&\coloneqq  \int_0^t \frac{dN_{gh}(u)}{Y_g(u)}, \quad \text{for } g \neq h, \text{ and} \\
			\hat{{A}}_{gg}(t)&\coloneqq  - \sum_{h \neq g} \hat{{A}}_{gh}(t).
	\end{split}
\end{align*}
This yields the following estimate of the transition probabilities
\begin{align*}
	\hat{\mathbf{P}}(s,t) = 
	\underset{ u \in (s,t]  }\prod \{  \mathbf{I} + d\hat{\mathbf{A}}(u)  \}.
\end{align*}
To characterize the large sample distribution of $\hat{\mathbf{P}}(s,t)$ we additionally define for $g,h \in {\cal{S}}$
\begin{align*}
	A_{gh}^*(t)\coloneqq \int_0^t J_{g}(u)dA_{gh}(u).
\end{align*}
Again we let $\mathbf{A}^*(t)$ the maxtrix of these elements, and  
\begin{align*}
	{\mathbf{P}}^*(s,t) = 
	\underset{ u \in (s,t]  }\prod \{  \mathbf{I} + d{\mathbf{A}}^*(u)  \}.
\end{align*}
On this basis, the matrix of transition probabilities may be rewritten as
\begin{align}\label{Sec07.eq09}
	\begin{split}
			\sqrt{n}\left( \hat{\mathbf{P}}(s,t) - {\mathbf{P}}^*(s,t) \right)  
		&= {\mathbf{Q}}(s,t) \cdot {\mathbf{P}}^*(s,t) \\
	\end{split}
\end{align}
with
\begin{align*}
	\begin{split}
		{\mathbf{Q}}(s,t) &\coloneqq  \int_0^t \hat{\mathbf{P}}(s,u-) \cdot d\mathbf{M}(u) \cdot {\mathbf{P}}^*(u,s), \qquad		\mathbf{M}(u)\coloneqq  \sqrt{n}(\hat{\mathbf{A}}(u) - \mathbf{A}^*(u))
	\end{split}
\end{align*}
This allows to derive its finite dimensional distributions in the large sample limit. First, $\mathbf{M}(t)$ is well-known to be a mean-zero ${\cal{F}}_{t}$-martingale that converges to a Gaussian process with independent increments as $n \to \infty$ by a central limit theorem of \cite{bib5}, because its optional covariation process converges in distribution to some deterministic covariance function and because its jumpsize is of order $n^{-1/2}$. Since $\hat{\mathbf{P}}$ and ${\mathbf{P}}^*$ are bounded by $1$, for any fixed $s \geq 0$, the process ${\mathbf{Q}}(s,t)$ itself is thus a mean-zero ${\cal{F}}_{t}$-martingale that converges to a Gaussian process with independent increments as $n \to \infty$ by the central limit theorem of \cite{bib5}. 
The covariance between any to components $Q_{gh}(s,t)$ and $Q_{g'h'}(s,t)$ of $\mathbf{Q}$ may be estimated using the covariation process
\begin{align*}
	\begin{split}
\sum_{i,i',j,j' \in {\cal{S}}} \int_0^t \hat{P}_{gi}(s,u-) \hat{P}_{g'i'}(s,u-)  \hat{P}_{jh}(u-,s)  \hat{P}_{j'h'}(u-,s) 
d[M_{ij},M_{i'j'}](u),
	\end{split}
\end{align*}
because $[\int H dM] = \int H^2 d[M]$ for a martingale $M$ and predictable process $H$. Exploting the independent increments structure, we have $\text{Cov}(Q_{gh}(s,t_2),Q_{g'h'}(s,t_1))= \text{Cov}(Q_{gh}(s,t_1),Q_{g'h'}(s,t_1))$ for any $t_1 \leq t_2$ in the large sample limit. This characterizes the finite dimensional distributions of the process $(\mathbf{Q}(s,t))_{t \geq 0}$ for any fixed $s \geq 0$, and in combination with Eq. \eqref{Sec07.eq09}  the finite dimensional distributions of the process $(\hat{\mathbf{P}}(s,t))_{t \geq 0}$ for any fixed $s \geq 0$: As linear transformation of a Gaussian process, $(\hat{\mathbf{P}}(s,t))_{t \geq 0}$ is in fact itself a Gaussian process. Notice, however, that in contrast to $(\mathbf{Q}(s,t))_{t \geq 0}$ the process $(\hat{\mathbf{P}}(s,t))_{t \geq 0}$ in general does not have independent increments. Since ${\mathbf{P}}^*(s,t)$ is almost the same as ${\mathbf{P}}(s,t)$ when sample size is large, we immediately see from Eq. \eqref{Sec07.eq09} that $\hat{\mathbf{P}}(s,t)$ consistently estimates ${\mathbf{P}}(s,t)$.

Knowledge of the finite dimensional distributions of $(\hat{\mathbf{P}}(s,t))_{t \geq 0}$ now enabled adaptive hypothesis test of hypotheses on the true transition probabilities ${\mathbf{P}}(s,t)$. Let us exemplarily consider a single component $P_{gh}(s,t)$ and a two arm randomized trial. Assume we wish to test the null hypothesis $H_0:P_{gh}^{(E)}(s,t) = P_{gh}^{(C)}(s,t)$ that the probability of transition from stage $g$ to $h$ between time $s$ and $t$ does not differ between some experimental treatment $E$ and control $C$. Let $\hat{\mathbf{P}}^{(x)}$ denote the estimated transition probability matrix observed in treatment group $x=E,C$. A natural statistic for testing $H_0$ is given by $\Psi(t,n)\coloneqq {\hat{P}}_{gh}^{(E)}(s,t) - {\hat{P}}_{gh}^{(C)}(s,t)$, which is normally distributed with mean zero when $H_0$ holds true by what we have seen above. Let the initial design reject $H_0$ whenever, $\Psi(t,n) \geq c$ for some critical value $c$. Now assume that we perform an interim analysis and gain interim estimates of the treatment specific shor-term transition probabilites $\hat{\mathbf{P}}^{(x)}(s,t_0)$ for some $s < t_0 < t$ and $x=C,E$. The interim estimate $\hat{\mathbf{P}}^{(x)}(s,t_0)$ is a surrogate for the long-term transition probability ${{P}}_{gh}^{(x)}(s,t)$. Assume we wish to adapt the sample size $n$ (and/or the upper limit of the transition time interval $t$) based on the observed value of $\hat{\mathbf{P}}^{(x)}(s,t_0)$. Let $n'$ and $t'$ denote the modified values for $n$ and $t$. By the conditional error principle of \cite{SM01}, control of the type I error rate is guaranteed if we also choose a new critical bound $c'$ such that  
\begin{align}\label{Sec07.eq12}
	\begin{split}
\mathbb{P}_{H_0}(\Psi(t,n)\geq c | \hat{\mathbf{P}}^{(x)}(s,t_0) = {\mathbf{p}}^{(x)}(s,t_0) ) = P_{H_0}(\Psi(t',n')\geq c'| \hat{\mathbf{P}}^{(x)}(s,t_0) = {\mathbf{p}}^{(x)}(s,t_0))
	\end{split}
\end{align}
for any ${\mathbf{p}}^{(x)}(s,t_0)$. To calculate the probabilities in Eq. \eqref{Sec07.eq12}, it suffices to know the finite dimensional distributions of the processes $(\hat{\mathbf{P}}^{(x)}(s,t))_{t \geq 0}$, $x=E,C$, which have been provided above. This achieves the implementation of the desired confirmatory adaptive design changes. Instead of hypotheses about individual components of the matrix ${\mathbf{P}}^{(x)}(s,t)$, hypotheses about linear combinations of components can also be tested by forming linear combinations of the statistics described above.

The methods described in earlier Chapters \ref{sec:single_arm_os}, \ref{sec:rct_os}, and \ref{sec:rct_pfs_os_gap} could also be placed in this general context. In a three state illness-death model $X(t) \in {\cal{S}} = \{0,1,2\}$ with the transient state of progression $1$ and the absorbing state of death $2$, we have 
\begin{align*}
	\begin{split}
S_{OS}(t)&\coloneqq  \mathbb{P}(OS > t) = P_{00}(0,t) +  P_{01}(0,t) \\
S_{PFS}(t)&\coloneqq  \mathbb{P}(PFS > t) = P_{00}(0,t) \\
S_{OS}(t)-S_{PFS}(t) &= P_{01}(0,t),
	\end{split}
\end{align*}
where PFS is the waiting time in the initial state $0$, and OS is the waiting time until absorption in state $2$.

\section{Discussion}\label{Discussion}

We considered confirmatory adaptive designs for survival trials with two correlated time-to-event endpoints. The special feature is that interim data from both time-to-event endpoints may be used to inform design modifications, while avoiding those problem arising from the patient-wise separation approach of \cite{JSJ11}. The most serious disadvantage oft he patient-wise separation approach is that either part of the primary endpoint data has to be neglected in the final test decision or worst case adjustment have to be done that result in a conservative test procedure \citep[c.f.][]{Mea16}.

Our approach exploits the independent increments structure of the underlying test statistics. Whereas the independent increment structure in classical adaptive survival trials \citep[e.g.][]{W06} only holds w.r.t. the filtration generated by a single time-to-event endpoint, however, the independent increment structure of our test statistics holds w.r.t. to the joint filtration generated by two time-to-event endpoints. In consequence, with our approach, design modifications may be informed by interim data from both time-to-event endpoints. Conversely, however, the full interim must still not be used to inform design modifications, because then the problems described by \cite{BauerPosch04} would reappear again. If interim data from three or more time-to-endpoints is to be used for design modification, all counting processes would have to be compensated for the finer filtration generated by three or more endpoints.

This sheds light on possible extensions to the methodology presented here. On the one hand, our methodology for adaptive survival trials with two correlated time-to-event endpoints could be expanded to settings where three or more time-to-event endpoints are present simultaneously as outlined in Section \ref{Sec07}. 
On the other hand, the methodology presented here assumes that all transition intensities obey the Markov property. While the Markov property makes clinical sense from our point of view and is advantageous in terms of calculation, it is not necessary from a purely mathematical point of view. In fact the counting process martingales derived in \cite{Danzer22} are valid in general without further assumptions on the transition intensities. Another possible extension is therefore to elaborate our methodology in settings where transition intensities obey other or potentially even milder structural assumptions than the Markov assumption (e.g. Semi-Markov models oder kernel based approaches).

Finally, the distribution of the test statistics was already calculated analytically in this manuscript not only under the null hypothesis but also under the contiguous alternatives. This opens up the possibility of also performing analytical power and sample size calculations for the methods presented. 


\section*{Acknowledgement}
This work was funded by the Deutsche Forschungsgemeinschaft (DFG, German Research Foundation) – 413730122. 

\vspace*{1pc}

\noindent {\bf{Conflict of Interest}}
\noindent {\it{The authors have declared no conflict of interest.}}

\appendix
\section*{Appendices}
\renewcommand{\thesection}{\Alph{section}}

\section{Technical details about the Markovian illness-death model}\label{sec:appendix:markov_idm}

\subsection{A Sufficient Criterion for Markovity}\label{subsec:appendix_markov_sufficient}

The criterion to be derived holds on the patient level. We thus fix a single patient and consider the treament group index $x$ and the stage index $\kappa$ as fixed in this Chapter. Accordingly, the index $(x \kappa)$ indicating treatment group and stage will be omitted in all expressions to simply notation in this Chapter.

In the context of the Illness-Death-Model introcuded in Section \ref{Sec02}, the Markov assumption says that $\alpha^{(2,\{2\})}(s,s_1,s)$ from \eqref{eq:02.01} does not depend on the value of $s_1$. We claim that this independence from $s_1$ is given, when there exists a (potentially unknown) survival time variable $\tilde{T}$ such that 
$
\mathbb{P}(T_i^2 > s| T_i^1 = s_1) = \mathbb{P}(\tilde{T} > s | \tilde{T} > s_1)
$, 
because then
\begin{align*}
\begin{split}
\mathbb{P}(s < T_i^2 \leq s + ds| T_i^2 > s , T_i^1 = s_1)
&= \frac{ \mathbb{P}(s < T_i^2 \leq s + ds| T_i^1 = s_1) }{  \mathbb{P}( T_i^2 > s | T_i^1 = s_1)  } \\
&= \frac{ - \mathbb{P}( T_i^2 > s + ds| T_i^1 = s_1) + \mathbb{P}(T_i^2 > s | T_i^1 = s_1) }{  \mathbb{P}( T_i^2 > s | T_i^1 = s_1)  } \\
&= 1 - \frac{ \mathbb{P}( T_i^2 > s + ds| T_i^1 = s_1) }{ \mathbb{P}( T_i^2 > s | T_i^1 = s_1)  } \\
&= 1 - \frac{ \mathbb{P}( \tilde{T} > s + ds| \tilde{T} > s_1) }{ \mathbb{P}( \tilde{T} > s | \tilde{T} > s_1 )  } \\
&= 1 - \frac{ \mathbb{P}( \tilde{T} > s + ds)/\mathbb{P}(\tilde{T} > s_1) }{ \mathbb{P}( \tilde{T} > s)/\mathbb{P}(\tilde{T} > s_1 )  } \\
&= 1 - \frac{ \mathbb{P}( \tilde{T} > s + ds) }{ \mathbb{P}( \tilde{T} > s) } \\
&= \frac{  \mathbb{P}( \tilde{T} > s) - \mathbb{P}( \tilde{T} > s + ds) }{ \mathbb{P}( \tilde{T} > s) } \\
&= \frac{  \mathbb{P}( s< \tilde{T} \leq s + ds) }{ \mathbb{P}( \tilde{T} > s) } \\
&=   \mathbb{P}( s < \tilde{T} \leq s + ds  |  \tilde{T} > s) =: \alpha_{\tilde{T}}(s) \cdot ds
\end{split}
\end{align*}
is independent from $s_1$. Heuristically, the model approach $\mathbb{P}(T_i^2 > s| T_i^1 = s_1) = \mathbb{P}(\tilde{T} > s | \tilde{T} > s_1)$ is not implausible, because this way $\mathbb{P}(T_i^2 > s| T_i^1 = s_1)$ is interpreted to ask for the probability of ''surviving'' beyond time $s$ given that one has ''survived'' at least up to time $s_1$.

\subsection{An Identity for the OS function under Markov assumption}\label{subsec:appendix_os_identity}

The identity to be derived holds true for the OS function of each patient. We may thus fix a single patient. Therefore, the patient index $i$, the treatment group index $x$ and the stage index $\kappa$ may be considered as fixed in this Chapter. Accordingly, the patient index $i$ as well as the index $(x \kappa)$ indicating treatment group and stage will be omitted in all expressions to simply notation in this Chapter.

Fix $t \geq 0$, $\epsilon > 0$ and let $Z\coloneqq  T^1 \cdot I(T^1 \leq t) \cdot I(T^2 > t)$, where $T^j$ is the time from entry to failure of type $j=1,2$. Distribution function and density of $Z$ are:
\begin{align*}
\begin{split}
&\bullet \quad \mathbb{P}(Z = 0) = \mathbb{P}(T^1 > 0 \text{ or } T^2 \leq t)  \\
&\bullet \quad \mathbb{P}(Z \leq z) = \mathbb{P}(Z \leq z, Z=0) + \mathbb{P}(Z \leq z, Z>0)  \\
& \quad \quad \quad = \mathbb{P}(T^1> t \text{ or } T^2 \leq t) \cdot I(z \geq 0) + \mathbb{P}(T^1 \leq t \wedge z, T^2 > t) \\
& \bullet \quad f_Z(z) \coloneqq  \frac{d}{dz} \mathbb{P}(Z \leq z) \\
& \quad \quad \quad = \mathbb{P}( T^1> t \text{ or } T^2 \leq t) \cdot \delta(z) + I(z \leq t) \cdot \frac{d}{dz} \mathbb{P}(T^1 \leq z, T^2 > t)
\end{split}
\end{align*}
Exploiting that $\alpha_{12}(t)\coloneqq  \lim_{\epsilon \to 0} \mathbb{P}(t < T^2 \leq t + \epsilon | T^1 = z, T^2 > t)/\epsilon$ does not depend on $z$ by Markov property, we find with $\alpha_{0j}(t) \coloneqq  \lim_{\epsilon \to 0} \mathbb{P}(t < T^j \leq t + \epsilon| T^1> t, T^2> t)/\epsilon$ ($j=1,2$) that
\begin{align*}
\begin{split}
&\mathbb{P}(t < T^2 \leq t + \epsilon) = \int_0^{\infty} \mathbb{P}(t < T^2 \leq t + \epsilon | Z = z) f_Z(z) dz \\
& = \mathbb{P}(t < T^2 \leq t + \epsilon, \{ T^1> t \text{ or } T^2 \leq t \} )  \\
& \qquad + \int_0^t \mathbb{P}(t < T^2 \leq t + \epsilon | T^1 \leq t, T^1 = z, T^2 > t) \frac{\partial}{\partial z} \mathbb{P}(T^1 \leq z, T^2 > t) dz \\
& = \mathbb{P}(t < T^2 \leq t + \epsilon, T^1> t)  \\
& \qquad + \int_0^t \mathbb{P}(t < T^2 \leq t + \epsilon | T^1 = z, T^2 > t) \frac{\partial}{\partial z} \mathbb{P}(T^1 \leq z, T^2 > t) dz \\
& = \mathbb{P}(t < T^2 \leq t + \epsilon| T^1> t, T^2> t) \mathbb{P}(T^2> t,T^1> t)  \\
& \qquad + \int_0^t \mathbb{P}(t < T^2 \leq t + \epsilon | T^1 = z, T^2 > t) \frac{\partial}{\partial z} \mathbb{P}(T^1 \leq z, T^2 > t) dz \\
& = \alpha_{02}(t) \cdot \epsilon \cdot \mathbb{P}(T^2> t,T^1> t)  + \alpha_{12}(t) \cdot \epsilon \cdot \int_0^t \frac{\partial}{\partial z} \mathbb{P}(T^1 \leq z, T^2 > t) dz + {\cal{O}}(\epsilon^2) \\
& = \alpha_{02}(t) \cdot \epsilon \cdot \mathbb{P}(T^2> t,T^1> t)  + \alpha_{12}(t) \cdot \epsilon \cdot \mathbb{P}(T^1 \leq t, T^2 > t) + {\cal{O}}(\epsilon^2). \\
\end{split}
\end{align*}
Since $\mathbb{P}(T^2> t,T^1> t) = \mathbb{P}(T^2 \wedge T^1 > t) = S_{PFS}(t)$ and $\mathbb{P}(T^1 \leq t, T^2 > t) = \mathbb{P}(T^2 > t) - \mathbb{P}(T^1 > t, T^2 > t) = S_{OS}(t) - S_{PFS}(t)$, and by letting $\epsilon \to 0$, we obtain the differential equation 
\begin{align*}
\begin{split}
\frac{d}{dt}S_{OS}(t) = - \alpha_{12}(t) \cdot S_{OS}(t) + [\alpha_{12}(t)-\alpha_{02}(t)] \cdot S_{PFS}(t),
\end{split}
\end{align*}
which is equivalent to the identity given in Eq. \eqref{eq:os_identity}. Alternatively, we could use the variation of constants method to solve above differential equation directly for $S_{OS}(t)$ which yields
\begin{align*}
\begin{split}
S_{OS}(t) = e^{-\Lambda_{12}(t)} \cdot \left\{ 1 + \int_0^t  e^{- [\Lambda_{01}(u)+\Lambda_{02}(u)-\Lambda_{12}(u)]} d[\Lambda_{12}(u)-\Lambda_{02}(u)]     \right\} 
\end{split}
\end{align*}
in terms of the cumulative transition intensities $\Lambda_{ij}(t)\coloneqq \int_0^t \alpha_{ij}(u)du$. Using that $S_{PFS}(t) = e^{- [\Lambda_{01}(u)+\Lambda_{02}(u)]}$ the latter identity may also be rewritten as 
\begin{align*}
\begin{split}
S_{OS}(t) - S_{PFS}(t) = e^{-\Lambda_{12}(t)} \cdot  \int_0^t  e^{- [\Lambda_{01}(u)+\Lambda_{02}(u)-\Lambda_{12}(u)]} d\Lambda_{01}(u)
\end{split}
\end{align*}
because with $\Lambda'(t)\coloneqq  \Lambda_{01}(u)+\Lambda_{02}(u)-\Lambda_{12}(u)$ we have
\begin{align*}
\begin{split}
&S_{OS}(t) = e^{-\Lambda_{12}(t)} \cdot \left\{ 1 + \int_0^t  e^{- [\Lambda_{01}(u)+\Lambda_{02}(u)-\Lambda_{12}(u)]} d[\Lambda_{12}(u)-\Lambda_{02}(u)]     \right\} \\
&= e^{-\Lambda_{12}(t)} \cdot \left\{ 1 + \int_0^t (-1) \cdot e^{- \Lambda'(u)} d\Lambda'(u) + \int_0^t  e^{- \Lambda'(u)} d\Lambda_{01}(u)     \right\} \\
&= e^{-\Lambda_{12}(t)} \cdot \left\{ 1 + [e^{- \Lambda'(t)} - 1 ] + \int_0^t  e^{- \Lambda'(u)} d\Lambda_{01}(u)     \right\} \\
&=  e^{- [\Lambda_{01}(t)+\Lambda_{02}(t)]} + e^{-\Lambda_{12}(t)} \cdot \int_0^t  e^{- \Lambda'(u)} d\Lambda_{01}(u).
\end{split}
\end{align*}

\subsection{Relations between the Transition Intensities Induced by the Proportional Hazards Assumption}

The relations to be derived hold true for each patient. We may thus fix a single patient. Therefore, the patient index $i$, the treatment group index $x$ and the stage index $\kappa$ may be considered as fixed in this Chapter. Accordingly, the patient index $i$ as well as the index $(x \kappa)$ indicating treatment group and stage will be omitted in all expressions to simply notation in this Chapter.

Let $S_{OS}(s)$ and $S_{PFS}(s)$ denote the true survival functions for OS and PFS, and assume that they are induced by a Markovian illness--death--model with transition intensities $\alpha_{01}(s), \alpha_{02}(s), \alpha_{12}(s)$. Let $S_{OS,0}(s)$, $S_{PFS,0}(s)$ and $\alpha_{01,0}(s), \alpha_{02,0}(s), \alpha_{12,0}(s)$ denote the corresponding expressions for some prefixed reference Markovian illness--death--model. Since the true and the reference survival functions are both derived from Markovian illness--death--models, the following equations holds true by Appendix \ref{subsec:appendix_os_identity}:
\begin{align}\label{A3.eq01}
\begin{split}
\dot{S}_{OS}(s)   & =  [ \alpha_{12}(u) - \alpha_{02}(u) ] \cdot S_{PFS}(t)   - \alpha_{12}(t) \cdot S_{OS}(u) \\
\dot{S}_{OS,0}(s) & =  [ \alpha_{12,0}(u) - \alpha_{02,0}(u) ] \cdot S_{PFS,0}(t)   - \alpha_{12,0}(t) \cdot S_{OS,0}(u).
\end{split}
\end{align}
Notice that $\dot{S}\coloneqq (d/ds) S$ means derivative w.r.t. $s$. Additionally assume validity of the proportional hazards assumption $S_{OS}(s) = S_{OS,0}(s)^{\omega_{OS}}$ and $S_{PFS}(s) = S_{PFS,0}(s)^{\omega_{PFS}}$ for some hazard ratios $0 < \omega_{OS}, \omega_{PFS} < 1$. Proportional hazards imply 
\begin{align}\label{A3.eq02}
\begin{split}
\frac{\dot{S}_{OS}(s)}{S_{OS}(s)} = \omega_{OS} \frac{\dot{S}_{OS,0}(s)}{S_{OS,0}(s)} \quad \text{and} \quad
\frac{\dot{S}_{PFS}(s)}{S_{PFS}(s)} = \omega_{PFS} \frac{\dot{S}_{PFS,0}(s)}{S_{PFS,0}(s)}.
\end{split}
\end{align}
If we take the margins $S_{OS,0}(s)$ and $S_{PFS,0}(s)$ as well as the hazard ratios $\omega_{OS}$ and $\omega_{PFS}$ as prespecified and then combine above equations \eqref{A3.eq01} and \eqref{A3.eq02}, we get the following the linear relationship between the true and referenc transition intensities:
\begin{align*}
\begin{split}
&\alpha_{12}(s) \cdot \left( \frac{S_{PFS,0}(s)^{\omega_{PFS}}}{S_{OS,0}(s)^{\omega_{OS}}} - 1 \right) - \alpha_{02}(s) \cdot \frac{S_{PFS,0}(s)^{\omega_{PFS}}}{S_{OS,0}(s)^{\omega_{OS}}} \\
&=
\alpha_{12,0}(s) \cdot \omega_{OS} \cdot \left( \frac{S_{PFS,0}(s)}{S_{OS,0}(s)} - 1 \right) - \alpha_{02,0}(s) \cdot \omega_{OS} \cdot \frac{S_{PFS,0}(s)}{S_{OS,0}(s)}.
\end{split}
\end{align*}
Since ${\dot{S}_{PFS}(s)}/{S_{PFS}(s)} = \int_0^t (\alpha_{01}(u) + \alpha_{02}(u)) du$ and ${\dot{S}_{PFS,0}(s)}/{S_{PFS,0}(s)} = \int_0^t (\alpha_{01,0}(u) + \alpha_{02,0}(u)) du$, the proportional hazards assumption for PFS is tantamount to the linear relation
\begin{align*}
\begin{split}
\alpha_{01}(t) + \alpha_{02}(t) = \omega_{PFS} \cdot (\alpha_{01,0}(t) + \alpha_{02,0}(t)).
\end{split}
\end{align*}
Thus, a given reference model ${\cal{A}}_0\coloneqq (\alpha_{01,0}(s), \alpha_{02,0}(s), \alpha_{12,0}(s))$ with given margins $S_{PFS,0}(s)$ and $S_{OS,0}(s)$, in combination with the PH assumption for PFS and OS, completely determines the true model ${\cal{A}}\coloneqq  (\alpha_{01}(s), \alpha_{02}(s), \alpha_{12}(s))$ in a simple manner once a third relation between ${\cal{A}}$ and ${\cal{A}}_0$ is given. Such a third relation could, for example, in a setting with negligible treatment related mortality, consist of the statement $\alpha_{02}(s) = \alpha_{02,0}(s) = 0$.

\subsection{Estimation of the Cumulative Hazard Functions $\hat{\Lambda}_{02}$, $\hat{\Lambda}_{12}$ and  $\hat{\Lambda}_{0*}$}

Estimation of the cumulative hazard function is done separately for each treatment group $x$ and stage $\kappa$. Therefore, the parameters $x$ and $\kappa$ may be considered as fixed and the index $(x \kappa)$ indicating treatment group and stage will be omitted in all expressions to simply notation in this Chapter.

In differential notation, the $({\cal{F}}_s)_{s \geq 0}$--martingale $M_i^{OS}(s)$ introduced in Eq. \eqref{eq:pfs_os_indiv_martingale} then reads as 
\begin{align}\label{A4.eq01}
dM_i^{OS}(s) = dN_i^{OS}(s) - I_{02,i}(s) \alpha_{02}(s) ds  - I_{12,i}(s) \alpha_{12}(s) ds 
\end{align}
with $I_{02,i}(s)$, $I_{12,i}(s)$ acc. to Eq. \eqref{eq:def_at_risk}. Let $j=0$ or $j=1$. To derive an estimator for $\Lambda_{j2}(s)$, we multiply Eq. \eqref{A4.eq01} with $I_{j2,i}(s)$ and use $I_{02,i}(s) \cdot I_{12,i}(s) = 0$ to obtain
\begin{align*}
I_{j2,i}(s) \alpha_{j2}(s) ds = I_{j2,i}(s) dN_i^{OS}(s) - I_{j2,i}(s) dM_i^{OS}(s). 
\end{align*}
Summing up over $i$ and dividing by $Y_{j2}\coloneqq  \sum_i I_{j2,i}$ after multiplication with $J_{j2}\coloneqq  I(Y_{j2}>0)$ yields
\begin{align*}
J_{j2}(s) \alpha_{j2}(s) ds = \sum_i \frac{I_{j2,i}(s)}{Y_{j2}(s)} dN_i^{OS}(s) - \sum_i \frac{I_{j2,i}(s)}{Y_{j2}(s)} dM_i^{OS}(s). 
\end{align*}
Integration over $[0,s]$ finally yields $\int_0^s J_{j2}(u) \alpha_{j2}(u) du = \hat{\Lambda}_{j2}(s) - n^{-1/2} M_{j2}(s)$ with $\hat{\Lambda}_{j2}$ acc. to Eq. \eqref{eq:na_estimators} and with mean--zero $({\cal{F}}_s)_{s \geq 0}$--martingale $M_{j2}$  as defined in Eq. \eqref{eq:02.04}. Using that $\inf_{u \in [0,s]} J_{j2}(u) \to 1$ in probability as $n \to \infty$, this shows that $\hat{\Lambda}_{j2}$ consistently estimates $\Lambda_{j2}$. 

Starting point for an estimator for $\Lambda_{0*}(s)$ is the $({\cal{F}}_s)_{s \geq 0}$--martingale $M_i^{PFS}(s)$ from Eq. \eqref{eq:pfs_os_indiv_martingale} 
\begin{align*}
dM_i^{PFS}(s) = dN_i^{PFS}(s) - I_{02,i}(s) \alpha_{0*}(s) ds. 
\end{align*}
Multiplication with $I_{02,i}(s)$ yields $I_{02,i}(s) \alpha_{0*}(s) ds = I_{02,i}(s) dN_i^{PFS}(s) - I_{02,i}(s) dM_i^{PFS}(s)$. 
Summing up over $i$ and dividing by $Y_{02} \coloneqq  \sum_i I_{02,i}$ after multiplication with $J_{02}\coloneqq  I(Y_{02}>0)$ yields
\begin{align*}
J_{02}(s) \alpha_{0*}(s) ds = \sum_i \frac{I_{02,i}(s)}{Y_{02}(s)} dN_i^{PFS}(s) - \sum_i \frac{I_{02,i}(s)}{Y_{02}(s)} dM_i^{PFS}(s). 
\end{align*}
Integration over $[0,s]$ finally yields $\int_0^s J_{02}(u) \alpha_{0*}(u) du = \hat{\Lambda}_{0*}(s) - n^{-1/2} M_{0*}(s)$ with $\hat{\Lambda}_{0*}$ acc. to Eq. \eqref{eq:na_estimators} and with mean--zero $({\cal{F}}_s)_{s \geq 0}$--martingale $M_{0*}$  as defined in Eq. \eqref{eq:02.04}. Using that $\inf_{u \in [0,s]} J_{02}(u) \to 1$ in probability as $n \to \infty$, this shows that $\hat{\Lambda}_{0*}$ consistently estimates $\Lambda_{0*}$. 

\newpage

\section{Distributional properties of the test statistics}\label{sec:appendix_distributional_properties}

In the following parts, we will prove the asymptotical statements required in sections \ref{sec:single_arm_pfs+os}, \ref{sec:single_arm_os}, \ref{sec:rct_os} and \ref{sec:rct_pfs_os_gap}. In those sections, several quantities are indexed by the treatment group $x \in \{E,C\}$ and the stage $\kappa$. Both parameters can be considered as fixed for our asymptotical assessments. Hence, unless otherwise specified, we will omit inidices $(E),(C),(E\kappa),(C\kappa)$ to simplify notation. We thus e.g. write $n$, $\hat{\Lambda}_{j2}(u)$, $M_{j2}(u), S_{OS}(u), S_{OS}(u), \Psi_{OS}(u),\ldots$ instead of $n^{(E \kappa)}(u)$, $\hat{\Lambda}_{j2}^{(E \kappa)}(u)$, $M_{j2}^{(E \kappa)}(u), S_{OS}^{(E)}(u), \Psi_{OS}^{(\kappa)}(u) \ldots$ etc. whenever no confusion is püossible.\\
Before proceeding with the individual results, we make some more general statements and give corresponding proofs that will be used throughout the following sections.

\begin{lemma}\label{lemma:uniform_lln}
	Let $Y$ denote the size of an "at risk"-set as defined in \eqref{eq:def_at_risk} with the individual censored observations being independent and identically distributed. For any compact interval $[a,b] \subset \mathbb{R}$, we have the convergence
	\begin{equation}\label{eq:convergence_at_risk_infimum}
		\sup_{u \in [a,b]} \left| \frac{Y(u)}{n} - y(u) \right| \overset{\mathbb{P}}{\to} 0
	\end{equation}
	as $n \to \infty$. Here, $y$ denotes the time-dependent expectation $y(u) \coloneqq  \mathbb{E}[I_i(u)] = \mathbb{P}[I_i(u) = 1]$ for the individual "at risk"-indicators $I$ as also defined in \eqref{eq:def_at_risk}. If $y$ is bounded from below on $[a,b]$, it also holds
	\begin{equation*}
		\sup_{u \in [a,b]} \left| \frac{n\cdot J(u)}{Y(u)} - \frac{1}{y(u)} \right| \overset{\mathbb{P}}{\to} 0.
	\end{equation*}
\end{lemma}
\begin{proof}
	The first statement is a consequence of a uniform Law of Large Numbers as e.g. given in Chapter A.5 of \cite{vanderVaart:1996}. The second one follows from the Continuous Mapping Theorem.
\end{proof}
\noindent In order to establish the asymptotic results in the following sections, we need to make assumptions about the asymptotic probability that patients are in a certain state. In particular, we assume the conditions stated in Theorems IV.1.1 and IV.1.2 of \cite{Andersen:1991} for all involved transitions. With regard to the fulfilment of these conditions by a Markov process with finite state space, we refer to Example IV.1.9 from \cite{Andersen:1991}. If the conditions of Theorem IV.1.2 are fulfilled and the transition intensities are bounded from below, they imply the following condition about the indicator function $J(s)=I(Y(s)>0)$ as defined in \eqref{eq:def_at_risk} for the different states of our model:
\begin{enumerate}[label=(C.\arabic*)]
	\item\label{condition:at_risk_process_general} The process $J$, depending on the sample size $n$, converges with
	\begin{equation*}
		\sqrt{n} \int_0^s (1 - J(u))\,du \overset{\mathbb{P}}{\to} 0.
	\end{equation*}
\end{enumerate}
It is a direct consequence of this condition that integrals of the form
\begin{equation*}
	\int_0^s S(u) J(u)\, d\Lambda(u) \quad \text{and} \int_0^s S(u)\, d\Lambda(u)
\end{equation*}
for a bounded, possibly random function $S$ and a continuous, monotonically increasing function $\Lambda$ are asymptotically equivalent. The condition follows e.g. from the following conditions that might be easier to verify:
\begin{enumerate}[label=(D.\arabic*)]
	\item\label{condition:lb_y} The function $y$ is bounded from below by a positive constant $\gamma > 0$ on $[0,s]$,
	\item\label{condition:lb_y_and_decay} There is some $\varepsilon>0$ s.t. the function $y$ is bounded from below by a positive constant $\gamma > 0$ on $[\varepsilon, s - \varepsilon]$ and 
	\begin{equation*}
		\sqrt{n} \int_0^{\varepsilon} (1-y(u))^n\, du \to 0 \quad \text{and} \quad \sqrt{n} \int_{s - \varepsilon}^s (1-y(u))^n\, du \to 0
	\end{equation*}
	as $n \to \infty$.
\end{enumerate} 
The follwoing Lemma will prove that each of the conditions \ref{condition:lb_y} and \ref{condition:lb_y_and_decay} imply \ref{condition:at_risk_process_general}.
\begin{lemma}\label{lemma:non_empty_at_risk_set}
	If the limiting function $y$ from Lemma \ref{lemma:uniform_lln} fulfills condition \ref{condition:lb_y}, then we have
	\begin{equation}\label{eq:j_with_arb_power}
		n^c \left( 1 - \inf_{u \in [a,b]} J(u) \right) \overset{\mathbb{P}}{\to} 0
	\end{equation}
	as $n \to \infty$ for all $c \in \mathbb{R}$ which implies condition \ref{condition:at_risk_process_general}. If the function fulfills condition \ref{condition:lb_y_and_decay}, the condition \ref{condition:at_risk_process_general} is also fulfilled.
\end{lemma} 
\begin{proof}
	We start with the first part of the statement. As a consequence of Lemma \ref{lemma:uniform_lln} and our assumption, we also have 
	\begin{equation*}
		\inf_{u \in [a,b]} Y(u) \eqqcolon Y^{\star} \overset{\mathbb{P}}{\to} y^{\star} \coloneqq  \inf_{u \in [a,b]} y(u) > 0.
	\end{equation*}
	Thus, for any $\delta > 0$, we have
	\begin{align*}
			&\mathbb{P}\left( n^k \cdot [1 - \inf_{u \in [a,b]} J_j(u) ] > \delta \right) \\
			= & \mathbb{P}\left( Y^{\star} = 0 \right) \\
			= & \mathbb{P}\left( Y^{\star} = 0, \left|\frac{Y^{\star}}{n}- y^{\star}\right| > \gamma/2 \right) + \mathbb{P}\left( Y^{\star} = 0, \left|\frac{Y^{\star}}{n} - y^{\star}\right| \leq \gamma/2 \right).	 
	\end{align*}
	The first summand in above equation vanishes as $n \to \infty$, because of \eqref{eq:convergence_at_risk_infimum}. The second summand equals zero, because $y^{\star} - \gamma/2 \geq \gamma/2 > 0$ by our assumptions, and thus $0 \notin [y^{\star} - \gamma/2, y^{\star} + \gamma/2]$. This proves the first part as \eqref{eq:j_with_arb_power} clearly implies \ref{condition:at_risk_process_general}.\\
	For the second part (i.e. assuming \ref{condition:lb_y_and_decay}), we decompose
	\begin{equation*}
		\int_0^s (1 - J(u))\, du = \int_0^{\varepsilon} (1 - J(u))\, du + \int_{\varepsilon}^{s-\varepsilon} (1 - J(u))\, du + \int_{s-\varepsilon}^{s} (1 - J(u))\, du.
	\end{equation*}
	The second summand vanishes by the same reasoning as in the first part. For the first and third summand, the reasoning is analogous, hence, we concentrate on the first summand. As convergence in $L^1$ implies convergence in probability, it suffices to show
	\begin{equation*}
		\int_0^{\varepsilon} (1 - J(u))\, du \overset{L^1}{\to} 0.
	\end{equation*}
	By Fubini's theorem and the independence of the individual observation this is equivalent to
	\begin{align*}
		\mathbb{E}\left[ \sqrt{n} \int_0^{\varepsilon} (1 - J(u))\, du \right] & =  \sqrt{n} \int_0^{\varepsilon} \mathbb{E}\left[ \prod_{i=1}^n (1 - I_i(u)) \right] \, du\\
		& = \sqrt{n} \int_0^{\varepsilon} \prod_{i=1}^n \mathbb{E}\left[  (1 - I_i(u)) \right] \, du\\
		& = \sqrt{n} \int_0^{\varepsilon} (1 - y(u))^n \, du
	\end{align*}
	which is guaranteed by \ref{condition:lb_y_and_decay}. 
\end{proof}
Hence, if $y$ is continuous and for all timepoints in $[0,s]$ (including its boundary points) there is a positive probability that a patient is in a certain state, condition \ref{condition:lb_y} holds. That might not be the case for the progredient state at timepoint 0. To show the pracitcal relevance of condition \ref{condition:lb_y_and_decay}, we will show that it holds for the time-homogeneous Markov model in the follwoing statement. The occupation probability function of the progredient state in the time-homogenenous Markov model can e.g. be found in \cite{Meller:2019}. 
\begin{lemma}
	The function $y\colon \mathbb{R}_+ \to [0,1]$ that describes the occupation probability of the progredient state in the time-homogeneous Markovian illness-death model which is given by 
	\begin{equation*}
		y(u) \coloneqq  \frac{\alpha_{01}}{\alpha_{12} - \alpha_{01} - \alpha_{02}} \left( \exp(-(\alpha_{01} + \alpha_{02})u) - \exp(-\alpha_{12}u) \right)
	\end{equation*}
	fulfills \ref{condition:lb_y_and_decay} for any closed interval $[0,\varepsilon]$ with $\varepsilon>0$. Here, $\alpha_{kl}$ denotes the (time-constant) transition intensities.
\end{lemma}
\begin{proof}
	Without loss of generality, we assume $\alpha_{12} > \alpha_{01} + \alpha_{02}$. Otherwise the signs in the following calculations will just change. At first, we want to bound $y$ from below. For any fixed $u$, consider the function $f_u(x) = \exp(-ux)$. From the mean value theorem, we know that
	\begin{equation*}
		\frac{f_u(\alpha_{01} + \alpha_{02}) - f_u(\alpha_{12})}{\alpha_{01} + \alpha_{02} - \alpha_{12}} = f'_u(c) = -u \cdot \exp(-c\cdot u) \leq -u \cdot \exp(-\alpha_{12} \cdot u)
	\end{equation*}
	for some $c \in [\alpha_{01} + \alpha_{02}, \alpha_{12}]$ and hence
	\begin{align*}
		y(u) \geq & \frac{\alpha_{01}}{\alpha_{12} - \alpha_{01} - \alpha_{02}} \cdot (\alpha_{01} + \alpha_{02} - \alpha_{12}) \cdot (-u) \cdot \exp(-\alpha_{12} \cdot u)\\
		= & \alpha_{01}\cdot u \cdot \exp(-\alpha_{12} \cdot u).
	\end{align*}
	As $(1-x) \leq \exp(-x)$ for all $x > 0$, we then have
	\begin{align*}
		1 - y(u) \leq & 1 - \alpha_{01}\cdot u \cdot \exp(-\alpha_{12} \cdot u)\\
		& \leq \exp(-\alpha_{01}\cdot u \cdot \exp(-\alpha_{12} \cdot u))\\
		& \leq \exp(-\underbrace{\alpha_{01} \cdot \exp(-\alpha_{12} \cdot \varepsilon)}_{\eqqcolon c_0} \cdot u).
	\end{align*}
	Now, we can bound the complete exprresion from above as
	\begin{align*}
		&\sqrt{n} \int_0^{\varepsilon} (1 - y(u))^n \, du\\
		\leq & \sqrt{n} \int_0^{\varepsilon} \exp(-c_0 \cdot n \cdot u) \, du\\
		= & \sqrt{n} \left( -\frac{1}{c_0 \cdot n} \cdot \exp(-c_0 \cdot n \cdot \varepsilon) + \frac{1}{c_0 \cdot n} \right)\\
		= &\frac{\sqrt{n}}{c_0 \cdot n} \left( 1 - \exp(-c_0 \cdot n \cdot \varepsilon) \right)\\
		\to & 0.
	\end{align*}
\end{proof}
\noindent We proceed with some further techincal results that will be required later.
\begin{lemma}\label{lemma:cmp_extended}
	Let $f \colon S \to \mathbb{R}$ be a continuous function and $S \subseteq \mathbb{R}^m$. Let the $S$-valued sequence of random variables $\hat{a}$ converge in probability against $a \in S$. Furthermore let $\hat{f}$ be a sequence of random functions that map from $S$ to $\mathbb{R}$ s.t.
	\begin{equation*}
		\sup_{x \in B_{\varepsilon}(a)} |\hat{f}(x) - f(x)| \overset{\mathbb{P}}{\to} 0
	\end{equation*}
	for some $\varepsilon > 0$. Then we have
	\begin{equation*}
		\hat{f}(\hat{a}) \overset{\mathbb{P}}{\to} f(a).
	\end{equation*}
\end{lemma}
\begin{proof}
	The difference $|\hat{f}(\hat{a}) - f(a)|$ can be bounded from above by
	\begin{equation*}
		|\hat{f}(\hat{a}) - f(\hat{a})| + |f(\hat{a}) - f(a)|.
	\end{equation*}
	The second part vanishes in probability by the Continuous Mapping Theorem. Concerning the first part, $\hat{a} \in B_{\varepsilon}(a)$ with a probability arbitrarily close to 1 if $n$ is large enough. If that is the case, the first summand can be bounded from above by $\sup_{x \in B_{\varepsilon}(a)} |\hat{f}(x) - f(x)|$ which vanishes in probability.
\end{proof}

\begin{lemma}\label{lemma:contmap_domconv}
	Let $\mathbb{R}^m$ denote the $m$-dimensional real numbers for some $m>0$. For each $x \in R^m$, let a sequence of random variables $(X_n(x))_{n=,0,1,2,\ldots}$ and a (limiting) random variable $X(x)$ be given such that $X_n(x) \to X(x)$ in probability as $n \to \infty$. Moreover let $\Gamma:\mathbb{R} \to \mathbb{R}$ be a bounded continuous map, and let $f:\mathbb{R}^m \to \mathbb{R}$ be a positive integrable function, i.e. $\int_{\mathbb{R}^m} f(x)dx < \infty$. Then $\int_{\mathbb{R}^m} \Gamma(X_n(x)) f(x)dx \to \int_{\mathbb{R}^m} \Gamma(X(x)) f(x)dx$ in probability as $n \to \infty$.
\end{lemma}
\begin{proof}
	By continuous mapping theorem, we have $\Gamma(X_n(x)) \to \Gamma(X(x))$ in probability as $n \to \infty$. Since $\Gamma(\cdot)$ is bounded by assumption, the familty of random variables $(\Gamma(X_n(x)))_{n=,0,1,2,\ldots}$ is uniformly integrable. Thus $\Gamma(X_n(x)) \to \Gamma(X(x))$ in mean as $n \to \infty$, i.e. $\mathbb{E}[| \Gamma(X_n(x)) - \Gamma(X(x))  |] \to 0$ as $n \to \infty$. We then conclude from Tonelli's theorem and dominated convergence theorem that
	\begin{align*}
		\begin{split}
			\lim_{n \to \infty} \mathbb{E} \left[ \int_{\mathbb{R}^m} | \Gamma(X_n(x)) - \Gamma(X(x))  | f(x)dx \right]
			&= \lim_{n \to \infty} \int_{\mathbb{R}^m} \mathbb{E} \left[ | \Gamma(X_n(x)) - \Gamma(X(x))  | \right] f(x)dx \\
			&=  \int_{\mathbb{R}^m} \lim_{n \to \infty} \mathbb{E} \left[ | \Gamma(X_n(x)) - \Gamma(X(x))  | \right] f(x)dx \\
			&= 0
		\end{split}
	\end{align*}
	Consequently, $\int_{\mathbb{R}^m} \Gamma(X_n(x)) f(x)dx \to \int_{\mathbb{R}^m} \Gamma(X(x))  f(x)dx$ in mean and thus also in probability as $n \to \infty$.
\end{proof}

\begin{lemma}\label{lemma:convinprob_inverse}
	Let $(\hat{G}_n)_{n \in \mathbb{N}}$ a sequence of continuous, strictly decreasing random functions from the real numbers to the interval $[0,1]$ that converges pointwise in probability to a deterministic function $G$ for all real numbers. Then
	\begin{equation*}
		\hat{G}_n^{-1}(\alpha) \underset{n\to \infty}{\overset{\mathbb{P}}{\longrightarrow}} {G}^{-1}(\alpha)
	\end{equation*}
	for all $\alpha \in [0,1]$ Here, $\hat{G}_n^{-1}$ and $G^{-1}$ denote the inverses of the respective functions.
\end{lemma}
\begin{proof}
	First of all, the generalized inverses exist as the functions are strictly increasing. Now, let $\epsilon>0$, $\gamma>0$ and $\alpha \in (0,1)$. We will show there is a natural number $n_0$ such that 
	\begin{align}\label{A7.eq125}
		\mathbb{P}\left( |\hat{G}^{-1}(\alpha) - {G}^{-1}(\alpha)| > \epsilon \right) \leq \gamma \qquad \text{for all } n \geq n_0.
	\end{align}
	Since $G^{-1}$ is continuous there is a $\delta > 0$ such that
	\begin{align*}
		|{G}^{-1}(\alpha + \delta) - {G}^{-1}(\alpha)| \leq \epsilon \text{ and }
		|{G}^{-1}(\alpha - \delta) - {G}^{-1}(\alpha)| \leq \epsilon.
	\end{align*}
	Since $G^{-1}$ is strictly decreasing, this implies
	\begin{align}\label{A7.eq127}
		{G}^{-1}(\alpha + \delta) \geq {G}^{-1}(\alpha) - \epsilon \text{ and }
		{G}^{-1}(\alpha - \delta) \leq {G}^{-1}(\alpha ) + \epsilon.
	\end{align}
	Since $\hat{G}(c) \underset{n\to \infty}{\overset{\mathbb{P}}{\longrightarrow}} {G}(c)$ for all $c \in R$, there is a natural number $n_0$ such that
	\begin{align*}
		\mathbb{P}\left( |\hat{G}({G}^{-1}(\alpha \pm \delta)) - {G}({G}^{-1}(\alpha \pm \delta))| \geq \delta \right) \leq \gamma/2
		\qquad \text{for all } n \geq n_0.
	\end{align*}
	Since ${G}({G}^{-1}(\alpha \pm \delta)) = \alpha \pm \delta$, this is tantamount to
	\begin{align}\label{A7.eq129}
		\mathbb{P}\left( |\hat{G}({G}^{-1}(\alpha \pm \delta)) - (\alpha \pm \delta)| < \delta \right) \geq 1- \gamma/2
		\qquad \text{for all } n \geq n_0.
	\end{align}
	Now, $|\hat{G}({G}^{-1}(\alpha + \delta)) - (\alpha + \delta)| < \delta$ implies $\hat{G}({G}^{-1}(\alpha + \delta)) > \alpha$. Likewise, $|\hat{G}({G}^{-1}(\alpha + \delta)) - (\alpha - \delta)| < \delta$ implies $\hat{G}({G}^{-1}(\alpha + \delta)) < \alpha$. We thus conclude from \eqref{A7.eq129} that for all $n \geq n_0$:
	\begin{align*}
		\mathbb{P}\left( \hat{G}({G}^{-1}(\alpha + \delta)) > \alpha \right) \geq 1- \gamma/2
		\text{ and }
		\mathbb{P}\left( \hat{G}({G}^{-1}(\alpha - \delta)) < \alpha \right) \geq 1- \gamma/2. 
	\end{align*}
	Since $\hat{G}^{-1}$ is strictly decreasing, this implies that for all $n \geq n_0$:
	\begin{align}\label{A7.eq131}
		\mathbb{P}\left( {G}^{-1}(\alpha + \delta) < \hat{G}^{-1}(\alpha) \right) \geq 1- \gamma/2
		\text{ and }
		\mathbb{P}\left( {G}^{-1}(\alpha - \delta) > \hat{G}^{-1}(\alpha) \right) \geq 1- \gamma/2. 
	\end{align}
	Thus, for all $n \geq n_0$:
	\begin{align}\label{A7.eq132}
		\begin{split}
			&\mathbb{P}\left( {G}^{-1}(\alpha + \delta) < \hat{G}^{-1}(\alpha) < {G}^{-1}(\alpha - \delta) \right) \\
			= & 1 - \mathbb{P}\left({G}^{-1}(\alpha + \delta) \geq \hat{G}^{-1}(\alpha) \text{ or } \hat{G}^{-1}(\alpha) \geq {G}^{-1}(\alpha - \delta) \right) \\
			\geq & 1 - \mathbb{P}\left({G}^{-1}(\alpha + \delta) \geq \hat{G}^{-1}(\alpha) \right) - \mathbb{P}\left(\hat{G}^{-1}(\alpha) \geq {G}^{-1}(\alpha - \delta) \right) \\
			= & \mathbb{P}\left({G}^{-1}(\alpha + \delta) < \hat{G}^{-1}(\alpha) \right) + \mathbb{P}\left(\hat{G}^{-1}(\alpha) < {G}^{-1}(\alpha - \delta) \right) -1 \\ 
			\geq & 1 - \gamma \text{ by Eq. \eqref{A7.eq131}}.
		\end{split}
	\end{align}
	By Eq. \eqref{A7.eq127} and since ${G}^{-1}$ is strictly decreasing, notice that
	\begin{align*}
		{G}^{-1}(\alpha) - \epsilon \leq {G}^{-1}(\alpha + \delta) \leq
		{G}^{-1}(\alpha - \delta) \leq {G}^{-1}(\alpha ) + \epsilon.
	\end{align*}
	This implies in combination with Eq. \eqref{A7.eq132} that 
	\begin{align*}
		\mathbb{P}\left( {G}^{-1}(\alpha) - \epsilon \leq \hat{G}^{-1}(\alpha) \leq {G}^{-1}(\alpha) + \epsilon \right) \geq 1 - \gamma
	\end{align*}
	for all $n \geq n_0$, which is exactly the statement from Eq. \eqref{A7.eq125} to be proven. 
\end{proof}

\begin{lemma}\label{lemma:conv_integral_product_estimators}
	For a stochastic integral of the form
	\begin{equation}\label{eq:integral_difference_estimator}
		\int_0^s (\hat{X}(u)\hat{Y}(u) - X(u)Y(u)) \, d\hat{U}(u)
	\end{equation}
	with continuous functions $X,Y,U$ on $[0,s]$ and its estimators $\hat{X}, \hat{Y}, \hat{U}$, we obtain convergence to $0$ in probability under the following conditions as $n \to \infty$:
	\begin{enumerate}[label=(\roman*)]
		\item $\sup_{u \in [0,s]} (\hat{X}(u) - X(u)) \overset{\mathbb{P}}{\to} 0$,
		\item $\sup_{u \in [0,s]} (\hat{Y}(u) - Y(u)) \overset{\mathbb{P}}{\to} 0$,
		\item $\sup_{u \in [0,s]} (\hat{U}(u) - U(u)) \overset{\mathbb{P}}{\to} 0$.
	\end{enumerate}
	Also, the convergence is uniform over all upper bounds of the integral range in \eqref{eq:integral_difference_estimator}. 
\end{lemma}
\begin{proof}
	We can rewrite
	\begin{equation*}
		\hat{X}\hat{Y} - XY = (\hat{X} - X)(\hat{Y} - Y) + Y(\hat{X} - X) + X(\hat{Y} - Y)
	\end{equation*}
	to establish the bound
	\begin{align*}
		&\int_0^s (\hat{X}(u)\hat{Y}(u) - X(u)Y(u)) \, d\hat{U}(u)\\
		\leq& \sup_{u \in [0,s]} \left| (\hat{X}(u) - X(u))(\hat{Y}(u) - Y(u)) + Y(u)(\hat{X}(u) - X(u)) + X(u)(\hat{Y}(u) - Y(u)) \right| \cdot \hat{U}(s)\\
		\leq&\Bigg( \sup_{u \in [0,s]}|\hat{X}(u) - X(u)| \cdot \sup_{u \in [0,s]}|\hat{Y}(u) - Y(u)| + \sup_{u \in [0,s]} |Y(u)| \cdot \sup_{u \in [0,s]}|\hat{X}(u) - X(u)| + \\
		& \qquad \sup_{u \in [0,s]} |X(u)| \cdot \sup_{u \in [0,s]}|\hat{Y}(u) - Y(u)| \Bigg) \cdot \hat{U}(s).
	\end{align*}
	By our assumptions, the first factor vanishes to $0$ and the second one converges to a constant as $n \to \infty$. This proves our statement.
\end{proof}

\begin{lemma}\label{lemma:F_convergence}
	Let $\Lambda_1, \Lambda_2$ denote two hazard functions and let $\hat{\Lambda}_1, \hat{\Lambda}_2$ be their Nelson-Aalen estimators. If for both transition intensities/hazards to be estimated, our assumption \ref{condition:at_risk_process_general} holds, then we have
	\begin{equation*}
		\sup_{u \in [0,s]} \left| \underbrace{\int_0^u \exp(-\hat{\Lambda}_1(v-)) d\hat{\Lambda}_2(v)}_{\eqqcolon \hat{F}(u)} - \underbrace{\int_0^u \exp(-\Lambda_1(v)) d\Lambda_2(v)}_{\eqqcolon F(u)} \right| \overset{\mathbb{P}}{\to} 0.
	\end{equation*} 
\end{lemma}

\begin{proof}
	By $J_l$ and $M_l$ for $l \in \{1,2\}$ we denote the at risk indicator and the counting process martingale corresponding to $\hat{\Lambda}_l$ as in \eqref{eq:def_at_risk} and in \eqref{eq:02.04}, respectively.\\
	We can decompose the difference as follows:
	\begin{align*}
		&\hat{F}(u) - F(u)\\
		=&\int_0^u n^{-1/2} \exp(-\hat{\Lambda}_1(v-)) \; d(\sqrt{n}(\hat{\Lambda}_2(v) - \Lambda_2(v))) + \int_0^u ( \exp(-\hat{\Lambda}_1(v-)) - \exp(-\Lambda_1(v)) ) \; d\Lambda_2(v)\\
		=&\int_0^u n^{-1/2} \exp(-\hat{\Lambda}_1(v-)) \; dM_2(v) + \int_0^u \exp(-\hat{\Lambda}_1(v-)) (J_2(v) - 1) d\Lambda_2(v)\\
		& + \int_0^u ( \exp(-\hat{\Lambda}_1(v-)) - \exp(-\Lambda_1(v)) ) \; d\Lambda_2(v)
	\end{align*}
	All of the three summands converge to $0$ uniformly over $[0,s]$. This can be shown as follows.\\[4pt]
	\textsc{First summand:}\\
	For an arbitrary $\varepsilon>0$ and $\delta>0$, we choose $n$ large enough s.t. 
	\begin{equation*}
		\mathbb{P} \left[ \sup_{u \in [0,s]} |\exp(-\hat{\Lambda}_1(u-)) - \exp(-\Lambda_1(u))| > \varepsilon/2 \right] \leq \delta/3
	\end{equation*}
	Then, we bound the supremum over $[0,s]$ by
	\begin{align*}
		&\mathbb{P}\left[\sup_{u \in [0,s]} \left| \int_0^u n^{-1/2} \exp(-\hat{\Lambda}_1(v-)) \; dM_2(v) \right| > \varepsilon\right]\\
		\leq & \mathbb{P}\left[\sup_{u \in [0,s]} \left| \int_0^u n^{-1/2} \underbrace{\exp(-\hat{\Lambda}_1(v-)) - \exp(-\Lambda_1(v))}_{|\dots| \leq 2} \; dM_2(v) \right| > \varepsilon/2 \right]\\
		&+\mathbb{P}\left[\sup_{u \in [0,s]} \left| \int_0^u n^{-1/2} \exp(-\Lambda_1(v)) \; dM_2(v) \right| > \varepsilon/2 \right]\\
		\leq& \mathbb{P}\left[\sup_{u \in [0,s]} M_2(u) > \sqrt{n} \varepsilon/4 \right]\\
		& + \mathbb{P}\left[\sup_{u \in [0,s]} \left| \int_0^u n^{-1/2} \exp(-\Lambda_1(v)) \; dM_2(v) \right| > \varepsilon/2 \right].
	\end{align*}
	For $n$ large enough, the first summand is smaller than $\delta/3$ as $M_2$ converges to a Wiener process. The last summand also becomes small enough by the inequality of \cite{bib6}.\\[4pt]
	\textsc{Second summand:}\\
	As mentioned above, it is a direct consequence of \ref{condition:at_risk_process_general} that this summand vanishes.\\[4pt]
	\textsc{Third summand:}\\
	Since $|\exp(x) - \exp(x^{\prime})| \leq |x - x^{\prime}|$ for all $x,x^{\prime}\geq 0$, we get convergence of this summand by uniform convergence of the Nelson-Aalen estimators as e.g. shown in Theorem IV.1.1 of \cite{Andersen:1991}.
\end{proof}

\begin{lemma}\label{lemma:taylor_expansions}
	Under the contiguous alternatives as introduced in \eqref{eq:contiguous_alternatives_both} and \eqref{eq:contiguous_alternatives_os}, we have the follwoing convergences for $x \in \{PFS,OS\}$ as $n \to \infty$:
	\begin{enumerate}[label=(\roman*)]
		\item $\sqrt{n} (1 - \omega_x) \to \gamma_x$,
		\item Uniformly on $u \in [0,s]$
		\begin{equation*}
			\sqrt{n} \Delta_{x}(u) \to S_{x,0}(u) \log(S_{x,0}(u)) \cdot \gamma_{x}
		\end{equation*}
	\end{enumerate}
\end{lemma}
\begin{proof}
	We require the two Taylor expansions of $1 - a^{-y}$ at $y=0$ which is given by
	\begin{align*}
		1 - a^{-y} = & \sum_{k=1}^{\infty} \frac{(-y)^k \log(a)^k}{k!}
	\end{align*}
	For the first result, we plug $a = e$ and $y = \gamma_x/\sqrt{n}$ into the first expansion and mutiply by $\sqrt{n}$ to obtain
	\begin{align*}
		\sqrt{n} (1 - \omega_x) = & \sqrt{n} \sum_{k=1}^{\infty} \frac{(-\gamma_x)^k}{\sqrt{n}^k \cdot k!}\\
		= & -\gamma_x + \frac{1}{\sqrt{n}} \sum_{k=2}^{\infty} \frac{(-\gamma_x)^k}{\sqrt{n}^{k-2} \cdot k!}.
	\end{align*}
	The sum on the right hand side is finite and hence we obtain the first convergence result. For the second one we plug in $a = S_{x,0}$ and $y = 1 - \omega_x$ to obtain
	\begin{align*}
		\sqrt{n} \Delta_{x}(u) = & \sqrt{n} S_{x,0}(u) (1 - S_{x,0}(u)^{\omega_x - 1})\\
		= & \sqrt{n} S_{x,0}(u) (1 - \omega_x) \log(S_{x,0}(u)) + \sqrt{n} \sum_{k=2}^{\infty} \frac{(\omega_x - 1)^k \log(S_{x,0}(u))^k}{k!}
	\end{align*}
	Now, we apply our first convergence result to see that the first summand on the right hand side converges to the desired quantity. For the remainder, we can pull $(\omega_x - 1)^2$ out of the sum. The factor $\sqrt{n}(\omega_x - 1)^2$ vanishes and the remaining series is bounded uniformly over $u \in [0,s]$.
\end{proof}

\begin{lemma}\label{lemma:rebolledo_condition_iii}
	Let $Y$ be a "number at risk" process on $[0,s]$ and $M$ be the corresponding compensated martingale for the Nelson-Aalen estimator as in \eqref{eq:na_martingale}. We define the process $\Psi$ of the form
	\begin{equation*}
		\Psi(u) \coloneqq  \int_0^u G(v) dM(v).
	\end{equation*}
	If it fulfills one of the following conditions
	\begin{enumerate}[label=(\roman*)]
		\item\label{condition:rebolledo_1} $Y$ fulfills condition \ref{condition:lb_y} and $G$ is a bounded deterministic function,
		\item\label{condition:rebolledo_2} $G=y=\mathbb{E}[Y/n]$ and $y$ is bounded from below on $[\delta,s]$ for any $\delta>0$ by some $\gamma_{\delta} > 0$,
	\end{enumerate}
	then it also fulfills condition (iii) of Rebolledo's Martingale Central Limit Theorem (as stated in Theorem II.5.1 of \cite{Andersen:1991}).\\
	This also holds for multivariate processes with components of the form given above.
\end{lemma}
\begin{proof}
	We start with assumption \ref{condition:rebolledo_1} and fix an arbitrary $\varepsilon > 0$. As in \cite{Andersen:1991}, let $\Psi_{\varepsilon}$ denote the martingale that only contains the jumps of $\Psi$ that are larger in absolute value than $\varepsilon$. The process $\Psi$ is only in danger of making a jump of size larger than $\varepsilon$ if $Y$ is small enough. Hence, it is given by
	\begin{equation*}
		\Psi_{\varepsilon}(u) = \int_0^u G(v) \mathbbm{1}\left( Y(v) < \frac{G(v) \cdot \sqrt{n}}{\varepsilon} \right) \, dM(v).
	\end{equation*} 
	Its predictable covariation process $\langle \Psi_{\varepsilon}\rangle$ is given by
	\begin{equation*}
		\langle \Psi_{\varepsilon}\rangle (u) = n \cdot \int_0^u G(v)^2 \cdot \alpha(v) \cdot \frac{J(v)}{Y(v)} \cdot \mathbbm{1}\left( Y(v) < \frac{G(v) \cdot \sqrt{n}}{\varepsilon} \right) dv. 
	\end{equation*}
	The random factor $n/Y(v)$ of the integrand converges by our assumption and Lemma \ref{lemma:uniform_lln} uniformly to some deterministic function $1/y(v)$ which is bounded from above by $1/\gamma$ over $[0,s]$. Hence, if the the process $\tilde{\Psi}$ with
	\begin{equation*}
		\tilde{\Psi}(u) \coloneqq  \int_0^u G(v)^2 \cdot \frac{2}{\gamma} \cdot \alpha(v) \cdot \mathbbm{1}\left( Y(v) < \frac{G(v) \cdot \sqrt{n}}{\varepsilon} \right) dv,
	\end{equation*}
	vanishes pointwise in probability, the condition (iii) of Rebolledo's Theorem holds for $\Psi$. Now, for some arbitrary $\eta > 0$ it holds
	\begin{align*}
		\mathbb{P}\left[\tilde{\Psi}(u) > \eta\right] & \leq \mathbb{P}\left[\tilde{\Psi}(s) > \eta\right]\\
		& \leq \mathbb{P}\left[\inf_{0 \leq v \leq s} \frac{Y(v)}{n} < \frac{\varepsilon}{\sqrt{n}} \sup_{0 \leq v \leq s} G(v)\right].
	\end{align*}
	The left hand side of the statement inside the probability converges against the constant $\gamma$ while the right hand side converges to 0. Hence, the probability vanishes as $n \to \infty$, which proves our statement.\\
	When assuming \ref{condition:rebolledo_2} we can fix any $\delta>0$ as in the condition, plug in $y$ for $G$ and decompose
	\begin{equation}\label{eq:jump_compensator_decomposition}
		\begin{split}
		\langle \Psi_{\varepsilon}\rangle (u) = & n \cdot \int_0^{\delta} y(v)^2 \cdot \alpha(v) \cdot \frac{J(v)}{Y(v)} \cdot \mathbbm{1}\left( Y(v) < \frac{y(v) \cdot \sqrt{n}}{\varepsilon} \right) dv\\
		& + n \cdot \int_{\delta}^u y(v)^2 \cdot \alpha(v) \cdot \frac{J(v)}{Y(v)} \cdot \mathbbm{1}\left( Y(v) < \frac{y(v) \cdot \sqrt{n}}{\varepsilon} \right) dv.
		\end{split}
	\end{equation}
	For the second summand of \eqref{eq:jump_compensator_decomposition}, we can apply the same reasoning as before. For the first one, we decompose into
	\begin{equation}\label{eq:jump_compensator_decomposition_II}
		\begin{split}
		& \int_0^{\delta} y(v) \cdot \alpha(v) \cdot \left( n\cdot y(v) \frac{J(v)}{Y(v)} - 1 \right) \cdot \mathbbm{1}\left( Y(v) < \frac{y(v) \cdot \sqrt{n}}{\varepsilon} \right) dv\\
		+ & \int_0^{\delta} y(v) \cdot \alpha(v) \cdot \mathbbm{1}\left( Y(v) < \frac{y(v) \cdot \sqrt{n}}{\varepsilon} \right) dv.
		\end{split}
	\end{equation}
	For the first summand of \eqref{eq:jump_compensator_decomposition_II}, we will show that it converges to $0$ in $L^1$ and hence also in probability.
	\begin{align*}
		& \lim_{n\to \infty} \mathbb{E} \left[ \int_0^{\delta} y(v) \cdot \alpha(v) \cdot \left( n\cdot y(v) \frac{J(v)}{Y(v)} - 1 \right) \cdot \mathbbm{1}\left( Y(v) < \frac{y(v) \cdot \sqrt{n}}{\varepsilon} \right) dv \right]\\
		\leq & \lim_{n\to \infty} C \cdot \mathbb{E} \left[\int_0^{\delta} n \cdot y(v) \cdot \frac{J(v)}{Y(v)} - 1 dv \right]\\
		= & C \cdot \lim_{n\to \infty} \int_0^{\delta} \mathbb{E} \left[ n \cdot y(v) \cdot \frac{J(v)}{Y(v)} - 1 \right] dv\\
		= & C \cdot \int_0^{\delta} \lim_{n\to \infty} \mathbb{E} \left[ n \cdot y(v) \cdot \frac{J(v)}{Y(v)} - 1 \right] dv\\
		= & 0
	\end{align*}
	where the third and the fourth line are a consequence of Tonelli's lemma and the Dominated Convergence Theorem, respectively. When proving that $\mathbb{P}[\langle \Psi_{\varepsilon}\rangle (u) < \eta]$ becomes small enough for any $\eta > 0$ as $n \to \infty$ we need to choose $\delta$ small enough s.t. the second summand of \eqref{eq:jump_compensator_decomposition_II} is smaller than $\eta/2$ s.t. the rest follows analogously to the previous arguments.\\
	The last statement is a consequence of the formulation of condition (iii) of Theorem II.5.1 of \cite{Andersen:1991} as we can reduce the multivariate process to its components.
\end{proof}

\begin{lemma}\label{lemma:na_convergence}
	Let $\hat{\Lambda}$ be a Nelson-Aalen estimator of a true cumulative hazard function $\Lambda$ as e.g. defined in \eqref{eq:na_estimators}. By $J$ and $Y$ we denote the corresponding variables that describe the at risk-process. We assume that they fulfill the conditions of Theorem IV.1.2 of \cite{Andersen:1991}. Then, we obtain the uniform convergence
	\begin{equation*}
		\sup_{u \in [0,s]} n^k |\hat{\Lambda}(u-) - \Lambda(u)| \overset{\mathbb{P}}{\to} 0
	\end{equation*}
	for any $k < 1/2$.
\end{lemma}
\begin{proof}
	By the triangle inequality, we have
	\begin{equation*}
		\sup_{u \in [0,s]} n^k |\hat{\Lambda}(u-) - \Lambda(u)| \leq \sup_{u \in [0,s]} n^k |\hat{\Lambda}(u-) - \hat{\Lambda}(u)| + \sup_{u \in [0,s]} n^k |\hat{\Lambda}(u) - \Lambda(u)|.
	\end{equation*}
	The first summand is the maximum jump size of $\hat{\Lambda}$ on $[0,s]$. As we assume that Theorem IV.1.2 of \cite{Andersen:1991} can be applied, the convergence
	\begin{equation}\label{eq:na_error}
		\sqrt{n} \left( \hat{\Lambda}(\cdot) - \Lambda(\cdot) \right) \overset{\mathcal{D}}{\to} U 
	\end{equation}
	holds, where $U$ is a Gaussian martingale with a structure of independent increments. This process has almost surely continuous sample paths. It follows from Theorem 13.4 of \cite{Billingsley:1999} that the maximum jump size of the process in \eqref{eq:na_error} converges to 0 in probability. This still holds if we replace $\sqrt{n}$ by $n^k$ for some $k < 1/2$.\\
	For the second summand we can also apply a triangle inequality to obtain
	\begin{equation}\label{eq:na_triangle}
		\begin{split}
			& \sup_{u \in [0,s]} n^k |\hat{\Lambda}(u) - \Lambda(u)|\\
			\leq & \sup_{u \in [0,s]} n^k |\hat{\Lambda}(u) - \tilde{\Lambda}(u)| + \sup_{u \in [0,s]} n^k |\tilde{\Lambda}(u) - \Lambda(u)|
		\end{split}
	\end{equation}
	where $\tilde{\Lambda}$ is the modified cumulative hazard function that only increases if $J=1$, i.e.
	\begin{equation*}
		\tilde{\Lambda}(u) \coloneqq  \int_0^u J(v) d\Lambda(v).
	\end{equation*}
	In particular, we obtain $n^k (\hat{\Lambda}(u) - \tilde{\Lambda}(u)) = n^{k-1/2} M(u)$ for a martingale as defined in \eqref{eq:na_martingale} and hence it is itself a martingale. For the first summand of the right hand side of \eqref{eq:na_triangle} we can hence apply Lenglart's inequality as stated in Chapter II.5.2.1 of \cite{Andersen:1991} to obtain
	\begin{align*}
		\mathbb{P}\left[ \sup_{u \in [0,s]} n^k |\hat{\Lambda}(u) - \Lambda(u)| > \varepsilon \right]=&\mathbb{P}\left[ \sup_{u \in [0,s]} |n^{k-1/2} M(u)| > \varepsilon \right]\\
		\leq & \frac{d}{\varepsilon^2} + \mathbb{P}[\langle n^{k-1/2} M(u) \rangle (s) > d]\\
		= & \frac{d}{\varepsilon^2} + \mathbb{P}\left[ n^{2k} \int_0^s \frac{J(v)}{Y(v)} d\Lambda(v) \right]\\
		\leq & \frac{d}{\varepsilon^2} + \mathbb{P}\left[ n^{2k} \cdot \Lambda(s) \cdot \frac{1}{\inf_{u \in [0,s]} Y(u)} > d \right]\\
		= & \frac{d}{\varepsilon^2} + \mathbb{P}\left[ \underbrace{\inf_{u \in [0,s]} \frac{Y(u)}{n}}_{\overset{\mathbb{P}}{\to} \inf_{u \in [0,s]} y(u) \geq \gamma} < \underbrace{n^{2k-1}}_{\to 0} \cdot \frac{\Lambda(s)}{d} \right]
	\end{align*}
	for arbitrary $\varepsilon, d > 0$. Now, for any fixed $\varepsilon$, we can choose $d$ small enough s.t. for $n$ large enough, the sum on the right hand side becomes arbitrarily small which proves convergence for the first summand of the right hand side of \eqref{eq:na_triangle}. For the second summand, we have for some arbitrary $\varepsilon > 0$ that
	\begin{align*}
		\mathbb{P}\left[ \sup_{u \in [0,s]} n^k |\tilde{\Lambda}(u) - \Lambda(u)| > \varepsilon \right] = & \mathbb{P}\left[  n^k \cdot \int_0^s (1 - J(v)) d\Lambda(v) > \varepsilon \right]
	\end{align*}
	This vanishes under the assumptions of Theorem IV.1.2 of \cite{Andersen:1991}.
\end{proof}

\begin{lemma}\label{lemma:ls_integral_reformulation}
	Let $\hat{\Lambda}$ be a Nelson-Aalen estimator of a cumulative hazard function $\Lambda$ as in the preceding Lemma \ref{lemma:na_convergence}. By $\hat{S}$ we denote the corresponding Nelson-Altshuler estimator $\hat{S}=\exp(-\hat{\Lambda})$ for the true survival function $S$. Also let $G$ be an increasing, continuous function. Then, as $n \to \infty$, we have
	\begin{equation*}
		\sqrt{n} \left| \int_0^s (\hat{S}(u) - S(u))\,dG(u) + \int_0^s (\hat{\Lambda}(u) - \Lambda(u))S(u)\,dG(u) \right| \overset{\mathbb{P}}{\to} 0.
	\end{equation*}
\end{lemma}
\begin{proof}
	From the power series representation of the exponential function, we obtain
	\begin{align*}
		& \sqrt{n} \int_0^s (\hat{S}(u) - S(u))\, dG(u)\\
		= & \sqrt{n} \int_0^s \left(\frac{\hat{S}(u)}{S(u)} - 1 \right)S(u)\,dG(u)\\
		=&\sqrt{n} \int_0^s \left( \exp(-(\hat{\Lambda}(u) - \Lambda(u))) - 1 \right) S(u)\,dG(u)\\
		=&-\sqrt{n} \int_0^s \left(\hat{\Lambda}(u) - \Lambda(u) \right) S(u)\,dG(u) + \underbrace{\sqrt{n} \int_0^s \sum_{k=2}^{\infty} \frac{1}{k!} (\Lambda(u) - \hat{\Lambda}(u))^k S(u)\,dG(u)}_{\eqqcolon R(s)}
	\end{align*}
	To prove our claim, we need to show that $|R(s)| \overset{\mathbb{P}}{\to} 0$. For this purpose, let $\Delta(u) \coloneqq  \hat{\Lambda}(u) - \Lambda(u)$ and $\Delta^{\star}(u)\coloneqq  \sup_{v \in [0,u]} |\Delta(v)|$. For an arbitrary $\varepsilon > 0$, we have
	\begin{align*}
		\mathbb{P}(|R(s)|>\epsilon)  & = \mathbb{P}(|R(s)|>\epsilon, \Delta^*(s) \leq 1) +  \mathbb{P}(|R(s)|>\epsilon, \Delta^*(s) > 1) \nonumber \\
		&\leq \mathbb{P}\left(  \sqrt{n} \int_0^s \sum_{k=1}^{\infty} [(k+1)!]^{-1} |\Delta(u)| \cdot |\Delta(u)|^k dF_j(u) >\epsilon, \Delta^*(s) \leq 1  \right) +  \mathbb{P}(\Delta^*(s) > 1) \nonumber \\
		&\leq \mathbb{P}\left(  \sqrt{n} \int_0^s \sum_{k=1}^{\infty} [(k+1)!]^{-1} |\Delta(u)| \cdot |\Delta(u)| dF_j(u) >\epsilon, \Delta^*(s) \leq 1  \right) +  \mathbb{P}(\Delta^*(s) > 1) \nonumber \\
		&\leq \mathbb{P}\left(  (e-2) \cdot \int_0^s \sqrt{n} \Delta(u)^2 dF_j(u) > \epsilon   \right) +  \mathbb{P}(\Delta^*(s) > 1) \label{eq:rest_term}
	\end{align*}
	The second summand vanishes by Lemma \ref{lemma:na_convergence}. To see that the first summand also vanishes as $n \to \infty$, note that the following inequality holds
	\begin{align*}
		\mathbb{P}\left(  (e-2) \cdot \int_0^s \sqrt{n} \Delta(u)^2 dF_j(u) > \epsilon   \right)
		&\leq \mathbb{P}\left(  \sqrt{n} \cdot \sup_{u \in [0,s]} \Delta(u)^2  > \frac{\epsilon}{(e-2) \cdot F_j(s)}   \right) \\
		& = \mathbb{P}\left(  n^{1/4} \Delta^*(s)  > \sqrt{\frac{\epsilon}{(e-2) \cdot F_j(s)} }  \right),
	\end{align*}
	which also vanishes by Lemma \ref{lemma:na_convergence}. Consequently, $\sqrt{n} \int_0^s [ \hat{S}_{PFS}(u) - S_{PFS}(u) ] d\Lambda_{j2}(u)$ equals $-\sqrt{n} \int_0^s [ \hat{\Lambda}_{0*}(u-) - \Lambda_{0*}(u) ] dF_j(u) \equiv - \int_0^s M_{0*}(u) dF_j(u)$ up to terms that vanish in probability as $n \to \infty$.
\end{proof}

\subsection{Proof of Distributional Properties of $\Psi_{OS}^{(\kappa)}(s)$ from Chapter \ref{Sec03.02}, Eq. \eqref{eq03.03}}
In this chapter we address the setting of a single-arm study. All patients receive the experimental therapy, i.e. the therapy group variable only takes the value $x=E$. We assume that condition \ref{condition:at_risk_process_general} holds for $J_{01}$ and $J_{02}$ and hence we can apply Lemma \ref{lemma:non_empty_at_risk_set}.\\
Assume that the contiguous alternatives $H_1$ as defined in Eq. \eqref{eq:contiguous_alternatives_both} hold true, and let $\Delta_x(u)\coloneqq  S_{x,0}(u) - S_{x}(u)$ for $x=PFS, OS$. Then, due to $dM_{j2}=\sqrt{n}(d\hat{\Lambda}_{j2} - J_{j2}d{\Lambda}_{j2})$:
\begin{align*}
\begin{split}
\Psi_{OS}(s) = & \sqrt{n} \left(  \int_0^s S_{PFS,0}(u) d\hat{\Lambda}_{02}(u) 				+ \int_0^s [S_{OS,0}(u) - S_{PFS,0}(u)] d\hat{\Lambda}_{12}(u) 				+ S_{OS,0}(u) - 1				\right) \\
= & \int_0^s S_{PFS,0}(u) dM_{02}(u) 				+ \int_0^s [S_{OS,0}(u) - S_{PFS,0}(u)] dM_{12}(u) \\
&+ \sqrt{n} \left(  \int_0^s S_{PFS,0}(u) J_{02}(u) d{\Lambda}_{02}(u) 				+ \int_0^s [S_{OS,0}(u) - S_{PFS,0}(u)] J_{12}(u) d{\Lambda}_{12}(u) 				+ S_{OS,0}(u) - 1				\right) \\
= & \int_0^s S_{PFS,0}(u) dM_{02}(u) 				+ \int_0^s [S_{OS,0}(u) - S_{PFS,0}(u)] dM_{12}(u) \\
& + \sqrt{n} \left(  \int_0^s S_{PFS}(u) J_{02}(u) d{\Lambda}_{02}(u) 				        + \int_0^s [S_{OS}(u) - S_{PFS}(u)] J_{12}(u) d{\Lambda}_{12}(u) 				+ S_{OS,0}(u) - 1	  	\right) \\
& + \sqrt{n} \left(  \int_0^s \Delta_{PFS}(u) J_{02}(u) d{\Lambda}_{02}(u) 				  + \int_0^s [\Delta_{OS}(u) - \Delta_{PFS}(u)] J_{12}(u) d{\Lambda}_{12}(u) 									 	\right). 
\end{split}
\end{align*}
By our assumption \ref{condition:at_risk_process_general} and by Lemma \ref{lemma:taylor_expansions} we see that the sum of the second and the third summand amouts to the drift term given in \eqref{eq:os_drift_intersection}. 
The first summand is an $({\cal{F}}_{s})_{s \geq 0}$--martingale with covariation process 
\begin{align*}
\begin{split}
& \int_0^s S_{PFS,0}(u)^2 d[M_{02}](u)  + 2 \cdot \int_0^s S_{PFS,0}(u)[S_{OS,0}(u) - S_{PFS,0}(u)] d[M_{02},M_{12}](u) 		\\ 
& + \int_0^s [S_{OS,0}(u) - S_{PFS,0}(u)]^2 d[M_{12}](u) \\
= & n \sum_i \int_0^s S_{PFS,0}(u)^2 \frac{I_{02,i}(u)}{Y_{02}(u)^2} dN_i^{OS}(u) + n \sum_i \int_0^s [S_{OS,0}(u) - S_{PFS,0}(u)]^2 \frac{I_{12,i}(u)}{Y_{12}(u)^2} dN_i^{OS}(u) \\
= & n \sum_{i}  \int_0^s \left( S_{PFS,0}(u)^2 \frac{I_{02,i}(u)}{Y_{02}(u)^2} + [S_{OS,0}(u) - S_{PFS,0}(u)]^2 \frac{I_{12,i}(u)}{Y_{12}(u)^2} \right) dN_i^{OS}(u). 	
\end{split}
\end{align*}
This follows with Eq. \eqref{eq:na_martingale_covariation}, as, for any $({\cal{F}}_{s})_{s \geq 0}$--martingale $M$ and predictable process $H$, the stochastic integral $\int H(u) dM(u)$ is an $({\cal{F}}_{s})_{s \geq 0}$--martingale with covariation process  $\int H(u)^2 d[M](u)$. All in all, under the contiguous alternatives $H_1$, $\Psi_{OS}(s)$ is the sum of an $({\cal{F}}_{s})_{s \geq 0}$--martingale with covariation process as in Eq. \eqref{eq03.04} and a deterministic drift $\mu_{OS}(s)$ as in Eq. \eqref{eq:os_drift_intersection}. 

Finally notice that $\Psi_{OS}(s)$ suffices the assumptions of Lemma \ref{lemma:rebolledo_condition_iii} and that $[\Psi_{OS}](s) \to \sigma_{OS}(s)^2$ in probability as $n \to \infty$ for some non--decreasing deterministic function $\sigma_{OS}(s)^2$ by a law of large numbers. Consequently, under the contiguous alternatives, $\Psi_{OS}(s)$ converges in distribution to a Gaussian process with independent increments, drift $\mu_{OS}(s)$, and variance function $\sigma_{OS}(s)^2$ as $n \to \infty$ by a central limit theorem for local martingales \citep{bib5}. In particular, $\Psi_{OS}(s)$ has zero drift when $H_0$ holds true. 

\subsection{Proof of Distributional Properties of $\Psi_{PFS}^{(\kappa)}(s)$ from Chapter \ref{Sec03.04}, Eq. \eqref{eq03.09}}

In this chapter we use the setting of a single-arm study. All patients receive the experimental therapy, i.e. the therapy group variable only takes the value $x=E$. We assume that condition \ref{condition:at_risk_process_general} holds for $y_0$ on $[0,s]$ and for $y_1$ on $[\varepsilon,s]$ for any $\varepsilon>0$ and hence we can apply Lemma \ref{lemma:non_empty_at_risk_set}.\\
Assume that the contiguous alternative $H_{1,n}$ defined in Eq. \eqref{eq:contiguous_alternatives_both} holds true.
Then:
\begin{align*}
\begin{split}
\Psi_{PFS}(s)&= \sqrt{n} \left(  \int_0^s S_{PFS,0}(u) d\hat{\Lambda}_{0*}(u) 						+ S_{PFS,0}(u) - 1				\right) \\
&= \sqrt{n} \left(  \int_0^s S_{PFS,0}(u) d[\hat{\Lambda}_{0*} - {\Lambda}_{0*}](u) 			+  \int_0^s S_{PFS,0}(u) d{\Lambda}_{0*}(u)			+ S_{PFS,0}(u) - 1				\right) \\
&= \sqrt{n} \left(  \int_0^s S_{PFS,0}(u) d[\hat{\Lambda}_{0*} - {\Lambda}_{0*}](u) 			+  \int_0^s S_{PFS,0}(u) d{\Lambda}_{0*}(u)			- \int_0^s S_{PFS,0}(u) d{\Lambda}_{0*,0}(u) 				\right) \\
&= \int_0^s S_{PFS,0}(u) dM_{0*}(u) 				+  \sqrt{n}  \int_0^s S_{PFS,0}(u) d[{\Lambda}_{0*}-{\Lambda}_{0*,0}](u) 			\\
&= \int_0^s S_{PFS,0}(u) dM_{0*}(u) 				+  \sqrt{n} \left( e^{-\frac{\gamma_{PFS}}{\sqrt{n}}} - 1 \right) \int_0^s S_{PFS,0}(u) d{\Lambda}_{0*,0}(u) \\
&= \int_0^s S_{PFS,0}(u) dM_{0*}(u) 				+  \sqrt{n} \left( e^{-\frac{\gamma{PFS}}{\sqrt{n}}} - 1 \right) \left( 1 - S_{PFS,0}(u) \right) 
\end{split}
\end{align*}
This shows that, under the contiguous alternatives $H_1$, $\Psi_{PFS}(s)$ is the sum of a $({\cal{F}}_{s})_{s \geq 0}$-martingale with covariation process $[\Psi_{PFS}](s)$ as given in Eq. \eqref{eq03.10} and a deterministic drift $\sqrt{n} \exp(\gamma_{PFS}/\sqrt{n} - 1) \cdot ( 1 - S_{PFS,0}(u)) $. 

Finally notice that $\Psi_{PFS}$ suffices the assumptions of Lemma \ref{lemma:rebolledo_condition_iii} and that $\sqrt{n} [ \exp(-\gamma_{PFS}/\sqrt{n}) - 1 ] \to \gamma_{PFS}$ (see Lemma \ref{lemma:taylor_expansions}) as $n \to \infty$, and that $[\Psi_{PFS}](s) \to \sigma_{PFS}(s)$ in probability as $n \to \infty$ for some non--decreasing deterministic function $\sigma_{PFS}(s)$ by a law of large numbers. Consequently, $\Psi_{PFS}(s)$ converges in distribution to a Gaussian process with independent increments, mean $\mu_{PFS}(s)$ as given in Eq. \eqref{eq03.50}, and variance function $\sigma_{PFS}(s)$ as $n \to \infty$ by a central limit theorem for local martingales \citep{bib5}. In particular, $\Psi_{PFS}(s)$ has zero drift when $H_0$ holds true. 

\subsection{Proof of Distributional Properties of $\Psi^{(\kappa)}(s)$  from Chapter \ref{Sec04.02}, Eq. \eqref{eq04.04}}

\subsubsection{Asymptotic distribution of $\Psi^{(\kappa)}$}

In this chapter we use the setting of a single-arm study. All patients receive the experimental therapy, i.e. the therapy group variable only takes the value $x=E$. We assume that condition \ref{condition:at_risk_process_general} holds for $y_0$ on $[0,s]$ and for $y_1$ on $[\varepsilon,s]$ for any $\varepsilon>0$ and hence we can apply Lemma \ref{lemma:non_empty_at_risk_set}.\\
Let us assume that the contiguous alternatives $H_1$ as defined in Eq. \eqref{eq:contiguous_alternatives_os} hold true, and let $\Delta_x(u)\coloneqq  S_{x,0}(u) - S_{x}(u)$ for $x \in \{PFS, OS\}$. Recalling the differential identity 
\begin{equation*}
	dM_{j2}(u) = \sqrt{n} \left( d\hat{\Lambda}_{j2}(u) - J_{j2}(u) d{\Lambda}_{j2}(u) \right)
\end{equation*}
from Eq. \eqref{eq:na_martingale}, we obtain: 
\begin{align*}
\begin{split}
\Psi(s)= & \sqrt{n} \left(  \int_0^s \hat{S}_{PFS}(u) d\hat{\Lambda}_{02}(u) 				+ \int_0^s [S_{OS,0}(u) - \hat{S}_{PFS}(u)] d\hat{\Lambda}_{12}(u) 				+ S_{OS,0}(u) - 1				\right) \\
= & \sqrt{n} \left(  \int_0^s S_{PFS}(u) J_{02}(u) d{\Lambda}_{02}(u) 				  + \int_0^s [S_{OS}(u) - S_{PFS}(u)] J_{12}(u) d{\Lambda}_{12}(u) 				+ S_{OS,0}(u) - 1	  	\right) \\
& + \sqrt{n} \int_0^s (S_{OS,0}(u) - S_{OS}(u)) J_{12}(u) d{\Lambda}_{12}(u)\\
& + \int_0^s S_{PFS}(u) dM_{02}(u) 														+  \int_0^s [S_{OS,0}(u) - S_{PFS}(u)] dM_{12}(u) \\
& + \sqrt{n} \cdot \int_0^s [ \hat{S}_{PFS}(u) - S_{PFS}(u) ] J_{02}(u) d\Lambda_{02}(u) 		-  \sqrt{n} \cdot \int_0^s [ \hat{S}_{PFS}(u) - S_{PFS}(u) ] J_{12}(u) d\Lambda_{12}(u) \\
& + \int_0^s [ \hat{S}_{PFS}(u) - S_{PFS}(u) ] dM_{02}(u) 		-  \int_0^s [ \hat{S}_{PFS}(u) - S_{PFS}(u) ] dM_{12}(u) 
\end{split}
\end{align*}
with $\hat{S}_{PFS}(u) \coloneqq  \exp(-\hat{\Lambda}_{0*}(u-))$. We now evaluate the four summands in the last equation above.
\\
\textsc{The First and Second Summand:} We can exploit condition \ref{condition:at_risk_process_general} and Lemma \ref{lemma:taylor_expansions} to see that this is asymptotically equivalent to the deterministic drift term
\begin{equation*}
	S_{OS,0}(s) \log(S_{OS,0}(s)) \cdot \gamma_{OS} - \gamma_{OS} \int_0^s S_{OS,0}(u)\log(S_{OS,0}(u)) d\Lambda_{12}(u).
\end{equation*}
\\
\textsc{The Third Summand:} The second summand is a mean-zero $({\cal{F}}_{s})_{s \geq 0}$--martingale with covariation process 
\begin{align*}
\begin{split}
& \int_0^s S_{PFS}(u)^2 d[M_{02}](u)  + 2 \cdot \int_0^s S_{PFS}(u)(S_{OS,0}(u) - S_{PFS}(u)) d[M_{02},M_{12}](u) 		\\ 
& \qquad + \int_0^s (S_{OS,0}(u) - S_{PFS}(u))^2 d[M_{12}](u) \\
&= n \sum_i \int_0^s S_{PFS}(u)^2 \frac{I_{02,i}(u)}{Y_{02}(u)^2} dN_i^{(OS)}(u) + n \sum_i \int_0^s [S_{OS,0}(u) - S_{PFS}(u)]^2 \frac{I_{12,i}(u)}{Y_{12}(u)^2} dN_i^{(OS)}(u) \\
&= n \sum_{i}  \int_0^s \left( S_{PFS}(u)^2 \frac{I_{02,i}(u)}{Y_{02}(u)^2} + [S_{OS,0}(u) - S_{PFS}(u)]^2 \frac{I_{12,i}(u)}{Y_{12}(u)^2} \right) dN_i^{OS}(u) 	
\end{split}
\end{align*}
This follows with Eq. \eqref{eq:na_martingale_covariation}, as, for any $({\cal{F}}_{s})_{s \geq 0}$--martingale $M$ and predictable process $H$, the stochastic integral $\int H(u) dM(u)$ is an $({\cal{F}}_{s})_{s \geq 0}$--martingale with covariation process  $\int H(u)^2 d[M](u)$. 
\\
\textsc{The Fifth Summand:} We claim that the fourth summand vanishes in probability as $n \to \infty$.
To see this, we use that $\sup_{u \in [0,s]} n^{k} |\hat{\Lambda}_{0*}(u-) - {\Lambda}_{0*}(u)| \to 0$ in probability as $n \to \infty$ by Lemma \ref{lemma:na_convergence}. It then follows that $\sup_{u \in [0,s]} |\hat{S}_{PFS}(u-) - {S}_{PFS}(u)| \to 0$ in probability as $n \to \infty$, because $|e^{-x}-e^{-x'}|\leq |x-x'|$ for all $x,x' \geq 0$. By a central limit theorem for local martingales \citep{bib5}, finally also notice that $M_{j2}$ converges in distribution to a Gaussian process as $n \to \infty$ and the covariation process $[M_{j2}]$ converges to some non-decreasing deterministic variance function as $n \to \infty$ by a law of large numbers. From appeal to Slutsky's theorem, it thus follows that the fifth summand vanishes in probability as $n \to \infty$.
\\
\textsc{The Fourth Summand:} We first introduce the abbreviation
\begin{align*}
\begin{split}
F_j(s) = \int_0^s \exp(-\Lambda_{0*}(u)) d\Lambda_{j2}(u), \qquad \text{for } j=0,1.
\end{split}
\end{align*}
The sum $\sqrt{n}|\int_0^s [ \hat{S}_{PFS}(u) - S_{PFS}(u) ] d\Lambda_{j2}(u) + \int_0^s [ \hat{\Lambda}_{0*}(u-) - \Lambda_{0*}(u) ] dF_j(u)|$ vanishes in probability as $n \to \infty$ by Lemma \ref{lemma:ls_integral_reformulation}. Now, integration by parts for semimartingales implies that
\begin{align*}
\begin{split}
\int_0^s M_{0*}(u) dF_j(u) = M_{0*}(s) F_j(s) -  \int_0^s F_j(u) dM_{0*}(u),
\end{split}
\end{align*}
since $F_j(u)$ is continuous and deterministic, and thus of finite variation. In conclusion, we obtain that
\begin{align*}
\begin{split}
\sqrt{n}\int_0^s [ \hat{S}_{PFS}(u) - S_{PFS}(u) ] d\Lambda_{j2}(u) = - M_{0*}(s) F_j(s) + \int_0^s F_j(u) dM_{0*}(u) + o_P(1).
\end{split}
\end{align*}
With condition \ref{condition:at_risk_process_general} this in turn also implies that
\begin{align*}
	\begin{split}
		\sqrt{n}\int_0^s [ \hat{S}_{PFS}(u) - S_{PFS}(u) ] J_{j2}(u) d\Lambda_{j2}(u) = - M_{0*}(s) F_j(s) + \int_0^s F_j(u) dM_{0*}(u) + o_P(1).
	\end{split}
\end{align*}
So, up to terms that vanish in probability as $n \to \infty$, we thus have
\begin{align*}
\begin{split}
\Psi(s)&= 
\gamma_{OS} \cdot S_{OS,0}(s) \log(S_{OS,0}(s))   \\
& + \int_0^s S_{PFS}(u) dM_{02}(u) 														+  \int_0^s [S_{OS,0}(u) - S_{PFS}(u)] dM_{12}(u) \\
& + M_{0*}(s) \cdot [F_{12}(s) - F_{02}(s)]     +   \int_0^s [F_{02}(u) - F_{12}(u)] dM_{0*}(u) 
\end{split}
\end{align*}
In particular, $\Psi(s)$ is zero in expectation when $H_0: \gamma_{OS} = 0$ holds true, because $M_{02}, M_{12}, M_{0*}$ are mean-zero martingales. To elaborate the asymptotic distribution of $\Psi(s)$ as $n \to \infty$, we first elaborate the asymptotic distribution of the multivariate $({\cal{F}}_{s})_{s \geq 0}$-martingale
\begin{equation*}
\Theta(s) \coloneqq 
\begin{pmatrix} \int_0^s S_{PFS}(u) dM_{02}(u) \\ \int_0^s [S_{OS,0}(u) - S_{PFS}(u)] dM_{12}(u) \\  M_{0*}(s)   \\    \int_0^s [F_{02}(u) - F_{12}(u)] dM_{0*}(u)   \end{pmatrix}
\end{equation*}
Its covariation matrix is
\begin{equation*}
[\Theta](s) \coloneqq 
\begin{pmatrix} [\Theta]_{11}(s) & 0  & [\Theta]_{13}(s) & [\Theta]_{14}(s) \\ 0 & [\Theta]_{22}(s) & 0 & 0 \\ [\Theta]_{31}(s) & 0  & [\Theta]_{33}(s) & [\Theta]_{34}(s)  \\    [\Theta]_{41}(s) & 0  & [\Theta]_{43}(s) & [\Theta]_{44}(s)   
\end{pmatrix}
\end{equation*}
with non-zero components
\begin{align*}
\begin{split}
[\Theta]_{11}(s) 													&= n \cdot \sum_{i=1}^n \int_0^s S_{PFS}(u)^2 													\frac{I_{02,i}(u)}{Y_{02}(u)^2} dN_i^{OS}(u)  \\
[\Theta]_{13}(s) \equiv [\Theta]_{31}(s)  &= n \cdot \sum_{i=1}^n \int_0^s S_{PFS}(u) 														\frac{I_{02,i}(u)}{Y_{02}(u)^2} dN_i^{OS}(u)  \\
[\Theta]_{14}(s) \equiv [\Theta]_{41}(s)  &= n \cdot \sum_{i=1}^n \int_0^s S_{PFS}(u) \cdot [F_{02}(u)-F_{12}(u)]	\frac{I_{02,i}(u)}{Y_{02}(u)^2} dN_i^{OS}(u)  \\
[\Theta]_{22}(s) 													&= n \cdot \sum_{i=1}^n \int_0^s [S_{OS,0}(u)-S_{PFS}(u)]^2 							\frac{I_{12,i}(u)}{Y_{12}(u)^2} dN_i^{OS}(u)  \\
[\Theta]_{33}(s) 													&= n \cdot \sum_{i=1}^n \int_0^s                          							\frac{I_{02,i}(u)}{Y_{02}(u)^2} dN_i^{PFS}(u) \\
[\Theta]_{34}(s) \equiv [\Theta]_{43}(s)  &= n \cdot \sum_{i=1}^n \int_0^s [F_{02}(u)-F_{12}(u)] 	    						\frac{I_{02,i}(u)}{Y_{02}(u)^2} dN_i^{PFS}(u) \\
[\Theta]_{44}(s) 													&= n \cdot \sum_{i=1}^n \int_0^s [F_{02}(u)-F_{12}(u)]^2  							\frac{I_{02,i}(u)}{Y_{02}(u)^2} dN_i^{PFS}(u). 
\end{split}
\end{align*}
It converges in probability as $n \to \infty$ to a deterministic covariance matrix $V(s)$ by a law of large numbers. Moreover, $\Theta(s)$ suffices the assumptions of Lemma \ref{lemma:rebolledo_condition_iii}. By a central limit theorem for local martingales \citep{bib5}, in the large sample limit, $\Theta(s)$ thus converges in distribution to a mean zero, $({\cal{F}}_{s})_{s \geq 0}$-adapted Gaussian process with independent increments, and variance function $V(s)$.

\subsubsection{Estimation of the variance structure of $\Psi^{(\kappa)}(s)$}

Since the matrix $[\Theta]$ depends on the unknown nuisance parameters $S_{PFS}(s)$, $F_{02}(s)$ and $F_{12}(s)$, it cannot be used directly as estimator for $V(s)$. Instead, we use estimators $\hat{S}_{PFS}(s)\coloneqq \exp(-\hat{\Lambda}_{0*}(s-))$ and $\hat{F}_{j2}(s)\coloneqq  \int_0^s \exp[- \hat{\Lambda}_{0*}(u-)] d \hat{\Lambda}_{j2}(u)$ for $S_{PFS}(s)$ and $F_{j2}(s)$, respectively, to introduce the matrix
\begin{equation*}
\hat{[\Theta]}(s) \coloneqq 
\begin{pmatrix} \hat{[\Theta]}_{11}(s) & 0  & \hat{[\Theta]}_{13}(s) & \hat{[\Theta]}_{14}(s) \\ 0 & \hat{[\Theta]}_{22}(s) & 0 & 0 \\ \hat{[\Theta]}_{31}(s) & 0  & \hat{[\Theta]}_{33}(s) & \hat{[\Theta]}_{34}(s)  \\    \hat{[\Theta]}_{41}(s) & 0  & \hat{[\Theta]}_{43}(s) & \hat{[\Theta]}_{44}(s)   
\end{pmatrix}
\end{equation*}
with non-zero components
\begin{align}\label{eq:covariance_matrix_estimate}
\begin{split}
\hat{[\Theta]}_{11}(s) 													&= n \cdot \sum_{i=1}^n \int_0^s \hat{S}_{PFS}(u)^2 													\frac{I_{02,i}(u)}{Y_{02}(u)^2} dN_i^{OS}(u)  \\
\hat{[\Theta]}_{13}(s) \equiv \hat{[\Theta]}_{31}(s)  &= n \cdot \sum_{i=1}^n \int_0^s \hat{S}_{PFS}(u) 														\frac{I_{02,i}(u)}{Y_{02}(u)^2} dN_i^{OS}(u)  \\
\hat{[\Theta]}_{14}(s) \equiv \hat{[\Theta]}_{41}(s)  &= n \cdot \sum_{i=1}^n \int_0^s \hat{S}_{PFS}(u) \cdot [\hat{F}_{02}(u)-\hat{F}_{12}(u)]	\frac{I_{02,i}(u)}{Y_{02}(u)^2} dN_i^{OS}(u)  \\
\hat{[\Theta]}_{22}(s) 													&= n \cdot \sum_{i=1}^n \int_0^s [{S}_{OS,0}(u)-\hat{S}_{PFS}(u)]^2 							\frac{I_{12,i}(u)}{Y_{12}(u)^2} dN_i^{OS}(u)  \\
\hat{[\Theta]}_{33}(s) 													&= n \cdot \sum_{i=1}^n \int_0^s                          							\frac{I_{02,i}(u)}{Y_{02}(u)^2} dN_i^{PFS}(u) \\
\hat{[\Theta]}_{34}(s) \equiv \hat{[\Theta]}_{43}(s)  &= n \cdot \sum_{i=1}^n \int_0^s [\hat{F}_{02}(u)-\hat{F}_{12}(u)] 	    						\frac{I_{02,i}(u)}{Y_{02}(u)^2} dN_i^{PFS}(u) \\
\hat{[\Theta]}_{44}(s) 													&= n \cdot \sum_{i=1}^n \int_0^s [\hat{F}_{02}(u)-\hat{F}_{12}(u)]^2  							\frac{I_{02,i}(u)}{Y_{02}(u)^2} dN_i^{PFS}(u) 
\end{split}
\end{align}
which we use as estimator $\hat{V}(s)\coloneqq \hat{[\Theta]}(s)$ for $V(s)$. Notice that ${S}_{OS,0}(u)$ is fixed by the null hypothesis and thus known. We want to apply Lemma \ref{lemma:conv_integral_product_estimators} to show that $\hat{[\Theta]}(s) - [\Theta](s)$ vanishes in probability. For the factor $\hat{S}_{PFS}$ the prerequisites of Lemma \ref{lemma:conv_integral_product_estimators} hold. This is also true for $\hat{F}_{02}$ and $\hat{F}_{12}$ according to Lemma \ref{lemma:F_convergence}. The integrators of the corresponding intervals are given by estimators of variances of Nelson-Aalen estimators. Their convergence is given in Theorem IV.1.2 of \cite{Andersen:1991}.\\
We may summarize as follows:
\begin{align*}
\begin{split}
&\bullet \Theta(s) \underset{n\to \infty}{\overset{\mathcal{D}}{\longrightarrow}} {\cal{N}}(0,V(s)) \quad \text{for each } s \geq 0 \\
&\bullet \hat{[\Theta]}(s) \underset{n\to \infty}{\overset{\mathbb{P}}{\longrightarrow}} V(s) \quad \text{for each } s \geq 0, \text{i.e. } \hat{[\Theta]}(s) \text{ consistently estimates } V(s) \\
&\bullet \Theta(s_1) \text{ and } \Theta(s_2) - \Theta(s_1) \quad \text{are approximately independent for each } 0 \leq s_1 \leq s_2 \\
\end{split}
\end{align*}
In particular, this characterizes the distributional properties of $\Psi(s)$ in the large sample limit, because $\Psi(s) = \mu(s) + L(s) \cdot \Theta(s) + o_P(1)$ with deterministic drift $\mu(s)$ and scale factor $L(s)$ as specified in Eqs. \eqref{A6.eq102} and \eqref{A6.eq103} below.

\subsubsection{Application in adaptive clinical trials}

In single-arm adaptive trials, patients are recruited sequentially in consecutive stages, and assigned to experimental treatment $E$. With a view to application to adaptive hypothesis testing in single-arm trials, we now reintroduce the index $(E \kappa)$ indicating treatment group and stage in all expressions, thus returning to our initial notation as introduced in Chapters \ref{Sec02} and \ref{sec:single_arm_os}. 
In the subset of patients from stage $\kappa$ we therefore calculate the statistic
\begin{align*}
\Psi^{(\kappa)}(s)\coloneqq  \sqrt{n^{(E \kappa)}} \left(   
										\int_0^s \hat{S}_{PFS}^{(E \kappa)}(u) d\hat{\Lambda}_{02}^{(E \kappa)}(u) 
										+ \int_0^s [S_{OS,0}(u) - \hat{S}_{PFS}^{(E \kappa)}(u)] d\hat{\Lambda}_{12}^{(E \kappa)}(u) 
										+ S_{OS,0}(u) - 1
										\right).
\end{align*}
As shown above, we have
\begin{align}\label{A6.eq101}
\Psi^{(\kappa)}(s) = \mu(s) + L(s) \cdot \Theta^{(E \kappa)}(s) + o_P(1)
\end{align}
with drift term
\begin{align}\label{A6.eq102}
\mu(s)\coloneqq  \gamma_{OS} \cdot S_{OS,0}(s) \log(S_{OS,0}(s)) - \gamma_{OS} \int_0^s S_{OS,0}(u)\log(S_{OS,0}(u)) d\Lambda_{12}(u),
\end{align}
scaling factor
\begin{align}\label{A6.eq103}
L(s)\coloneqq (1, 1, F_{12}(s) - F_{02}(s), 1), 
\quad {F}_{j2}(s)\coloneqq  \int_0^s \exp[- {\Lambda}_{0*}^{(E)}(u)] d {\Lambda}_{j2}^{(E)}(u) 
\end{align}
and stage-wise multivariate $({\cal{F}}_s)_{s \geq 0}$-martingales
\begin{equation*}
\Theta^{(E \kappa)}(s) \coloneqq 
\begin{pmatrix} \int_0^s S_{PFS}^{(E)}(u) dM_{02}^{(E \kappa)}(u) \\ \int_0^s [S_{OS,0}(u) - S_{PFS}^{(E)}(u)] dM_{12}^{(E \kappa)}(u) \\  M_{0*}^{(E \kappa)}(s)   \\    \int_0^s [F_{02}(u) - F_{12}(u)] dM_{0*}^{(E \kappa)}(u)   \end{pmatrix}
\end{equation*}
that converges in distribution to a mean-zero Gaussian process with independent increments and covariance function $V(s)$, say. Since patients from different stages have identically distributed survival, this limiting covariance function $V(s)$ does not depend on the actual stage $\kappa$. At each stage of the design, the unobservable quantities ${S}_{PFS}(s)$, $F_{02}(u)$, $F_{12}(u)$, $L(s)$, $\Theta^{(E \kappa)}(s)$ and $V(s)$ may be estimated from the data via
\begin{align*}
\begin{split}
\hat{S}_{PFS}^{(E \kappa)}(s) &\coloneqq  \exp[- \hat{\Lambda}_{0*}^{(E  \kappa)}(s-)] \\
\hat{F}_{j2}^{(E  \kappa)}(s) &\coloneqq  \int_0^s \exp[- \hat{\Lambda}_{0*}^{(E  \kappa)}(u-)] d \hat{\Lambda}_{j2}^{(E  \kappa)}(u) \\
\hat{L}^{(E  \kappa)}(s)			&\coloneqq (1, 1, \hat{F}_{12}^{(E  \kappa)}(s) - \hat{F}_{02}^{(E  \kappa)}(s), 1) \\
\hat{\Theta}^{(E \kappa)}(s) 				&\coloneqq 
\begin{pmatrix} \int_0^s \hat{S}_{PFS}^{(E)}(u) dM_{02}^{(E \kappa)}(u) \\ \int_0^s [S_{OS,0}(u) - \hat{S}_{PFS}^{(E)}(u)] dM_{12}^{(E \kappa)}(u) \\  M_{0*}^{(E \kappa)}(s)   \\    \int_0^s [\hat{F}_{02}(u) - \hat{F}_{12}(u)] dM_{0*}^{(E \kappa)}(u)   \end{pmatrix} \\
\hat{V}^{(E \kappa)}(s)  &\coloneqq  [\hat{\Theta}^{(E \kappa)}](s)
\end{split}
\end{align*}
where $[\hat{\Theta}^{(E \kappa)}](s) =  \left(  [\hat{\Theta}^{(E \kappa)}]_{ij}(s) \right)_{i,j=1,\ldots,4} $ is the covariation matrix of the $({\cal{F}}_s)_{s \geq 0}$-martingale $\Theta^{(E \kappa)}(s)$ where unknown nuisance parameters are replaced by their estimates. Its components are the same as those in \eqref{eq:covariance_matrix_estimate}, where the variable quantities are replaced by their stage and group specific counterparts. This amounts to the non-zero entries shown in \eqref{eq:covariance_matrix_estimate_perstage}.\\
In this context we consider testing the null hypothesis of no drift as compared to the historic control: 
\begin{align}\label{eq:null_hyp_os}
\begin{split}
H_0: S_{OS}^{(E)}(s) = S_{OS,0}(s) \quad \Longleftrightarrow \quad H_0: \mu(s) = 0 \quad \text{for all } s \geq 0.
\end{split}
\end{align}
For testing $H_0$ we use the stage-wise statistics $\Psi^{(\kappa)}(s)$. For constructing an adaptive test of $H_0$ we will in fact need knowledge of the asymptotical joint distribution of $\Psi^{(\kappa)}(s)$ and $M_{0*}^{(E \kappa)}(s)$. We have
\begin{equation}\label{A6.eq108}
\begin{pmatrix} \Psi^{(\kappa)}(s) \\ M_{0*}^{(E \kappa)}(s) 
\end{pmatrix}
\sim
{\cal{N}}
\left(
\begin{pmatrix} \mu(s) \\ 0 
\end{pmatrix}
,
\begin{pmatrix} L(s) \cdot V(s) \cdot L(s)^{\textsc{T}} & L(s) \cdot V_{\cdot 3}(s) \\ L(s) \cdot V_{\cdot 3}(s) & V_{33}(s)
\end{pmatrix}
\right)
\end{equation}
where $V_{\cdot 3}(s)$ is the third column of the matrix $V(s)$, and $V_{33}(s)$ is the entry in row three and column three of $V(s)$. In particular,
\begin{equation*}
\Psi^{(\kappa)}(s)\sim{\cal{N}}\left(\mu(s),\sigma(s)^2\right) \quad \text{with } \sigma(s)^2\coloneqq  L(s) \cdot V(s) \cdot L(s)^{\textsc{T}}.
\end{equation*}
Finally, for any $0 < s_1 < s_2$, we will also consider the scaled increments of $\Theta^{(E \kappa)}$ given by
\begin{equation*}
\tilde{\Psi}^{(\kappa)}(s_2,s_1) \coloneqq  L(s_2) \cdot [\Theta^{(E \kappa)}(s_2) - \Theta^{(E \kappa)}(s_1)]. 
\end{equation*}
Since $(\Theta^{(E \kappa)}(s))_{s \geq 0}$ and $(M_{0*}^{(E \kappa)}(s))_{s \geq 0}$ are both $({\cal{F}}_s)_{s \geq 0}$-martingales with asymptotically independent increments, the random variable $\tilde{\Psi}^{(\kappa)}(s_2,s_1)$ has the following properties in the large sample limit:
\begin{align*}
\begin{split}
&\bullet \quad \tilde{\Psi}^{(\kappa)}(s_2,s_1) \sim {\cal{N}}\left( 0 , {\varsigma}(s_2,s_1)^2 \right), \quad \text{with } {\varsigma}(s_2,s_1)^2\coloneqq  L(s_2) \cdot [V(s_2) - V(s_1)] \cdot L(s_2)^{\textsc{T}} \\
&\bullet \quad \tilde{\Psi}^{(\kappa)}(s_2,s_1) \text{ is independent from the random vector } (\Psi^{(\kappa)}(s_1),M_{0*}^{(E \kappa)}(s_1)).
\end{split}
\end{align*}
At each stage $\kappa$ of the design, the variance parameters ${\sigma}(s)^2 = {\varsigma}(s,0)^2$ and ${\varsigma}(s_2,s_1)^2$ may be estimated consistently via
\begin{align}\label{eq:covariance_estimate}
\begin{split}
\hat{\sigma}^{(E \kappa)}(s)^2 						&\coloneqq  \hat{L}^{(E \kappa)}(s)   \cdot \hat{V}^{(E \kappa)}(s) \cdot \hat{L}^{(E \kappa)}(s)^{\textsc{T}} \\
\hat{\varsigma}^{(E \kappa)}(s_2,s_1)^2 	&\coloneqq  \hat{L}^{(E \kappa)}(s_2) \cdot [\hat{V}^{(E \kappa)}(s_2) - \hat{V}^{(E \kappa)}(s_1)] \cdot \hat{L}^{(E \kappa)}(s_2)^{\textsc{T}}. 
\end{split}
\end{align}
\ \\ \\ 
We next explicitly construct a two-stage adaptive test of $H_0$ as in \eqref{eq:null_hyp_os} in the context of a single-arm clinical trial for comparing the survival under experimental treatment $E$ to a prefixed reference survival curve. Let $\Psi^{(1)}(s_1)$ denote the statistic from Eq. \eqref{A6.eq101} calculated in the set of stage 1 patients at some early point of time $s_1 > 0$ that is characteristic for the evaluation of short-term response to treatment, e.g. $s_1 = 6$ months after the start of treatment. Let $Z_{11}\coloneqq  \Psi^{(1)}(s_1) / \sigma^{(E 1)}(s_1)$ denote the corresponding standardized statistic. Likewise, let $\hat{\Lambda}_{0*}^{(E 1)}(s_1)$ denote the Nelsen-Aalen estimate of the cumulative short-term hazard $\Lambda_{0*}^{(E)}(s_1)$ for PFS under the experimental treatment, derived in stage 1 patients. Let further $S_2\coloneqq \xi(Z_{11}, \hat{\Lambda}_{0*}^{(E 1)}(s_1)) > s_1$ be some long-term point of time that is chosen at the interim analysis in dependence of the interim estimates $Z_{11}$ and $\hat{\Lambda}_{0*}^{(E 1)}(s_1)$ by means of a continuous function $\xi: \mathbb{R} \times [0,\infty) \to (s_1,\infty)$, $(x_1,x_2) \mapsto \xi(x_1,x_2)$. Let $\Psi^{(j)}(S_2)$ denote the statistic from Eq. \eqref{A6.eq101} evaluated in the set of patients from stage $j=1,2$ at the random long-term point of time $S_2 > s_1$.

On this basis consider the standardized statistics
\begin{align*}
\begin{split}
\hat{Z}_{11} &\coloneqq  \frac{\Psi^{(1)}(s_1)}{\hat{\sigma}^{(E1)}(s_1)} \\
\hat{Z}_{12} &\coloneqq  \frac{\Psi^{(1)}(S_2) - \Psi^{(1)}(s_1)}{\hat{{\varsigma}}^{(E1)}(S_2,s_1)} \\
\hat{Z}_{2}  &\coloneqq  \frac{\Psi^{(2)}(S_2)}{\hat{\sigma}^{(E2)}(S_2)} 
\end{split}
\end{align*}
For some prefixed weights $0<w_1,w_2<1$ with $w_1^2+w_2^2=1$, we are interested in calculating the probability of the event
\begin{align*}
\begin{split}
\hat{\cal{R}} \coloneqq  \{ \hat{Z}_{11} \geq c_1 \} \cup \{    c_0 \leq \hat{Z}_{11} < c_1 ,   w_1 \hat{Z}_{12} + w_2 \hat{Z}_{2} \geq c_2 \}.
\end{split}
\end{align*}
For this purpose, let us also introduce the auxiliary statistics
\begin{align*}
\begin{split}
Z_{11} \coloneqq  \frac{\Psi^{(1)}(s_1)}{\sigma(s_1)}, \quad
Z_{12} \coloneqq  \frac{\Psi^{(1)}(S_2) - \Psi^{(1)}(s_1)}{{\varsigma}(S_2,s_1)}, \quad
Z_{2}  \coloneqq  \frac{\Psi^{(2)}(S_2)}{\sigma(S_2)} 
\end{split}
\end{align*}
as well as the event
\begin{align*}
\begin{split}
{\cal{R}} \coloneqq  \{ Z_{11} \geq c_1 \} \cup \{    c_0 \leq Z_{11} < c_1 ,   w_1 Z_{12} + w_2 Z_{2} \geq c_2 \}.
\end{split}
\end{align*}
Notice that $(\hat{Z}_{11}, \hat{Z}_{12}, \hat{Z}_{2})$ has the same distribution as $({Z}_{11}, {Z}_{12}, {Z}_{2})$ in the large sample limit by Slutsky's theorem. Thus $\mathbb{P}(\hat{\cal{R}}) \approx \mathbb{P}({\cal{R}})$ when sample size is large, and we are reduced to calculate $\mathbb{P}({\cal{R}})$. Clearly, we have
\begin{align}\label{A6.eq117}
\begin{split}
\mathbb{P} \left( \{ Z_{11} \geq c_1 \} \right) = \mathbb{P} \left( \frac{\Psi^{(1)}(s_1) - \mu(s_1)}{\sigma(s_1)} \geq c_1 - \frac{\mu(s_1)}{\sigma(s_1)} \right) = 1 - \Phi\left( c_1 - \frac{\mu(s_1)}{\sigma(s_1)} \right).
\end{split}
\end{align}
To evaluate $\mathbb{P}( c_0 \geq Z_{11} < c_1 ,   w_1 Z_{12} + w_2 Z_{2} \geq c_2 )$ first notice that with $\Delta F(s)\coloneqq  F_{12}(s) - F_{02}(s)$ we have
\begin{align*}
\begin{split}
&\Psi^{(1)}(S_2) - \Psi^{(1)}(s_1) - [\mu(S_2) - \mu(s_1) ] \\
&= (1,1,0,1)^T \cdot [\Theta^{(E1)}(S_2) - \Theta^{(E1)}(s_1)] + \Delta F(S_2) \cdot M_{0*}^{(E1)}(S_2) - \Delta F(s_1) \cdot M_{0*}^{(E1)}(s_1) \\
&= (1,1,0,1)^T \cdot [\Theta^{(E1)}(S_2) - \Theta^{(E1)}(s_1)] + \Delta F(S_2) \cdot [ M_{0*}^{(E1)}(S_2) + M_{0*}^{(E1)}(s_1)] - [\Delta F(S_2) - \Delta F(s_1) ] \cdot M_{0*}^{(E1)}(s_1) \\
&= L(S_2) \cdot [\Theta^{(E1)}(S_2) - \Theta^{(E1)}(s_1)] +  [\Delta F(S_2) - \Delta F(s_1) ] \cdot M_{0*}^{(E1)}(s_1) \\
&= \tilde{\Psi}^{(1)}(S_2,s_1) +  [\Delta F(S_2) - \Delta F(s_1) ] \cdot M_{0*}^{(E1)}(s_1).
\end{split}
\end{align*}
On the set where $J_{02}^{(E1)}(s_1)=1$, we have $M_{0*}(s)/\sqrt{n^{(E1)}} = \hat{\Lambda}_{0*}^{(E1)}(s) - {\Lambda}_{0*}^{(E)}(s)$ and thus $S_2 \equiv \xi(Z_{11}, M_{0*}(s_1)/\sqrt{n^{(E1)}}(s_1) - {\Lambda}_{0*}^{(E)}(s_1)) $. Also let 
\begin{align*}
\begin{split}
s_2(z,m)\coloneqq  
\xi(x_1,x_2)|_{x_1=z,x_2=\frac{m}{\sqrt{n^{(E1)}}(s_1)} - {\Lambda}_{0*}^{(E)}(s_1)}
\equiv\xi\left( z,  \frac{m}{\sqrt{n^{(E1)}}(s_1)} - {\Lambda}_{0*}^{(E)}(s_1) \right).
\end{split}
\end{align*}
Thus, by independence of $\tilde{\Psi}^{(1)}(s_2(Z_11, M_{0*}(s_1)/\sqrt{n^{(E1)}}(s_1) - {\Lambda}_{0*}^{(E)}(s_1)), s_1)$ from $(\Psi^{(1)}(s_1),M_{0*}^{(E1)}(s_1))$ and since $\lim_{n^{(E1)} \to \infty} \mathbb{P}(J_{02}^{(E1)}(s_1)=1) = 1$ we approximately have for sufficiently large $n^{(E1)}$ and $n^{(E2)}$: 
\begin{align*}
\begin{split}
&\mathbb{P}\left( c_0 \leq Z_{11} < c_1, w_1 Z_{12} + w_2 Z_{2} \geq c_2 \right) \\
&= 	\int_{c_0}^{c_1} dz \int_{-\infty}^{\infty} dm f_{Z_{11},M_{0*}^{(E1)}(s_1)}(z,m) \\
& \qquad \qquad		\mathbb{P}\left( w_1 \frac{\Psi^{(1)}(S_2) - \Psi^{(1)}(s_1)}{ {\varsigma}(S_2,s_1)}  + w_2 \frac{\Psi^{(2)}(S_2)}{\sigma(S_2)} \geq c_2 \Big| Z_{11} = z , M_{0*}^{(E1)}(s_1) = m \right) \\
&= 	\int_{c_0}^{c_1} dz \int_{-\infty}^{\infty} dm f_{Z_{11},M_{0*}^{(E1)}(s_1)}(z,m) \\
& \qquad \qquad		
	\mathbb{P}\Bigg( w_1 \frac{\tilde{\Psi}^{(1)}(s_2(z,m),s_1) }{ {\varsigma}(s_2(z,m),s_1)} 
			  + w_2 \frac{\Psi^{(2)}(s_2(z,m)) - \mu(s_2(z,m))}{\sigma(s_2(z,m))} \\
				& \qquad  \qquad  + w_2 \frac{\mu(s_2(z,m))}{\sigma(s_2(z,m))} 
				+ w_1 \frac{[\Delta F(s_2(z,m)) - \Delta F(s_1) ] \cdot m + [\mu(s_2(z,m))-\mu(s_1)]}{ {\varsigma}(s_2(z,m),s_1)}   
				\geq c_2 \Bigg) \\
&= 	\int_{c_0}^{c_1} dz \int_{-\infty}^{\infty} dm f_{Z_{11},M_{0*}(s_1)}(z,m) \\
& \qquad \qquad		\Phi\left( -c_2 + w_2 \frac{\mu(s_2(z,m))}{\sigma(s_2(z,m))} + w_1 \frac{[\Delta F(s_2(z,m)) - \Delta F(s_1) ] \cdot m + [\mu(s_2(z,m))-\mu(s_1)]}{ {\varsigma}(s_2(z,m),s_1)} \right)
\end{split}
\end{align*}
where $f_{Z_{11},M_{0*}^{(E1)}(s_1)}(z,m)$ is the joint density of the vector $(Z_{11},M_{0*}^{(E1)}(s_1))$ which approximately follows the bivariate normal distribution given in \eqref{A6.eq108}.

To derive an approximate and asymptocially valid level-$\alpha$ adaptive test for $H_0:S_{OS}^{(E)}(s)=S_{OS,0}(s)$ based on the rejection region $\hat{\cal{R}}$, we choose the critical bounds $c_1,c_2$ such that the stage-1 type one error rate $\mathbb{P}_{H_0}(\{ Z_{11} \geq c_1 \})$ equals $\alpha_1$ and such that the overall type one error rate $\mathbb{P}_{H_0}(\hat{\cal{R}}) \approx \mathbb{P}_{H_0}({\cal{R}})$ equals the desired significance level $\alpha>\alpha_1$. Since $\mu(\cdot)=0$ under $H_0$, we immediately obtain from \eqref{A6.eq117} that $c_1=\Phi^{-1}(1-\alpha_1)$. We then choose $c_2$ such that $\mathbb{P}_{H_0}( c_0 \leq Z_{11} < c_1, w_1 Z_{12} + w_2 Z_{2} \geq c_2)$ equals $\alpha - \alpha_1$. This means choosing $c_2$ such that
\begin{align*}
\begin{split}
&\alpha - \alpha_1  
= 	\int_{c_0}^{c_1} dz \int_{-\infty}^{\infty} dm f(z,m) \Phi\left( -c_2 + w_1 \frac{[\Delta F(s_2(z,m)) - \Delta F(s_1) ] \cdot m ]}{ {\varsigma}(s_2(z,m),s_1)} \right), \\ 
&\text{where } f(z,m) \text{ is the density of } \begin{pmatrix} Z \\ M
\end{pmatrix}
\sim
{\cal{N}}
\left(
\begin{pmatrix} 0 \\ 0
\end{pmatrix}
,
\begin{pmatrix} 1 & \frac{{L}(s_1) \cdot {V}_{\cdot 3}(s_1)}{{\sigma}(s_1)}   \\ \frac{{L}(s_1) \cdot {V}_{\cdot 3}(s_1)}{{\sigma}(s_1)}  & {V}_{33}(s_1)
\end{pmatrix}
\right)
\end{split}
\end{align*}
what, however, cannot be directly put into practice as it requires knowledge of various nuisance parameters. Instead we determine $\hat{c}_2$ such that
\begin{align*}
\begin{split}
&\alpha - \alpha_1  
= 	\int_{c_0}^{c_1} dz \int_{-\infty}^{\infty} dm f(z,m) \Phi\left( -\hat{c}_2 + w_1 \frac{[\Delta \hat{F}(\hat{s}_2(z,m)) - \Delta \hat{F}(s_1) ] \cdot m ]}{ \hat{{\varsigma}}(\hat{s}_2(z,m),s_1)} \right), \\ 
&\text{where } f(z,m) \text{ is the density of } \begin{pmatrix} Z \\ M 
\end{pmatrix}
\sim
{\cal{N}}
\left(
\begin{pmatrix} 0 \\ 0
\end{pmatrix}
,
\begin{pmatrix} 1 & \frac{\hat{L}^{(E1)}(s_1) \cdot \hat{V}_{\cdot 3}^{(E1)}(s_1)}{\hat{\sigma}^{(E1)}(s_1)}   \\ \frac{\hat{L}^{(E1)}(s_1) \cdot \hat{V}_{\cdot 3}^{(E1)}(s_1)}{\hat{\sigma}^{(E1)}(s_1)}  & \hat{V}_{33}^{(E1)}(s_1)
\end{pmatrix}
\right).
\end{split}
\end{align*}
Here $m, z \in R$ are real numbers, 
\begin{align*}
\begin{split}
\hat{s}_2(z,m) \coloneqq  \xi(x_1,x_2)|_{x_1=z, x_2 = \frac{m}{\sqrt{n^{(E 1)}}} - \hat{\Lambda}_{0*}^{(E 1)}(s_1) } 
=\xi\left(z, \frac{m}{\sqrt{n^{(E 1)}}} - \hat{\Lambda}_{0*}^{(E 1)}(s_1) \right) 
\end{split}
\end{align*}
for the prefixed adaptation rule $\xi(x_1,x_2)$, and 
\begin{align*}
\begin{split}
\Delta \hat{F}(s)\coloneqq  \hat{F}_{02}^{(E 1)}(s) - \hat{F}_{12}^{(E 1)}(s) \qquad \text{for any } s \geq 0.
\end{split}
\end{align*}
We claim that $P_{H_0}( c_0 \leq Z_{11} < c_1, w_1 Z_{12} + w_2 Z_{2} \geq \hat{c}_2) \to \alpha - \alpha_1$ as sample size increases. To see this, let
\begin{align*}
	\begin{split}
		\psi(z,m) 		&\coloneqq  \frac{[\Delta F(s_2(z,m)) - \Delta F(s_1) ] \cdot m ]}{ {\varsigma}(s_2(z,m),s_1)}, \\
	    \hat{\psi}(z,m) &\coloneqq  \frac{[\Delta \hat{F}(\hat{s}_2(z,m)) - \Delta \hat{F}(s_1) ] \cdot m ]}{ \hat{{\varsigma}}(\hat{s}_2(z,m),s_1)}, \\
	    G(c) &\coloneqq  \int_{c_0}^{c_1} dz \int_{-\infty}^{\infty} dm f(z,m) \Phi\left( -c + w_1 \cdot \psi(z,m)  \right), \\ 
	    \hat{G}(c) &\coloneqq  \int_{c_0}^{c_1} dz \int_{-\infty}^{\infty} dm f(z,m) \Phi\left( -c + w_1 \cdot \hat{\psi}(z,m)  \right).
	\end{split}
\end{align*}
We need to show that 
\begin{equation}\label{eq:psi_convergence}
	\hat{\psi}(z,m) \underset{n\to \infty}{\overset{\mathbb{P}}{\longrightarrow}} {\psi}(z,m)
\end{equation}
for all $z,m$. First, we have $\hat{\varsigma}(\hat{s}_2(z,m),s_1) \to \varsigma(s_2(z,m),s_1)$ by Lemma \ref{lemma:cmp_extended} which can be applied because of the specification of $\hat{s}_2$. This statement also requires uniform convergence of $\hat{\varsigma}$ as defined in \eqref{eq:covariance_estimate}. This convergence holds because of the uniform convergence of the components of $\hat{L}$ and $\hat{V}$. The former is a consequence of Lemma \ref{lemma:F_convergence}. The latter can be proven by decomposing $\hat{V} - V$ into $(\hat{[\Theta]} - [\Theta]) + ([\Theta] - V)$. The first of those two summands vanishes uniformly by appplying the arguments of Lemma \ref{lemma:conv_integral_product_estimators} uniformly over all upper bounds of the integral. Uniform convergence of the second summand follows from Rebolledos Martingale Central Limit Theorem. The desired convergence in \eqref{eq:psi_convergence} now follows from plugging those convergences together and application of Slutsky's Theorem.\\
Now, as \eqref{eq:psi_convergence} holds for all $z,m$, it follows that $\hat{G}(c) \underset{n\to \infty}{\overset{\mathbb{P}}{\longrightarrow}} {G}(c)$ for all $c \in R$ by Lemma \ref{lemma:contmap_domconv}. Moreover, $\hat{G}(c)$ and ${G}(c)$ are both continuous and strictly decreasing in $c$. Thus, their inverses $\hat{G}^{-1}(\alpha)$ and ${G}^{-1}(\alpha)$ exist, and from Lemma \ref{lemma:convinprob_inverse} we get
\begin{align}\label{eq:inverse_level_function}
	\hat{G}^{-1}(\alpha) \underset{n\to \infty}{\overset{\mathbb{P}}{\longrightarrow}} {G}^{-1}(\alpha).
\end{align}
Notice that $\lim_{n \to \infty} \mathbb{P}_{H_0}( c_0 \leq Z_{11} < c_1, w_1 Z_{12} + w_2 Z_{2} \geq \hat{c}_2) = \alpha - \alpha_1$ is a consequence of Eq. \eqref{eq:inverse_level_function} and Lemma \ref{lemma:contmap_domconv}. Indeed, choosing $\hat{c}_2\coloneqq \hat{G}^{-1}(\alpha-\alpha_1)$ yields $\hat{c}_2 \to {G}^{-1}(\alpha-\alpha_1)$ by \eqref{eq:inverse_level_function} and thus $G(\hat{c}_2) \to \alpha - \alpha_1$ in probability as $n \to \infty$ by Lemma \ref{lemma:contmap_domconv}. All in all, this finished the proof that the rejection region $\{ \hat{Z}_{11} \geq c_1 \} \cup \{    c_0 \leq \hat{Z}_{11} < c_1 ,   w_1 \hat{Z}_{12} + w_2 \hat{Z}_{22} \geq \hat{c}_2 \}$ yields the desired approximate adaptive level-$\alpha$ test of $H_0$.\\

\subsection{Proof of Distributional Properties of $\Psi^{(\kappa)}$ from Chapter \ref{Sec05.02}, Eq. \eqref{eq05.04}, and $\Psi^{(\kappa)}$ from Chapter \ref{Sec06.02}, Eq. \eqref{eq06.04}}

Let a filtered probability space $(\Omega, {\cal{F}}, ({\cal{F}}_s)_{s \geq 0},\mathbb{P})$ be given together with 
\begin{itemize}
	\item state-specific left-continuous "number at risk"-processes $Y_j:\Omega \times [0,\infty) \to \{0,1, \ldots, n\}$ (depending on sample size $n$) with limits $y_j$ according to Lemma \ref{lemma:uniform_lln}, for which we assume that \ref{condition:at_risk_process_general} holds for $J_{02}$ and $J_{12}$
	\item corresponding "at risk indicator"-processes $J_j(s)\coloneqq  I(Y_j(s) > 0)$, for which we can apply Lemma \ref{lemma:non_empty_at_risk_set} according to the previous assumption,
	\item transition-specific Nelson-Aalen estimators $\hat{\Lambda}_j:\Omega \times [0,\infty) \to [0,\infty)$ for corresponding hazard functions $\Lambda_j:[0,\infty) \to [0,\infty)$ for which the convergence of Lemma \ref{lemma:na_convergence} holds,
	\item corresponding martingale formulations $M_j$ as in \eqref{eq:na_martingale} with covariation processes as in \eqref{eq:na_martingale_covariation} that converge to deterministic functions as $n \to \infty$.
\end{itemize}

In this setting, we are interested in the asymptotic distribution of stochastic processes of the form
\begin{align}\label{A8.Eq02}
\begin{split}
\Psi(s) \coloneqq   e^{-\hat{\Lambda}_3(s)} \cdot \left( c  +  \int_0^s  e^{-\hat{\Lambda}_1(u-)} d\hat{\Lambda}_2(u)    \right)
\end{split}
\end{align}
for some real number $c$ where $\hat{\Lambda}_1, \hat{\Lambda}_2$ and $\hat{\Lambda}_3$ are Nelson-Aalen estimates or sums or differences of Nelson-Aalen estimates. 

\subsubsection{Asymptotic distribution of $\Psi$}

For this purpose, we introduce the abbreviation  
\begin{align*}
\begin{split}
\hat{F}(s) \coloneqq    \int_0^s  e^{-\hat{\Lambda}_1(u-)} d\hat{\Lambda}_2(u)
\qquad \text{and} \qquad
{F}(s) \coloneqq    \int_0^s  e^{-{\Lambda}_1(u)} d{\Lambda}_2(u),
\end{split}
\end{align*}
and in a first step evaluate the asymptotic distribution of $\hat{F}(s)$ as $n \to \infty$. The process $\hat{F}(s)$ may be decomposed as follows:
\begin{align}\label{A8.Eq01}
\begin{split}
& \sqrt{n} \hat{F}(s) 
= \sqrt{n}  \int_0^s  e^{-{\Lambda}_1(u)} J_2(u) d{\Lambda}_2(u) 
+ \int_0^s  e^{-{\Lambda}_1(u)} dM_2(u) \\
&+ \sqrt{n}  \int_0^s  [e^{-\hat{\Lambda}_1(u-)} - e^{-{\Lambda}_1(u)}] J_2(u)d{\Lambda}_2(u)
+    \int_0^s  [e^{-\hat{\Lambda}_1(u-)} - e^{-{\Lambda}_1(u)}] dM_2(u).
\end{split}
\end{align}
We evaluate the four summands on the right hand side of Eq. \eqref{A8.Eq01} separately. 
\\[.5em]
\textsc{First Summand of Eq. \eqref{A8.Eq01}}: Under condition \ref{condition:at_risk_process_general}, this summand coincides with $\sqrt{n}  \int_0^s  e^{-{\Lambda}_1(u)} d{\Lambda}_2(u)$ up to terms that vanish in probability as $n \to \infty$, because
\begin{align*}
	\begin{split}
		\Big|\sqrt{n}  \int_0^s  e^{-{\Lambda}_1(u)} [1-J_2(u)] d{\Lambda}_2(u) \Big|
		\leq 
		\sqrt{n} \left( 1 - \inf_{u \in [0,s]} J_2(u) \right) \Lambda_2(s).
	\end{split}
\end{align*}
\\[.5em]
\textsc{Second Summand of Eq. \eqref{A8.Eq01}}: By application of Lemma \ref{lemma:rebolledo_condition_iii}, this summand is a mean-zero martingale with covariation $\int_0^s e^{-2\Lambda_1(u)} d[M_2](u)$.
\\[.5em]
\textsc{Third Summand of Eq. \eqref{A8.Eq01}}: Up to terms that vanish in probability as $n \to \infty$, the third summand coincides with $\sqrt{n}  \int_0^s  [e^{-\hat{\Lambda}_1(u-)} - e^{-{\Lambda}_1(u)}] d{\Lambda}_2(u)$ because of condition \ref{condition:at_risk_process_general}.
Next notice that with $dF(u)= e^{-\Lambda_1(u)}d\Lambda_2(u)$
\begin{align*}
	\begin{split}
		& \sqrt{n}  \int_0^s  [e^{-\hat{\Lambda}_1(u-)} - e^{-{\Lambda}_1(u)}] d{\Lambda}_2(u) \\
		= & - \sqrt{n} \int_0^s \left( \hat{\Lambda}_1(u-) - {\Lambda}_1(u) \right)  dF(u) + \sqrt{n} \int_0^s \sum_{k = 2}^{\infty} \frac{\left(- \hat{\Lambda}_1(u-) + {\Lambda}_1(u) \right)^k }{k!} dF(u) \\
		= & - \int_0^s M_1(u)  dF(u) - \sqrt{n} \int_0^s \left( \hat{\Lambda}_1(u-) - \hat{\Lambda}_1(u) \right)  dF(u) + R(s)
	\end{split}
\end{align*}
with $R(s) \overset{\mathbb{P}}{\to} 0$ by Lemma \ref{lemma:ls_integral_reformulation}. Up to terms that vanish in probability as $n \to \infty$ we thus have
\begin{align*}
	\begin{split}
		\sqrt{n}  \int_0^s  [e^{-\hat{\Lambda}_1(u-)} - e^{-{\Lambda}_1(u)}] J_2(u)d{\Lambda}_2(u)
		& =  - \int_0^s M_1(u)  dF(u) \\
		& = -   M_1(s) F(s) + \int_0^s F(u) d M_1(u),
	\end{split}
\end{align*}
where we used partial integration for semimartingales in the last step (Note: $F_j(u)$ is continuous and deterministic, and thus of finite variation).
\\[.5em]
\textsc{Fourth Summand of Eq. \eqref{A8.Eq01}}: We can argue as for the first summand in the decomposition of Lemma \ref{lemma:na_convergence} to see that this summand vanishes as $n \to \infty$.
\\[.5em]
Hence, up to terms that vanish in probability as $n \to \infty$, we have
\begin{align}\label{A8.Eq03a}
\begin{split}
\sqrt{n} \hat{F}(s) = \sqrt{n} F(s) + \int_0^s e^{-\Lambda_1(u)} dM_2(u)  - M_1(s) F(s) + \int_0^s F(u) d M_1(u)
\end{split}
\end{align}
In particular, we obtain the interim result that $\hat{F}(s)- F(s)$ is asymptotically normally distributed for each $s \geq 0$, i.e. 
\begin{align}\label{A8.Eq03}
\begin{split}
\sqrt{n} \left( \hat{F}(s) - F(s) \right)
\underset{n\to \infty}{\overset{\mathcal{D}}{\longrightarrow}}
{\cal{N}}\left(0,\nu(s)^2\right)
\end{split}
\end{align}
for some deterministic variance function $\nu(s)^2$. On this basis we now elaborate the asymptotic distribution of $\sqrt{n} \Psi(s) \equiv \sqrt{n} e^{- \hat{\Lambda}_3(s)} \cdot [\hat{F}(s) + c]$ from Eq. \eqref{A8.Eq02} in the limit $n \to \infty$. The process $\sqrt{n}\Psi(s)$ may be decomposed as follows:
\begin{align}\label{A8.Eq04}
\begin{split}
\sqrt{n} \Psi(s) \equiv   \sqrt{n} e^{- \hat{\Lambda}_3(s)} \cdot [\hat{F}(s) + c] = & \sqrt{n} e^{- {\Lambda}_3(s)} \cdot [{F}(s) + c] \\
& + e^{- {\Lambda}_3(s)} \cdot \sqrt{n} [\hat{F}(s) - F(s)] \\
& + [e^{- \hat{\Lambda}_3(s)} - e^{- {\Lambda}_3(s)}] \cdot \sqrt{n} [\hat{F}(s) - F(s)]  \\
& + [e^{- \hat{\Lambda}_3(s)} - e^{- {\Lambda}_3(s)}] \cdot \sqrt{n} [F(s) + c].
\end{split}
\end{align}
Again, we evaluate the four summands in \eqref{A8.Eq04} separately.
\\
\textsc{First Summand of Eq. \eqref{A8.Eq04}}: This is a deterministic drift that is not further transformed. 
\\
\textsc{Second Summand of Eq. \eqref{A8.Eq04}}: By Eq. \eqref{A8.Eq03a} the second summand equals 
\begin{align*}
\begin{split}
e^{- {\Lambda}_3(s)} \cdot \left(   \int_0^s e^{-\Lambda_1(u)} dM_2(u)  + \int_0^s F(u) d M_1(u)  - M_1(s) F(s)  \right)  + {\cal{O}}_P(1).
\end{split}
\end{align*}
\\
\textsc{Third Summand of Eq. \eqref{A8.Eq04}}: Since $|e^{-x}-e^{-x'}|\leq |x-x'|$ for all $x,x' \geq 0$, notice that $|e^{-\hat{\Lambda}_2(s)}-e^{-{\Lambda}_2(s)}|$ vanishes in probability as $n \to \infty$ as $|\hat{\Lambda}_2(s)- \Lambda_2(s)|$ does. It thus follows from our interim result \eqref{A8.Eq03} and appeal to Slutsky's theorem that the third summand of Eq. \eqref{A8.Eq04} vanishes in probability as $n \to \infty$, too. 
\\
\textsc{Fourth Summand of Eq. \eqref{A8.Eq04}}: We have
\begin{align*}
\begin{split}
& \sqrt{n}  [e^{-\hat{\Lambda}_3(s)} - e^{-{\Lambda}_3(s)}] \\
= & \sqrt{n}  \left( e^{ - [\hat{\Lambda}_3(s) - {\Lambda}_3(s) ] } - 1 \right) e^{-{\Lambda}_3(s)} \\
= & - \sqrt{n} \left( \hat{\Lambda}_3(s) - {\Lambda}_3(s) \right) e^{-{\Lambda}_3(s)} + R(s)
\end{split}
\end{align*}
with $R(s)\coloneqq  \sqrt{n} \sum_{k = 2}^{\infty} \frac{\left(- \hat{\Lambda}_3(s) + {\Lambda}_3(s) \right)^k }{k!} e^{-{\Lambda}_3(s)}$. We claim that $R(s)$ vanishes in probability as $n \to \infty$. Let $\Delta(s)\coloneqq |\hat{\Lambda}_3(u)-{\Lambda}_3(u)|$. Then
\begin{align*}
\begin{split}
 \mathbb{P}\left( |R(s)| > \epsilon \right)
&= \mathbb{P}\left( |R(s)| > \epsilon, \Delta(s) \leq 1 \right) + \mathbb{P}\left( |R(s)| > \epsilon, \Delta(s) > 1 \right) \\
& \leq \mathbb{P}\left(     \sqrt{n} \sum_{k = 2}^{\infty} \frac{\Delta(s)^2 }{k!} e^{-{\Lambda}_3(s)} > \epsilon , \Delta(s) \leq 1    \right) + \mathbb{P}\left( \Delta(s) > 1 \right) \\
& \leq \mathbb{P}\left(     n^{1/4} \Delta(s) > \sqrt{\frac{e^{{\Lambda}_3(s)} \cdot \epsilon}{e-2}}   \right) + \mathbb{P}\left( \Delta(s) > 1 \right) 
\end{split}
\end{align*}
which vanishes as $n \to \infty$, because $n^k \Delta(s)$ vanishes in probability as $n \to \infty$ for any $k<1/2$ by Lemma \ref{lemma:na_convergence}. Since $\sqrt{n}|\int_0^s J_3(u) d\Lambda_3(u) - \Lambda_3(s)|$ vanishes in probability as $n \to \infty$ under our assumption \ref{condition:at_risk_process_general}, we conclude that
\begin{align*}
\begin{split}
\sqrt{n}  [e^{-\hat{\Lambda}_3(s)} - e^{-{\Lambda}_3(s)}] = - M_3(s) \cdot e^{-{\Lambda}_3(s)}
\end{split}
\end{align*}
up to terms that vanish in probability as $n \to \infty$. Consequently, the fourth summand is
\begin{align*}
\begin{split}
\sqrt{n}  [e^{-\hat{\Lambda}_3(s)} - e^{-{\Lambda}_3(s)}] \cdot [F(s) + c] = - e^{-{\Lambda}_3(s)} \cdot M_3(s) \cdot [F(s) + c]  + o_P(1)
\end{split}
\end{align*}
All in all, we have seen that
\begin{align*}
\begin{split}
\sqrt{n} \Psi(s) = &
 \sqrt{n} \cdot e^{-{\Lambda}_3(s)} \cdot [F(s) + c] \\ 
&+
e^{- {\Lambda}_3(s)} \cdot \left(   \int_0^s e^{-\Lambda_1(u)} dM_2(u)  + \int_0^s F(u) d M_1(u)  - M_1(s) \cdot F(s) - M_3(s) \cdot [F(s) + c]  \right) \\
&+ o_P(1)
\end{split}
\end{align*}
In our use cases we will always have $M_1+M_3 \equiv M_{0*}$ (see Eqs. \eqref{eq05.08} and \eqref{eq06.08}) with $M_{0*}$ according to \eqref{eq:02.04}. Using this, $\Psi(s)$ may be simplified to
\begin{align*}
	\begin{split}
		\sqrt{n} \Psi(s) = &
		\sqrt{n} \cdot e^{-{\Lambda}_3(s)} \cdot [F(s) + c] \\ 
		&+
		e^{- {\Lambda}_3(s)} \cdot \left(   \int_0^s e^{-\Lambda_1(u)} dM_2(u)  + \int_0^s F(u) d M_1(u)  -  F(s) \cdot M_{0*}(s) - c \cdot M_3(s)  \right) \\
		&+ o_P(1) 
	\end{split}
\end{align*}
To elaborate the asymptotic distribution of $\Psi(s)$ as $n \to \infty$, we first elaborate the asymptotic distribution of the multivariate $({\cal{F}}_{s})_{s \geq 0}$-martingale
\begin{equation*}
\Theta(s) \coloneqq 
\begin{pmatrix} \int_0^s F(u) dM_1(u) \\ \int_0^s e^{-\Lambda_1(u)} dM_{2}(u) \\  M_{0*}(s)   \\   M_3(u)   \end{pmatrix}.
\end{equation*}
Knowledge of the asymptotic distribution of $\Theta(s)$ will in a second step yield the approximate distribution of $\Psi(s)$ as $n \to \infty$, because up to terms that vanish in probability as $n \to \infty$, we have
\begin{equation*}
\sqrt{n} \left( \Psi(s) - \mu(s) \right) \coloneqq  L(s) \cdot \Theta(s) + {\cal{O}}_P(1)
\end{equation*}
with scalar drift function $\mu(s)\coloneqq  e^{-\Lambda_3(u)}[F(s) + c]$ and vector valued scale function $L(s)\coloneqq  e^{-\Lambda_3(u)} \cdot (1,1,-F(s),-c)^t$. 

\subsubsection{Estimation of the variance structure of $\Psi^{(\kappa)}(s)$}
The multivariate $({\cal{F}}_{s})_{s \geq 0}$-martingale $\Theta(s)$ has covariation matrix $[\Theta](s)\coloneqq ([\Theta]_{ij}(s))_{i,j=1,\ldots,4}$ with components
\begin{align*}
\begin{split}
[\Theta]_{11}(s) 													&= \int_0^s F(u)^2 d[M_1](u) \\
[\Theta]_{12}(s) \equiv [\Theta]_{21}(s)  &= \int_0^s F(u) e^{-\Lambda_1(u)} d[M_1,M_2](u) \\
[\Theta]_{13}(s) \equiv [\Theta]_{31}(s)  &= \int_0^s F(u) d[M_1,M_{0*}](u) \\
[\Theta]_{14}(s) \equiv [\Theta]_{41}(s)  &= \int_0^s F(u) d[M_1,M_3](u) \\
[\Theta]_{22}(s) 													&= \int_0^s e^{-2 \Lambda_3(u)} d[M_2](u)  \\
[\Theta]_{23}(s) \equiv [\Theta]_{32}(s)  &= \int_0^s e^{- \Lambda_1(u)} d[M_2,M_{0*}](u) \\
[\Theta]_{24}(s) \equiv [\Theta]_{42}(s)  &= \int_0^s e^{- \Lambda_1(u)} d[M_2,M_3](u) \\
[\Theta]_{33}(s) 													&= [M_{0*}](s) \\
[\Theta]_{34}(s) \equiv [\Theta]_{43}(s)  &= [M_3,M_{0*}](s)  \\
[\Theta]_{44}(s) 													&= [M_3](s)
\end{split}
\end{align*}
Notice that $[\Theta](s)$ converges in probability as $n \to \infty$ to a deterministic covariance matrix $V(s)$ by a law of large numbers for each $s \geq 0$. Moreover, $\Theta(s)$ suffices the assumptions of Lemma \ref{lemma:rebolledo_condition_iii}. By a central limit theorem for local martingales \citep{bib5}, in the large sample limit $n \to \infty$, $\Theta(s)$ thus converges in distribution to a mean zero, $({\cal{F}}_{s})_{s \geq 0}$-adapted  Gaussian process with independent increments, and some deterministic covariance matrix function $V(s)$, say.

Since the matrix $[\Theta](s)$ depends on the unknown nuisance parameters $F(s)$ and $e^{-\Lambda_1(s)}$, it cannot be used directly as estimator for $V(s)$. Instead, we use the estimators $\hat{F}(u)$ and $e^{-\hat{\Lambda}_1(u-)}$ for $F(u)$ and $e^{-\Lambda_1(u)}$, respectively, to introduce the matrix $\hat{[\Theta]}(s)\coloneqq (\hat{[\Theta]}_{ij}(s))_{i,j=1,\ldots,4}$ with components as in \eqref{eq:covariance_estimate_two_arm_os} which we use as estimator $\hat{V}(s)\coloneqq \hat{[\Theta]}(s)$ for $V(s)$. 

In order to see that $\hat{[\Theta]}(s)$ estimates $V(s)$, it remains to show that all components of the matrix $\hat{[\Theta]}(s) - {[\Theta]}(s)$ vanish in probability as $n \to \infty$. We want to apply Lemma \ref{lemma:conv_integral_product_estimators} to show that. For the factor $\hat{S}_{PFS}$ the prerequisites of Lemma \ref{lemma:conv_integral_product_estimators} hold. This is also true for $\hat{F}_{02}$ and $\hat{F}_{12}$ according to Lemma \ref{lemma:F_convergence}. The integrators of the corresponding intervals are given by estimators of variances of Nelson-Aalen estimators. Their convergence is given in Theorem IV.1.2 of \cite{Andersen:1991}.\\
We may summarize as follows:
\begin{align*}
\begin{split}
&\bullet \Theta(s) \underset{n\to \infty}{\overset{\mathcal{D}}{\longrightarrow}} {\cal{N}}(0,V(s)) \quad \text{for each } s \geq 0. \\
&\bullet \hat{[\Theta]}(s) \underset{n\to \infty}{\overset{\mathbb{P}}{\longrightarrow}} V(s) \quad \text{for each } s \geq 0, \text{i.e. } \hat{[\Theta]}(s) \text{ consistently estimates } V(s). \\
&\bullet \Theta(s_1) \text{ and } \Theta(s_2) - \Theta(s_1) \quad \text{are approximately independent for each } 0 \leq s_1 \leq s_2. \\
\end{split}
\end{align*}
In particular, this characterizes the distributional properties of $\Psi(s) - \mu(s) = L(s) \cdot \Theta(s)$ in the large sample limit.

\subsubsection{Application in adaptive clinical trials}
In adaptive randomized trials, patients are recruited sequentially in consecutive stages, and randomly assigned to the treatment arms in each stage. With a view to application to adaptive hypothesis testing in randomized trials, let us now assume that for each combination of therapy group and stage there is one tuple $(n,(\Lambda_j)_{j \in \{1,2,3\}},(\hat{\Lambda}_j)_{j \in \{1,2,3\}},(Y_j)_{j \in \{1,2,3\}},(J_j)_{j \in \{1,2,3\}},(M_j)_{j \in \{1,2,3\}},V)$.
More specifically, we denote by $(n^{(x \kappa)},(\Lambda_j^{(x \kappa)})_{j \in \{1,2,3\}}, (\hat{\Lambda}_j^{(x \kappa)})_{j \in \{1,2,3\}}, (Y_j^{(x \kappa)})_{j \in \{1,2,3\}}, (J_j^{(x \kappa)})_{j \in \{1,2,3\}}, (M_j^{(x \kappa)})_{j \in \{1,2,3\}}, V^{(x \kappa)})$ the tuple corresponding to patients from stage $\kappa$ randomized to treatment arm $x$.
In practical terms, $n^{(x \kappa)}$ is the number of patients from stage $\kappa$ randomized to treatment arm $x$, and $\Lambda_j^{(x \kappa)}$ is their cumulative hazard of type $j$ with (Nelsen-Aalen-type) estimate $\hat{\Lambda}_j^{(x \kappa)}$. $M_j^{(x \kappa)}$ is the compensating martingale of $\hat{\Lambda}_j^{(x \kappa)}$, which is normally distributed with mean zero in the limit $n^{(x \kappa)} \to \infty$.
Across all stages of the design, patients from the same treatment arm are assumed to come from the same population. Accordingly, the cumulative hazards $\Lambda_j^{(x)}\coloneqq  \Lambda_j^{(x \kappa)}$ as well as the limiting variance function $V^{(x)}\coloneqq  V^{(x \kappa)}$ do not depend on the specific stage $\kappa$.
In the subset of patients from stage $\kappa$ randomized to treatment arm $x$ we calculate the statistic
\begin{align*}
\begin{split}
\Psi^{(x \kappa)}(s) 
\coloneqq  e^{-\hat{\Lambda}_3^{(x \kappa)}(s)} \cdot \left( c  +  \int_0^s  e^{-\hat{\Lambda}_1^{(x \kappa)}(u-)} d\hat{\Lambda}_2^{(x \kappa)}(u)    \right).
\end{split}
\end{align*}
As shown above, we have
\begin{equation}\label{A8.Eq07}
\sqrt{n^{(x \kappa)}} \left( \Psi^{(x \kappa)}(s) -  \mu^{(x)}(s)  \right) =  L^{(x)}(s) \cdot \Theta^{(x \kappa)}(s) + {\cal{O}}_P(1),
\end{equation}
with drift term
\begin{align}\label{A8.Eq07bb}
\begin{split}
\mu^{(x)}(s)\coloneqq  e^{-\Lambda_3^{(x)}(u)}[F^{(x)}(s) + c], \qquad {F}^{(x)}(s) \coloneqq  \int_0^s  e^{-{\Lambda}_1^{(x)}(u)} d{\Lambda}_2^{(x)}(u),
\end{split}
\end{align}
scaling factor 
\begin{equation*}
L^{(x)}(s)\coloneqq  e^{-\Lambda_3^{(x)}(u)} \cdot (1,1,-F^{(x)}(s),-c),
\end{equation*}
and multivariate mean-zero $({\cal{F}}_s)_{s \geq 0}$-martingale
\begin{equation*}
\Theta^{(x \kappa)}(s) \coloneqq 
\begin{pmatrix} \int_0^s F^{(x)}(u) dM_1^{(x \kappa)}(u) \\ \int_0^s e^{-\Lambda_1^{(x)}(u)} dM_{2}^{(x \kappa)}(u) \\  M_{0*}^{(x \kappa)}(s)   \\   M_3^{(x \kappa)}(u)   \end{pmatrix}
\end{equation*}
that converges in distribution to a mean-zero Gaussian process with independent increments and some matrix-valued caovaiance process $V^{(x)}(s)$, say. 
As consequence of Eq. \eqref{A8.Eq07}, for sufficiently large $n^{(x \kappa)}$, the process $\Psi^{(x \kappa)}(s)$ thus has the following approximate distributional properties:
\begin{align*}
\begin{split}
& \bullet  \Psi^{(x \kappa)}(s) \sim {\cal{N}}\left( \mu^{(x)}(s), \frac{ L^{(x)}(s) \cdot V^{(x)} \cdot L^{(x)}(s)^{\textsc{T}} }{n^{(x \kappa)}} \right) \quad \text{for all } s \geq 0. \\
& \bullet  \Psi^{(x \kappa)}(s_2) - \Psi^{(x \kappa)}(s_1) \text{ and } \Psi^{(x \kappa)}(s_1) \text{ are independent for all } 0 \leq s_1 \leq s_2. \\
\end{split}
\end{align*}
At each stage $\kappa$ of the design, the unobservable quantities $\mu^{(x)}(s)$, ${F}^{(x)}(s)$, $L^{(x)}(s)$ and $V^{(x)}(s)$ may be estimated from the data via
\begin{align}\label{eq:nuisance_parameter_estimates}
\begin{split}
\hat{F}^{(x \kappa)}(s) &\coloneqq  \int_0^s  e^{-\hat{\Lambda}_1^{(x \kappa)}(u-)} d\hat{\Lambda}_2^{(x \kappa)}(u), \\
\hat{\mu}^{(x \kappa)}(s) &\coloneqq  e^{-\hat{\Lambda}_3^{(x \kappa)}(u)}[\hat{F}^{(x \kappa)}(s) + c], \\
\hat{L}^{(x \kappa)}(s)   &\coloneqq  e^{-\hat{\Lambda}_3^{(x \kappa)}(u)} \cdot (1,1,-\hat{F}^{(x \kappa)}(s),-c), \\
\hat{V}^{(x \kappa)}(s) &\coloneqq  [\hat{\Theta}^{(x \kappa)}](s), 
\end{split}
\end{align}
where $[\hat{\Theta}^{(x \kappa)}](s)$ is the optional covariation process of the multivariate mean-zero $({\cal{F}}_s)_{s \geq 0}$-martingale 
\begin{align}\label{eq:mv_martingale_with_nuisance_parameter_estimates}
	\begin{split}
 \hat{\Theta}^{(x \kappa)}(s) &\coloneqq 
		\begin{pmatrix} \int_0^s \hat{F}^{(x \kappa)}(u) dM_1^{(x \kappa)}(u) \\ \int_0^s e^{-\hat{\Lambda}_1^{(x \kappa)}(u)} dM_{2}^{(x \kappa)}(u) \\  M_{0*}^{(x \kappa)}(s)   \\   M_3^{(x \kappa)}(u)   \end{pmatrix}.
	\end{split}
\end{align}
In this context, we consider testing the null hypothesis of equal drift in the experimental ($x=E$) and control ($x=C$) treatment group
\begin{align}\label{A8.Eq09}
\begin{split}
H_0: \mu^{(E)}(s) = \mu^{(C)}(s) \quad \text{for all } s \geq 0
\end{split}
\end{align}
as well as the (stage-wise) contiguous alternatives
\begin{align*}
\begin{split}
H_1: \mu^{(E)}(s) = \left( \mu^{(C)}(s) \right)^{\exp\left( - \frac{ \gamma }{ \sqrt{n^{(\kappa)}} }  \right)} \quad \text{for all } s \geq 0
\end{split}
\end{align*}
with $n^{(\kappa)}\coloneqq  n^{(C \kappa)} + n^{(E \kappa)}$ and some $\gamma > 0$. In the large sample limit, we then obtain for the drift difference $\Delta\mu(s)\coloneqq  \mu^{(C)}(s) -  \mu^{(E)}(s)$ 
\begin{align*}
\begin{split}
\lim_{ n^{(\kappa)} \to \infty} \sqrt{n^{(\kappa)}} \cdot \Delta\mu(s)
&= \lim_{ n^{(\kappa)} \to \infty} \mu^{(C)}(s) \cdot \sqrt{n^{(\kappa)}} \cdot \left( 1 - \mu^{(C)}(s)^{\exp\left( - \frac{ \gamma }{  \sqrt{n^{(\kappa)}}  }  \right) - 1} \right) \\
&= \gamma \cdot \mu^{(C)}(s) \cdot \log\mu^{(C)}(s).
\end{split}
\end{align*}
For testing $H_0$ we use the stage-wise statistics
\begin{align}\label{A8.Eq12}
\begin{split}
\Delta\Psi^{(\kappa)}(s)\coloneqq  \Psi^{(C  \kappa)}(s) - \Psi^{(E \kappa)}(s).
\end{split}
\end{align}
For constructing an adaptive test of $H_0$ we will in fact need knowledge of the joint distribution of $\Delta\Psi^{(\kappa)}(s)$, $\Theta^{(C \kappa)}(s)$ and $\Theta^{(E \kappa)}(s)$. We have
\begin{align*}
 \Bigg[ \begin{pmatrix} \sqrt{ n^{(\kappa)} } \cdot \Delta\Psi^{(\kappa)}(s) \\ \Theta^{(C \kappa)}(s) \\  \Theta^{(E \kappa)}(s)   \end{pmatrix} 
- \begin{pmatrix} \sqrt{ n^{(\kappa)} } \cdot \Delta\mu(s) \\ 0_{4 \times 1} \\  0_{4 \times 1}   \end{pmatrix} \Bigg]
 =
\begin{pmatrix}  {\cal{L}}^{(C \kappa)}(s) & - {\cal{L}}^{(E \kappa)}(s) \\ 1_{4} & 0_{4 \times 4} \\  0_{4 \times 4} & 1_{4}   \end{pmatrix} \cdot
\begin{pmatrix} \Theta^{(C \kappa)}(s) \\  \Theta^{(E \kappa)}(s)   \end{pmatrix}, \\
\end{align*}
where
\begin{align*}  
 & {\cal{L}}^{(C \kappa)}(s) \coloneqq  \sqrt{ 1 + r^{(\kappa)}  } \cdot L^{(C)}(s) \text{ and } {\cal{L}}^{(E \kappa)}(s) \coloneqq  \sqrt{ 1 + 1/r^{(\kappa)}  } \cdot L^{(E)}(s),  \\
 & r^{(\kappa)}\coloneqq  n^{(E \kappa)}/n^{(C \kappa)} \text{ is the randomization ratio in stage } \kappa,  \\
 & 1_m \text{ is the identity matrix of size } m\text{, and} \\
 & 0_{m \times m'} \text{ is the } m \times m' \text{ matrix with zero entries, only}.
\end{align*}
Since $\Theta^{(C \kappa)}(s)$ and $\Theta^{(E \kappa)}(s)$ are independent, the vector $\left(\sqrt{ n^{(\kappa)} } \cdot \Delta\Psi^{(\kappa)}(s), \Theta^{(C \kappa)}(s) ,  \Theta^{(E \kappa)}(s) \right)^{\textsc{T}}$ thus converges to a normal distribution with 
\begin{align}\label{A8.Eq13}
\begin{split}
&\bullet \text{mean: } (\gamma \cdot \mu^{(C)}(s) \cdot \log\mu^{(C)}(s), 0_{1 \times 4}, 0_{1 \times 4})^{\textsc{T}}, \text{ and} \\
&\bullet \text{covariance matrix: } 
  \begin{pmatrix} 
	\sum_{x=C,E} {\cal{L}}^{(x \kappa)}(s) \cdot V^{(x)}(s) \cdot {\cal{L}}^{(x \kappa)}(s)^{\textsc{T}}           &  {\cal{L}}^{(C \kappa)}(s) \cdot V^{(C)}(s) &  - {\cal{L}}^{(E \kappa)}(s) \cdot V^{(E)}(s)   \\
	V^{(C)}(s) \cdot {\cal{L}}^{(C \kappa)}(s)^{\textsc{T}}                                 					   &                                  V^{(C)}(s) &  0_{4 \times 4}   \\
	- V^{(E)}(s) \cdot {\cal{L}}^{(E \kappa)}(s)^{\textsc{T}}                                						   &                              0_{4 \times 4} &  V^{(E)}(s)   \\
  \end{pmatrix} 
	\end{split}
\end{align}
as $n^{(\kappa)} \to \infty$. In particular, for sufficiently large $n^{(\kappa)}$, we have
\begin{align*}
& \Delta\Psi^{(\kappa)}(s) \sim 
{\cal{N}}\left( \gamma \cdot \mu^{(C)}(s) \cdot \log\mu^{(C)}(s), \sigma^{(\kappa)}(s)^2   \right) \\
&  \text{with } \sigma^{(\kappa)}(s)^2 \coloneqq \sum_{x=C,E} {\cal{L}}^{(x \kappa)}(s) \cdot V^{(x)}(s) \cdot {\cal{L}}^{(x \kappa)}(s)^{\textsc{T}}.
\end{align*}
Finally, for any $0 < s_1 < s_2$, we will also consider
\begin{align*}
\Delta\tilde{\Psi}^{(\kappa)}(s_2,s_1) \coloneqq  
		{\cal{L}}^{(C \kappa)}(s_2) \cdot \left(  \Theta^{(C \kappa)}(s_2)     -   \Theta^{(C \kappa)}(s_1)    \right) - 
		{\cal{L}}^{(E \kappa)}(s_2) \cdot \left(  \Theta^{(E \kappa)}(s_2)     -   \Theta^{(E \kappa)}(s_1)    \right).
\end{align*}
Since $(\Theta^{(x \kappa)}(s))_{s \geq 0}$ is an $({\cal{F}}_s)_{s \geq 0}$-martingale with asymptotically independent increments, the random variable $\Delta\tilde{\Psi}^{(\kappa)}(s_2,s_1)$ has the following approximate distributional properties:
\begin{align*}
\begin{split}
&\bullet \quad \Delta\tilde{\Psi}^{(\kappa)}(s_2,s_1) \sim 
								{\cal{N}}\left( 0, \varsigma^{(\kappa)}(s_2,s_1)^2   \right) \\
								& \qquad \text{with } 
								\varsigma^{(\kappa)}(s_2,s_1)^2 \coloneqq \sum_{x=C,E} {\cal{L}}^{(x \kappa)}(s_2) \cdot \left( V^{(x)}(s_2) - V^{(x)}(s_1) \right) \cdot {\cal{L}}^{(x \kappa)}(s_2)^{\textsc{T}}, \\
&\bullet \quad \Delta\tilde{\Psi}^{(\kappa)}(s_2,s_1) \text{ is independent from the random vector } \left(n^{(\kappa)}\Delta\Psi^{(\kappa)}(s_1), \Theta^{(C \kappa)}(s_1) ,  \Theta^{(E \kappa)}(s_1) \right)^{\textsc{T}}.
\end{split}
\end{align*}
At each stage $\kappa$ of the design, the nuisance parameters $V^{(x)}(s)$, ${\cal{L}}^{(C \kappa)}(s)$ and ${\cal{L}}^{(E \kappa)}(s)$ may be estimated consistently by $\hat{V}^{(x \kappa)}(s)\coloneqq [\hat{\Theta}^{(x \kappa)}](s)$, $\hat{\cal{L}}^{(C \kappa)}(s) \coloneqq  \sqrt{ 1 + r^{(\kappa)}  } \cdot \hat{L}^{(C \kappa)}(s)$, and $\hat{\cal{L}}^{(E \kappa)}(s) \coloneqq  \sqrt{ 1 + 1/r^{(\kappa)}  } \cdot \hat{L}^{(E \kappa)}(s)$ from \eqref{eq:nuisance_parameter_estimates}. In particular, this yields stage-wise consistent estimators
\begin{align}\label{eq:covariance_estimate_II}
\begin{split}
\hat{\sigma}^{(\kappa)}(s)^2 &\coloneqq \sum_{x=C,E} \hat{\cal{L}}^{(x \kappa)}(s) \cdot \hat{V}^{(x \kappa)}(s) \cdot \hat{\cal{L}}^{(x \kappa)}(s)^{\textsc{T}} \\
\hat{\varsigma}^{(\kappa)}(s_2,s_1)^2 &\coloneqq \sum_{x=C,E} \hat{\cal{L}}^{(x \kappa)}(s_2) \cdot \left( \hat{V}^{(x \kappa)}(s_2) - \hat{V}^{(x \kappa)}(s_1) \right) \cdot \hat{\cal{L}}^{(x \kappa)}(s_2)^{\textsc{T}}
\end{split}
\end{align}
of $\sigma^{(\kappa)}(s)^2$ and $\varsigma^{(\kappa)}(s_2,s_1)^2$. This way, we have achieved full knowledge of the joint distribution of the vector $\left(\sqrt{ n^{(\kappa)} } \cdot \Delta\Psi^{(\kappa)}(s), \Theta^{(C \kappa)}(s) ,  \Theta^{(E \kappa)}(s) \right)^{\textsc{T}}$, as required for constructing adaptive tests of $H_0$.
\ \\ \\ \\
We next explicitly construct a two-stage adaptive test of $H_0$ in the context of a two-arm randomized clinical trial comparing an experimental treatment $E$ to control treatment $C$.
Let $\Delta\Psi^{(1)}(s_1)$ denote the statistic from Eq. \eqref{A8.Eq12} calculated in the set of stage 1 patients at some early point of time $s_1 > 0$ that is characteristic for the short-term response to treatment, e.g. $s_1 = 6$ months after the start of treatment. 
Let $\hat{Z}_{11} \coloneqq  \sqrt{n^{(1)}} \Delta\Psi^{(1)}(s_1)/\hat{\sigma}^{(1)}(s_1)$ be the interim statistics (a measure for the treatment effect regarding OS).
Likewise, let $\hat{\Lambda}_{0*}^{(x 1)}(s_1)$ denote the Nelsen-Aalen estimate of the cumulative short-term hazard $\Lambda_{0*}^{(x)}(s_1)$ for PFS in treatment group $x=E,C$, derived in stage 1 patients. Let further $S_2\coloneqq \xi(\hat{Z}_{11},\hat{\Lambda}_{0*}^{(C1)}(s_1)-\hat{\Lambda}_{0*}^{(E1)}(s_1)) > s_1$ be some long-term point of time that is chosen at the interim analysis dependent on $\hat{Z}_{11}$ and $\hat{\Lambda}_{0*}^{(C 1)}(s_1)-\hat{\Lambda}_{0*}^{(E 1)}(s_1)$ with some continuous function $\xi: \mathbb{R} \times \mathbb{R} \to (s_1,\infty)$, $(x_1,x_2) \mapsto \xi(x_1,x_2)$. Analogously, let $\Delta\Psi^{(\kappa)}(S_2)$ denote the statistic from Eq. \eqref{A8.Eq12} evaluated in the set of patients from stage $\kappa=1,2$ at the random long-term point of time $S_2 > s_1$.\\
On this basis we consider the standardized statistics
\begin{align*}
\begin{split}
\hat{Z}_{11} &\coloneqq  \frac{ \sqrt{n^{(1)}} \Delta\Psi^{(1)}(s_1)}{\hat{\sigma}^{(1)}(s_1)} \\
\hat{Z}_{12} &\coloneqq  \frac{ \sqrt{n^{(1)}} \Delta\Psi^{(1)}(S_2) - \sqrt{n^{(1)}} \Delta\Psi^{(1)}(s_1)}{\hat{{\varsigma}}^{(1)}(S_2,s_1)} \\
\hat{Z}_{2}  &\coloneqq  \frac{ \sqrt{n^{(2)}} \Delta\Psi^{(2)}(S_2)}{\hat{\sigma}^{(2)}(S_2)} 
\end{split}
\end{align*}
For some prefixed weights $0<w_1,w_2<1$ with $w_1^2+w_2^2=1$, we are interested in calculating the probability of the event
\begin{align*}
\begin{split}
\hat{\cal{R}} \coloneqq  \{ \hat{Z}_{11} \geq c_1 \} \cup \{    c_0 \leq \hat{Z}_{11} < c_1 ,   w_1 \hat{Z}_{12} + w_2 \hat{Z}_{2} \geq c_2 \}
\end{split}
\end{align*}
For this purpose, let us also introduce the auxiliary statistics
\begin{align*}
\begin{split}
Z_{11} &\coloneqq  \frac{\sqrt{n^{(1)}} \Delta\Psi^{(1)}(s_1)}{\sigma^{(1)}(s_1)} \\
Z_{12} &\coloneqq  \frac{\sqrt{n^{(1)}} \Delta\Psi^{(1)}(S_2) - \sqrt{n^{(1)}} \Delta\Psi^{(1)}(s_1)}{{\varsigma}^{(1)}(S_2,s_1)} \\
Z_{2}  &\coloneqq  \frac{\sqrt{n^{(2)}} \Delta\Psi^{(2)}(S_2)}{\sigma^{(2)}(S_2)} 
\end{split}
\end{align*}
as well as the event
\begin{align*}
\begin{split}
{\cal{R}} \coloneqq  \{ Z_{11} \geq c_1 \} \cup \{    c_0 \leq Z_{11} < c_1 ,   w_1 Z_{12} + w_2 Z_{2} \geq c_2 \}.
\end{split}
\end{align*}
Notice that $(\hat{Z}_{11}, \hat{Z}_{12}, \hat{Z}_{2})$ has the same distribution as $({Z}_{11}, {Z}_{12}, {Z}_{2})$ in the large sample limit. Thus $\mathbb{P}(\hat{\cal{R}}) \approx \mathbb{P}({\cal{R}})$ when sample size is large, and we are reduced to calculate $\mathbb{P}({\cal{R}})$. Clearly, we have
\begin{align}\label{A8.Eq19}
\begin{split}
\mathbb{P} \left( \{ Z_{11} \geq c_1 \} \right) = \mathbb{P} \left( \frac{\Delta\Psi(s_1) - \Delta\mu(s_1)}{\sigma(s_1)} \geq c_1 - \Delta\frac{\mu(s_1)}{\sigma(s_1)} \right) = 1 - \Phi\left( c_1 - \frac{\Delta\mu(s_1)}{\sigma(s_1)} \right).
\end{split}
\end{align}
To evaluate $\mathbb{P}( c_0 \geq Z_{11} < c_1 ,   w_1 Z_{12} + w_2 Z_{2} \geq c_2 )$ first notice that we have
\begin{align*}
\begin{split}
&\sqrt{n^{(\kappa)}} \cdot \left( \Delta\Psi^{(\kappa)}(S_2) - \Delta\Psi^{(\kappa)}(s_1) - [\Delta\mu(S_2) - \Delta\mu(s_1) ] \right) \\
= & {\cal{L}}^{(C \kappa)}(S_2) \cdot \Theta^{(C \kappa)}(S_2) - {\cal{L}}^{(C \kappa)}(s_1) \cdot \Theta^{(C \kappa)}(s_1) 
 - \left( {\cal{L}}^{(E \kappa)}(S_2) \cdot \Theta^{(E \kappa)}(S_2) - {\cal{L}}^{(E \kappa)}(s_1) \cdot \Theta^{(E \kappa)}(s_1)  \right) \\
= & {\cal{L}}^{(C \kappa)}(S_2) \cdot \left( \Theta^{(C \kappa)}(S_2) - \Theta^{(C \kappa)}(s_1) \right)
  - {\cal{L}}^{(E \kappa)}(S_2) \cdot \left( \Theta^{(E \kappa)}(S_2) - \Theta^{(E \kappa)}(s_1) \right)  \\
& +	\Theta^{(C \kappa)}(s_1) \cdot \left( {\cal{L}}^{(C \kappa)}(S_2) - {\cal{L}}^{(C \kappa)}(s_1) \right)
	-	\Theta^{(E \kappa)}(s_1) \cdot \left( {\cal{L}}^{(E \kappa)}(S_2) - {\cal{L}}^{(E \kappa)}(s_1) \right) \\
= & \Delta\tilde{\Psi}^{(\kappa)}(S_2,s_1) 
  +	\sum_{x=C,E} (-1)^{I(x=E)} \cdot \Theta^{(x \kappa)}(s_1) \cdot \left( {\cal{L}}^{(x \kappa)}(S_2) - {\cal{L}}^{(x \kappa)}(s_1) \right),
\end{split}
\end{align*}
with the indicator $I(x=E)$ being equal to $1$ if $x=E$, and $0$ otherwise. Finally notice that, on the set where $J_{02}^{(x \kappa)}=1$, we have
\begin{align*}
\begin{split}
\hat{\Lambda}^{(x \kappa)}(s) 
&= \frac{M_{0 *}^{(x \kappa)}(s)}{\sqrt{n^{(x \kappa)}}}  +  {\Lambda}^{(x)}(s)
= \frac{ pr_3( \hat{\Theta}^{(x \kappa)}(s) )}{\sqrt{n^{(x \kappa)}}} +  {\Lambda}^{(x)}(s), \text{ and thus} \\
S_2 &=  \xi\left( \hat{Z}_{11},  \frac{pr_3( \hat{\Theta}^{(C 1)}(s) )}{\sqrt{n^{(C 1)}}} -  \frac{ pr_3( \hat{\Theta}^{(E 1)}(s) )}{\sqrt{n^{(E 1)}}}  +  {\Lambda}^{(C)}(s) - {\Lambda}^{(E)}(s)    \right)
\end{split}
\end{align*}
with $pr_3:(x_1,x_2,x_3,x_4) \mapsto x_3$ denoting projection onto the third component. Also, let
\begin{equation}\label{eq:final_analysis_date_function}
\begin{split}
s_2(z,\theta_E,\theta_C) &\coloneqq \xi(x_1,x_2)|_{x_1=z,x_2=\frac{ \theta_C }{\sqrt{n^{(C 1)}}} - \frac{  \theta_E }{\sqrt{n^{(E 1)}}}  +  {\Lambda}^{(C)}(s)  -  {\Lambda}^{(E)}(s)} \\
&= \xi\left( z,  \frac{ \theta_C }{\sqrt{n^{(C 1)}}} - \frac{  \theta_E }{\sqrt{n^{(E 1)}}}  +  {\Lambda}^{(C)}(s)  -  {\Lambda}^{(E)}(s) \right).
\end{split}
\end{equation}
Then, by independence of $\Delta \tilde{\Psi}^{(\kappa)}(s_2(z,\theta_E,\theta_C),s_1)$ from $(\sqrt{n^{(\kappa)}} \Delta\Psi(s_1)^{(\kappa)}, \hat{\Theta}^{(E \kappa)}(s_1), \hat{\Theta}^{(C \kappa)}(s_1))$, and since $\lim_{n^{(\kappa)} \to \infty} \mathbb{P}(J_{02}^{(x \kappa)}(s_1)=1) = 1$ we approximately have for sufficiently large $n^{(\kappa)}$
\begin{align*}
\begin{split}
&\mathbb{P}\left( c_0 \leq Z_{11} < c_1, w_1 Z_{12} + w_2 Z_{2} \geq c_2 \right) \\
= &	\int_{c_0}^{c_1} dz \int_{\mathbb{R}^4} d\theta_E \int_{\mathbb{R}^4} d\theta_C f_{Z_{11},\hat{\Theta}^{(E \kappa)}(s_1), \hat{\Theta}^{(C \kappa)}(s_1)}(z,\theta_E,\theta_C) \\
& \qquad \qquad		\mathbb{P}\Bigg( w_1 \sqrt{n^{(1)}} \frac{\Delta\Psi^{(1)}(S_2) - \Delta\Psi^{(1)}(s_1)}{ {\varsigma}^{(1)}(S_2,s_1)}  + w_2 \sqrt{n^{(2)}} \frac{\Delta\Psi^{(2)}(S_2)}{\sigma^{(2)}(S_2)} \geq c_2 \Big| Z_{11} = z , \\ 
& \qquad \qquad \qquad   \hat{\Theta}^{(E \kappa)}(s_1) = \theta_E, \hat{\Theta}^{(C \kappa)}(s_1) = \theta_C \Bigg) \\
= & 	\int_{c_0}^{c_1} dz \int_{\mathbb{R}^4} d\theta_E \int_{\mathbb{R}^4} d\theta_C f_{Z_{11},\hat{\Theta}^{(E \kappa)}(s_1), \hat{\Theta}^{(C \kappa)}(s_1)}(z,\theta_E,\theta_C) \\
& \qquad \qquad		
	\mathbb{P}\Bigg( w_1 \frac{\Delta\tilde{\Psi}^{(1)}(s_2(z, \theta_E, \theta_C),s_1) }{ {\varsigma}^{(1)}(s_2(z, \theta_E, \theta_C),s_1)} 
			  + w_2 \sqrt{n^{(2)}} \frac{\Delta\Psi^{(2)}(s_2(z, \theta_E, \theta_C)) - \Delta\mu(s_2(z, \theta_E, \theta_C))}{\sigma^{(2)}(s_2(z, \theta_E, \theta_C))} \\
				& \qquad  \qquad  + w_2 \sqrt{n^{(2)}} \frac{\Delta\mu(s_2(z, \theta_E, \theta_C))}{\sigma^{(2)}(s_2(z, \theta_E, \theta_C))} 
								  + w_1 \sqrt{n^{(1)}} \frac{ \Delta\mu(s_2(z, \theta_E, \theta_C))  -\Delta\mu(s_1)}{\varsigma^{(1)}(s_2(z, \theta_E, \theta_C),s_1)}  \\
				& \qquad  \qquad  + w_1 \frac{   \sum_{x=C,E} (-1)^{I(x=E)} \cdot \theta_x \cdot \left( {\cal{L}}^{(x \kappa)}(s_2(z, \theta_E, \theta_C)) - {\cal{L}}^{(x \kappa)}(s_1) \right)    }{  {\varsigma}^{(1)}(s_2(z, \theta_E, \theta_C),s_1) }   
				\geq c_2 \Bigg) \\
= &	\int_{c_0}^{c_1} dz \int_{\mathbb{R}^4} d\theta_E \int_{\mathbb{R}^4} d\theta_C f_{Z_{11},\hat{\Theta}^{(E \kappa)}(s_1), \hat{\Theta}^{(C \kappa)}(s_1)}(z,\theta_E,\theta_C) \\
& \qquad \qquad		\Phi\Bigg( -c_2 + w_2 \sqrt{n^{(2)}} \frac{\Delta\mu(s_2(z, \theta_E, \theta_C))}{\sigma^{(2)}(s_2(z, \theta_E, \theta_C))} 
+ w_1 \sqrt{n^{(1)}} \frac{ \Delta\mu(s_2(z, \theta_E, \theta_C))  -\Delta\mu(s_1)}{\varsigma^{(1)}(s_2(z, \theta_E, \theta_C),s_1)}  \\ 
				& \qquad  \qquad  + w_1 \frac{   \sum_{x=C,E} (-1)^{I(x=E)} \cdot \theta_x \cdot \left( {\cal{L}}^{(x \kappa)}(s_2(z, \theta_E, \theta_C)) - {\cal{L}}^{(x \kappa)}(s_1) \right)    }{  {\varsigma}^{(1)}(s_2(z, \theta_E, \theta_C),s_1) }  \Bigg)
\end{split}
\end{align*}
where $f_{Z_{11},\hat{\Theta}^{(E \kappa)}(s_1), \hat{\Theta}^{(C \kappa)}(s_1)}(z,\theta_E,\theta_C)$ is the joint density of the vector $(Z_{11},\hat{\Theta}^{(E \kappa)}(s_1), \hat{\Theta}^{(C \kappa)}(s_1))$ that approximately follows the multivariate normal distribution given in \eqref{A8.Eq13}.\\
To derive an approximate level-$\alpha$ adaptive test for $H_0: \mu^{(E)}(s) = \mu^{(C)}(s)$ from \eqref{A8.Eq09} based on the rejection region $\hat{\cal{R}}$, we choose the critical bounds $c_1,c_2$ such that the stage-1 type one error rate $\mathbb{P}_{H_0}(\{ Z_{11} \geq c_1 \})$ equals $\alpha_1$ and such that the overall type one error rate $P_{H_0}(\hat{\cal{R}}) \approx P_{H_0}({\cal{R}})$ equals the desired significance level $\alpha>\alpha_1$. Since $\Delta\mu(\cdot)=0$ under $H_0$, we immediately obtain from \eqref{A8.Eq19} that $c_1=\Phi^{-1}(1-\alpha_1)$. We then have to choose $c_2$ such that $\mathbb{P}_{H_0}( c_0 \leq Z_{11} < c_1, w_1 Z_{12} + w_2 Z_{2} \geq c_2)$ equals $\alpha - \alpha_1$. This is tantamout to finding a constant $c_2$ such that
\begin{align*}
\begin{split}
\alpha - \alpha_1  
= &	\int_{c_0}^{c_1} dz \int_{\mathbb{R}^4} d\theta_E \int_{\mathbb{R}^4} d\theta_C f(z,\theta_E,\theta_C) \\
& \Phi\left( -c_2 + w_1 \frac{   \sum_{x=C,E} (-1)^{I(x=E)} \cdot \theta_x \cdot \left( {\cal{L}}^{(x \kappa)}(s_2(z, \theta_E, \theta_C)) - {\cal{L}}^{(x \kappa)}(s_1) \right)    }{  {\varsigma}^{(1)}(s_2(z, \theta_E, \theta_C),s_1) }  \right),
\end{split}
\end{align*}
where $f(z,\theta_E,\theta_C)$ is the density of the 9D random vector 
\begin{align*}
\begin{pmatrix} Z \\ {\Theta}_E \\ {\Theta}_C
\end{pmatrix}
\sim
{\cal{N}}
\left( 0_{9 \times 1} , \Sigma^2 \right), \text{ with }
\Sigma^2 \coloneqq  
  \begin{pmatrix} 
	1                                                                                                        &  \frac{  {\cal{L}}^{(E \kappa)}(s_1) \cdot {V}^{(E)}(s_1)  }{{\sigma}^{(1)}(s_1)}  &  \frac{   {\cal{L}}^{(C \kappa)}(s_1) \cdot {V}^{(C)}(s_1)   }{{\sigma}^{(1)}(s_1)}     \\
	\frac{  {V}^{(E)}(s_1) \cdot {\cal{L}}^{(E \kappa)}(s_1)^{\textsc{T}}   }{{\sigma}^{(1)}(s_1)}        	 &                                  {V}^{(E)}(s_1)       																			&  0_{4 \times 4}   \\
	\frac{  {V}^{(C)}(s_1) \cdot {\cal{L}}^{(C \kappa)}(s_1)^{\textsc{T}}  }{{\sigma}^{(1)}(s_1)}        	 &                                  0_{4 \times 4}           																			&  {V}^{(C)}(s_1)   \\
  \end{pmatrix}. 
\end{align*}
Note that this cannot be directly put into practice as it requires knowledge of various nuisance parameters. Instead we choose a random variable $\hat{c}_2$ such that
\begin{align*}
\alpha - \alpha_1  
= & \int_{c_0}^{c_1} dz \int_{\mathbb{R}^4} d\theta_E \int_{\mathbb{R}^4} d\theta_C f(z,\theta_E,\theta_C) \\
& \Phi\left( -\hat{c}_2 + w_1 \frac{   \sum_{x=C,E} (-1)^{I(x=E)} \cdot \theta_x \cdot \left( \hat{\cal{L}}^{(x \kappa)}(\hat{s}_2(z, \theta_E, \theta_C)) - \hat{\cal{L}}^{(x \kappa)}(s_1) \right)    }{  \hat{\varsigma}^{(1)}(\hat{s}_2(z, \theta_E, \theta_C),s_1) }  \right),
\end{align*}
where $f(z,\theta_E,\theta_C)$ is the density of a 9D random vector
\begin{align*}
\begin{pmatrix} Z \\ {\Theta}_E \\ {\Theta}_C
\end{pmatrix}
\sim
{\cal{N}}
\left( 0_{9 \times 1} , \Sigma^2 \right), \text{ with }
\Sigma^2 \coloneqq  
  \begin{pmatrix} 
	1                                                                                                        &  \frac{  \hat{\cal{L}}^{(E \kappa)}(s_1) \cdot \hat{V}^{(E)}(s_1)  }{\hat{\sigma}^{(1)}(s_1)}  &  \frac{   \hat{\cal{L}}^{(C \kappa)}(s_1) \cdot \hat{V}^{(C)}(s_1)   }{\hat{\sigma}^{(1)}(s_1)}     \\
	\frac{  \hat{V}^{(E)}(s_1) \cdot \hat{\cal{L}}^{(E \kappa)}(s_1)^{\textsc{T}}   }{\hat{\sigma}^{(1)}(s_1)}        	 &                                  \hat{V}^{(E)}(s_1)       																			&  0_{4 \times 4}   \\
	\frac{  \hat{V}^{(C)}(s_1) \cdot \hat{\cal{L}}^{(C \kappa)}(s_1)^{\textsc{T}}   }{\hat{\sigma}^{(1)}(s_1)}        	 &                                  0_{4 \times 4}           																			&  \hat{V}^{(C)}(s_1)   \\
  \end{pmatrix} 
\end{align*}
with $\hat{F}^{(x \kappa)}, \hat{\mu}^{(x \kappa)}, \hat{L}^{(x \kappa)}$ and $\hat{V}^{(x)}$ as defined in \eqref{eq:nuisance_parameter_estimates}, $\hat{\Theta}^{(x \kappa)}$ as defined in \eqref{eq:mv_martingale_with_nuisance_parameter_estimates} and $s_2$ as defined in \eqref{eq:final_analysis_date_function} and the derived quantities
\begin{align*}
\hat{\cal{L}}^{(C \kappa)}(s) &\coloneqq  \sqrt{ 1 + r^{(\kappa)}  } \cdot \hat{L}^{(C \kappa)}(s), \\
\hat{\cal{L}}^{(E \kappa)}(s) &\coloneqq  \sqrt{ 1 + 1/r^{(\kappa)}  } \cdot \hat{L}^{(E \kappa)}(s), \\
\hat{\sigma}^{(\kappa)}(s)^2 &\coloneqq \sum_{x=C,E} \hat{\cal{L}}^{(x \kappa)}(s) \cdot \hat{V}^{(x \kappa)}(s) \cdot \hat{\cal{L}}^{(x \kappa)}(s)^{\textsc{T}}, \\
\hat{\varsigma}^{(\kappa)}(s_2,s_1)^2 &\coloneqq \sum_{x=C,E} \hat{\cal{L}}^{(x \kappa)}(s_2) \cdot \left( \hat{V}^{(x \kappa)}(s_2) - \hat{V}^{(x \kappa)}(s_1) \right) \cdot \hat{\cal{L}}^{(x \kappa)}(s_2)^{\textsc{T}}\\
\hat{s}_2(z,\theta_E,\theta_C) &\coloneqq \xi(x_1,x_2)|_{x_1=z,x_2=\frac{ \theta_C }{\sqrt{n^{(C 1)}}} - \frac{  \theta_E }{\sqrt{n^{(E 1)}}}  +  \hat{\Lambda}^{(C)}(s)  -  \hat{\Lambda}^{(E)}(s)}
\end{align*}
We claim that $\mathbb{P}_{H_0}( c_0 \leq Z_{11} < c_1, w_1 Z_{12} + w_2 Z_{2} \geq \hat{c}_2) \to \alpha - \alpha_1$ as sample size increases. To see this, let
\begin{align*}
	\begin{split}
		\psi(z,\theta_E,\theta_C) 		&\coloneqq  \frac{   \sum_{x=C,E} (-1)^{I(x=E)} \cdot \theta_x \cdot \left( {\cal{L}}^{(x \kappa)}(s_2(z, \theta_E, \theta_C)) - {\cal{L}}^{(x \kappa)}(s_1) \right)    }{  {\varsigma}^{(1)}(s_2(z, \theta_E, \theta_C),s_1) }, \\
		\hat{\psi}(z,\theta_E,\theta_C) &\coloneqq  \frac{   \sum_{x=C,E} (-1)^{I(x=E)} \cdot \theta_x \cdot \left( \hat{\cal{L}}^{(x \kappa)}(\hat{s}_2(z, \theta_E, \theta_C)) - \hat{\cal{L}}^{(x \kappa)}(s_1) \right)    }{  \hat{\varsigma}^{(1)}(\hat{s}_2(z, \theta_E, \theta_C),s_1) }, \\
		G(c) &\coloneqq  \int_{c_0}^{c_1} dz \int_{\mathbb{R}^4} d\theta_E \int_{\mathbb{R}^4} d\theta_C f(z,\theta_E,\theta_C) \Phi\left( -c + w_1 \cdot \psi(z,\theta_E,\theta_C)  \right), \\ 
		\hat{G}(c) &\coloneqq  \int_{c_0}^{c_1} dz \int_{\mathbb{R}^4} d\theta_E \int_{\mathbb{R}^4} d\theta_C f(z,\theta_E,\theta_C) \Phi\left( -c + w_1 \cdot \hat{\psi}(z,\theta_E,\theta_C)  \right).
	\end{split}
\end{align*}
We need to show that 
\begin{equation}\label{eq:psi_convergence_II}
	\hat{\psi}(z,\theta_E,\theta_C) \underset{n\to \infty}{\overset{\mathbb{P}}{\longrightarrow}} {\psi}(z,\theta_E,\theta_C)
\end{equation}
for all $z,\theta_E,\theta_C$. First, we have $\hat{\varsigma}(\hat{s}_2(z,\theta_E,\theta_C),s_1) \to \varsigma(s_2(z,\theta_E,\theta_C),s_1)$ by Lemma \ref{lemma:cmp_extended} which can be applied because of the specification of $\hat{s}_2$. This statement also requires uniform convergence of $\hat{\varsigma}$ as defined in \eqref{eq:covariance_estimate_II}. This convergence holds because of the uniform convergence of the components of $\hat{\cal{L}}$ and $\hat{V}$. The former is a consequence of Lemma \ref{lemma:F_convergence}. The latter can be proven by decomposing $\hat{V} - V$ into $(\hat{[\Theta]} - [\Theta]) + ([\Theta] - V)$. The first of those two summands vanishes uniformly by appplying the arguments of Lemma \ref{lemma:conv_integral_product_estimators} uniformly over all upper bounds of the integral. Uniform convergence of the second summand follows from Rebolledos Martingale Central Limit Theorem. The desired convergence in \eqref{eq:psi_convergence_II} now follows from plugging those convergences together and application of Slutsky's Theorem.\\
Now, as \eqref{eq:psi_convergence_II} holds for all $z, \theta_E, \theta_C$, it follows that $\hat{G}(c) \underset{n\to \infty}{\overset{\mathbb{P}}{\longrightarrow}} {G}(c)$ for all $c \in R$ by Lemma \ref{lemma:contmap_domconv}. Moreover, $\hat{G}(c)$ and ${G}(c)$ are both continuous and strictly decreasing in $c$. Thus, their inverses $\hat{G}^{-1}(\alpha)$ and ${G}^{-1}(\alpha)$ exist, and from Lemma \ref{lemma:convinprob_inverse} we get
\begin{align}\label{eq:inverse_level_function_II}
	\hat{G}^{-1}(\alpha) \underset{n\to \infty}{\overset{\mathbb{P}}{\longrightarrow}} {G}^{-1}(\alpha).
\end{align}
Notice that $\lim_{n \to \infty} \mathbb{P}_{H_0}( c_0 \leq Z_{11} < c_1, w_1 Z_{12} + w_2 Z_{2} \geq \hat{c}_2) = \alpha - \alpha_1$ is a consequence of Eq. \eqref{eq:inverse_level_function_II} and Lemma \ref{lemma:contmap_domconv}. Indeed, choosing $\hat{c}_2\coloneqq \hat{G}^{-1}(\alpha-\alpha_1)$ yields $\hat{c}_2 \to {G}^{-1}(\alpha-\alpha_1)$ by \eqref{eq:inverse_level_function_II} and thus $G(\hat{c}_2) \to \alpha - \alpha_1$ in probability as $n \to \infty$ by Lemma \ref{lemma:contmap_domconv}. All in all, this finished the proof that the rejection region $\{ \hat{Z}_{11} \geq c_1 \} \cup \{    c_0 \leq \hat{Z}_{11} < c_1 ,   w_1 \hat{Z}_{12} + w_2 \hat{Z}_{22} \geq \hat{c}_2 \}$ yields the desired approximate adaptive level-$\alpha$ test of $H_0$.\\
\newpage



\end{document}